\documentclass
{elsarticle}

\usepackage[mathscr]{euscript}

\usepackage{lineno,hyperref}
\modulolinenumbers[5]
\usepackage{listings}
\usepackage{xcolor}
\usepackage{amsthm}
\usepackage{amsfonts}
\usepackage{prooftree}
\usepackage{float}
\floatstyle{boxed}
\restylefloat{figure}   
\usepackage{mathpartir,amsmath,amssymb
}

\journal{Journal of \LaTeX\ Templates}

\theoremstyle{plain}
\newtheorem{theorem}{Theorem}
\newtheorem{lemma}[theorem]{{\bf Lemma}}
\newtheorem{definition}[theorem]{{\bf Definition}}
\newtheorem{proposition}[theorem]{{\bf Proposition}}
\newtheorem*{lemma*}{Lemma}
\newtheorem*{proposition*}{Proposition}
\newtheorem*{theorem*}{Theorem}
\theoremstyle{definition}
\newtheorem{example}[theorem]{{\bf Example}}

\newcommand{\refToFigure}[1]{Fig.\ref{fig:#1}}
\newcommand{\refToSection}[1]{Sect.\ref{sect:#1}}
\newcommand{\refToLemma}[1]{Lemma~\ref{lemma:#1}}
\newcommand{\refToDef}[1]{Definition~\ref{def:#1}}
\newcommand{\refToProp}[1]{Proposition~\ref{prop:#1}}
\newcommand{\refToTheorem}[1]{Theorem~\ref{theo:#1}}
\newcommand{\refToExample}[1]{Example~\ref{ex:#1}}
\newcommand{\Space}{\hskip 0.4em}
\newcommand{\BigSpace}{\hskip 1.5em}
\newcommand{\Q}{\lstinline}
\newcommand{\ul}[1]{\fbox{\texttt{#1}}}
\definecolor{light-gray}{gray}{0.80}
\newcommand{\store}[1]{\colorbox{light-gray}{\texttt{#1}}}

\newcommand{\Tuple}[1]    {\langle{#1}\rangle}
\newcommand{\Pair}[2]     {\Tuple{{#1},{#2}}}
\newcommand{\FourTuple}[4]     {\Tuple{{#1},{#2},{#3},{#4}}}
\newcommand{\SubstFun}[2]{#1[#2]}
\newcommand{\Reduct}[2]{{#1}_{|#2}}
\newcommand{\Finer}[2]{{#1}\sqsubseteq{#2}}
\newcommand{\variant}{\EZ{\approx^-}}

\newcommand{\NamedRule}[4]{\scriptstyle{\textsc{(#1)}}
\displaystyle                  
\frac{%
\begin{array}{l}#2\end{array}%
}{#3}        
{\begin{array}{l}#4\end{array}}    
}
\newcommand{\rn}[1]{{\scriptsize (\textsc{#1})}}					

\newcommand{\NamedRuleOL}[3]{\scriptstyle{\textsc{(#1)}}
\displaystyle                  
{%
\begin{array}{l}#2\end{array}}\           
\begin{array}{l}#3\end{array}     
}

\newenvironment{grammatica}{$\begin{array}{lcll}}{\end{array}$}
\newcommand{\produzione}[3]{#1&::=&#2&\mbox{#3}}
\newcommand{\produzioneinline}[2]{#1::=#2}
\newcommand{\seguitoproduzione}[2]{&&#1&\mbox{#2}}
\newcommand{\terminale}[1]{\texttt{#1}}
\newcommand{\metavariable}[1]{{\it #1}}

\newcommand{\this}{\terminale{this}}
\newcommand{\resV}{\terminale{res}}
\newcommand{\x}{\metavariable{x}}
\newcommand{\y}{\metavariable{y}}
\newcommand{\z}{\metavariable{z}}
\newcommand{\xs}{\metavariable{xs}}
\newcommand{\ys}{\metavariable{ys}}

\newcommand{\m}{\metavariable{m}}
\newcommand{\e}{\metavariable{e}}
\newcommand{\es}{\metavariable{es}}
\newcommand{\C}{\metavariable{C}}
\newcommand{\D}{\metavariable{D}}
\newcommand{\f}{\metavariable{f}}
\newcommand{\dec}{\metavariable{d}}
\newcommand{\decs}{\metavariable{ds}}

\newcommand{\Field}[2]{#1\, #2}
\newcommand{\FieldAccess}[2]{#1\terminale{.}#2}
\newcommand{\MethCall}[3]{#1{\terminale{.}}#2\terminale{(}#3\terminale{)}}
\newcommand{\FieldAssign}[3]{#1\terminale{.}#2\terminale{=}#3}
\newcommand{\ConstrCall}[2]{\terminale{new}\ #1\terminale{(}#2\terminale{)}}
\newcommand{\Block}[2]{\terminale{\{}#1\;#2\terminale{\}}}
\newcommand{\Dec}[3]{#1\,#2\terminale{=}#3;}
\newcommand{\DecP}[3]{#1#2\terminale{=}#3;}
\newcommand{\Param}[2]{#1\ #2}
\newcommand{\Sequence}[2]{#1\terminale{;}#2}

\newcommand{\capsule}{\terminale{a}}
\newcommand{\Type}[2]{#2^{#1}}
\newcommand{\TypeDec}[2]{#1{:}#2}

\newcommand{\ReturnTypeNew}[2]{#1\,{\mid}\,#2}
\newcommand{\Method}[4]{\FourTuple{#1}{#2}{#3}{#4}}
\newcommand{\T}{\metavariable{T}}
\newcommand{\intType}{\terminale{int}}
\newcommand{\sharingRel}{{\cal S}}
\newcommand{\X}{\metavariable{X}} 
\newcommand{\Y}{\metavariable{Y}} 
\newcommand{\Z}{\metavariable{Z}} 

\newcommand{\Sum}[2]{#1+#2}

\newcommand{\SubstEqRel}[3]{#1[#2/#3]}

\newcommand{\Remove}[2]{#1{\setminus}#2}

\newcommand{\Closure}[2]{[#1]_{#2}}
\newcommand{\BlockLab}[3]{ \terminale{\{}^{#3} #1\;#2\terminale{\}}}

\newcommand{\TypeCheckDecs}[3]{#1\vdash #2:#3}
\newcommand{\TypeCheckAnnotate}[5]{#1\vdash #2:#3{\mid}#4{\leadsto}#5}
\newcommand{\TypeCheck}[4]{#1\vdash #2:#3\,{\mid}\,#4}
\newcommand{\IsCapsule}[1]{\aux{capsule}(#1)}

\newcommand{\HoleCtx}[1]{{\cal B}_{#1}}
\newcommand{\dv}{\metavariable{dv}}
\newcommand{\dvs}{\metavariable{dvs}}
\newcommand{\ctx}{{\cal{E}}}
\newcommand{\valCtx}{{\cal{V}}}
\newcommand{\ctxP}{{\cal E}'}
\newcommand{\ctxS}{{\cal E}''}
\newcommand{\emptyctx}{[\ ]}

\newcommand{\val}{\metavariable{v}}
\newcommand{\valPrime}{\metavariable{u}}
\newcommand{\vs}{\metavariable{vs}}
\newcommand{\Ctx}[1]{\ctx[#1]}

\newcommand{\ValCtx}[1]{\valCtx[#1]}
\newcommand{\CtxP}[1]{{\ctxP}[#1]}
\newcommand{\CtxS}[1]{{\ctxS}[#1]}
\newcommand{\decCtx}[2]{#1_{#2}}
\newcommand{\DecCtx}[3]{\decCtx{#1}{#2}[#3]}

\newcommand{\reduce}[2]{#1\longrightarrow#2}

\newcommand{\Subst}[3]{#1[#2/#3]}
\newcommand{\SubstVal}[3]{#1[#2/#3]}

\newcommand{\UpdateCtx}[4]{\updateCtx{#1}{#2}{#3}[#4]}
\newcommand{\UpdateCtxX}[4]{\updateCtxX{#1}{#2}{#3}[#4]}
\newcommand{\updateCtx}[3]{\ctx^{#2.#3{=}#1}}
\newcommand{\updateCtxX}[3]{\ctx_x^{#2.#3{=}#1}}
\newcommand{\redex}{\rho}

\newcommand{\congruence}[2]{{#1}\cong{#2}}

\newcommand{\aux}[1]{\textsf{#1}}
\newcommand{\fields}[1]{\aux{fields}(#1)}
\newcommand{\method}[2]{{\aux{meth}(#1,#2)}}
\newcommand{\FV}[1]{\aux{FV}(#1)}
\newcommand{\HB}[1]{\aux{HB}(#1)}
\newcommand{\dom}[1]{\aux{dom}(#1)}
\newcommand{\class}[2]{\aux{class}(#1,#2)}
\newcommand{\extractDec}[2]{\aux{dec}(#1,#2)}
\newcommand{\extractAllDec}[1]{\EZ{\aux{store}}(#1)}

\newcommand{\deriv}{{\cal D}}

\newcommand{\Ctt}{{\tt C}}
\newcommand{\Dtt}{{\tt D}}
\newcommand{\ftt}{{\tt f}}

\newcommand{\xtt}{{\tt x}}
\newcommand{\ytt}{{\tt y}}
\newcommand{\ztt}{{\tt z}}
\newcommand{\zU}{{\tt z1}}
\newcommand{\zD}{{\tt z2}}
\newcommand{\mix}{{\tt mix}}
\newcommand{\clone}{{\tt clone}}
\newcommand{\cU}{{\tt c1}}
\newcommand{\cD}{{\tt c2}}
\newcommand{\cT}{{\tt c3}}

\newcommand{\outC}{{\tt outer}}
\newcommand{\inC}{{\tt inner}}
\newcommand{\restt}{{\tt r}}
\newcommand{\eU}{{\ett^\texttt{i}}}
\newcommand{\eUA}{{\ett^\texttt{ia}}}
\newcommand{\eDA}{{\ett^\texttt{oa}}}
\newcommand{\eD}{{\ett^\texttt{o}}}
\newcommand{\ett}{{\tt e}}

\newcommand{\TypeCheckGround}[3]{\vdash #1:#2}
\newcommand{\IsWellTyped}[1]{\vdash #1}
\newcommand{\decctx}[1]{{\cal D}_{#1}}
\newcommand{\decctxS}[1]{{\cal D}''_{#1}}
\newcommand{\DecctxS}[2]{{\cal D}''_{#1}[#2]}
\newcommand{\decctxP}[1]{{\cal D}'_{#1}}
\newcommand{\Decctx}[2]{{\cal D}_{#1}[#2]}
\newcommand{\DecctxP}[2]{{\cal D}'_{#1}[#2]}
\newcommand{\TypeEnv}[1]{\Gamma_{#1}}

\newcommand{\connected}[3]{#2\stackrel{#1}{\longrightarrow}#3}

\newcommand{\nogarbage}{garbage-free}
\newcommand{\remGarbage}{\EZ{\aux{gc}}}

\newenvironment{myitemize}
               {\begin{itemize}\vspace{-2pt}\topsep0pt\parskip0pt\partopsep0pt\itemsep0pt\leftmargin-100pt\itemsep-1pt\labelwidth0pt\labelsep3pt}
               {\vspace{-1pt}\end{itemize}}
\newenvironment{myitemizeA}
               {\begin{itemize}\vspace{-4pt}\topsep0pt\parskip0pt\partopsep0pt\itemsep0pt\leftmargin-100pt\itemsep-2pt\labelwidth0pt\labelsep3pt}
               {\vspace{-2pt}\end{itemize}}

\newenvironment{mydefinition}
               {
               \begin{definition}\vspace{-1pt}
               }
               {\end{definition}}

\newenvironment{myproposition}
               {\vspace{-1pt}
               \begin{proposition}\vspace{-1pt}
               }
               {\end{proposition}}          
\newenvironment{myexample}
               {\vspace{-1pt}
               \begin{example}\vspace{-1pt}
               }
               {\end{example}}

\newcommand{\mux}{\mu\hspace{.015cm}\x}
\newcommand{\ax}{{\tt a}\hspace{.015cm}\x}
\newcommand{\induced}[1]{\EZ{\sharingRel(#1)}}
\newcommand{\erase}[1]{{#1}^-}

\newif\ifsubmit
\submittrue
\ifsubmit

\newcommand{\TRComm}[1]{}
\newcommand{\EZ}[1]{#1}
\newcommand{\EZComm}[1]{}
\newcommand{\PG}[1]{#1}
\newcommand{\PGComm}[1]{}
\newcommand{\MS}[1]{#1}
\newcommand{\MSComm}[1]{}
\else

\newcommand{\TRComm}[1]{{\scriptsize \textcolor{red}{[Tim{:} #1]}}}
\newcommand{\EZ}[1]{\textcolor{blue}{#1}}
\newcommand{\EZComm}[1]{{\scriptsize \textcolor{blue}{[Elena{:} #1]}}}
\newcommand{\PG}[1]{\textcolor{magenta}{#1}}
\newcommand{\PGComm}[1]{{\scriptsize \textcolor{magenta}{[Paola{:} #1]}}}
\newcommand{\MS}[1]{\textcolor{green}{#1}}
\newcommand{\MSComm}[1]{{\scriptsize \textcolor{green}
{[Marco{:} #1]}}}
\fi

\begin{document}

\begin{frontmatter}

\title{Tracing sharing in an imperative pure calculus\\
(extended version)}

\author[unipo]{Paola Giannini}
\ead{giannini@di.unipmn.it}

\author[unipots]{Tim Richter}
\ead{tim@cs.uni-potsdam.de}

\author[univic]{Marco Servetto}   
\ead{marco.servetto@ecs.vuw.ac.nz}

\author[unige]{Elena Zucca}
\ead{elena.zucca@unige.it}

\address[unipo]{Universit\`{a} del Piemonte Orientale, Italy}
\address[unipots]{Universit\"at Potsdam, Germany}
\address[univic]{Victoria University of Wellington, New Zealand}
\address[unige]{Universit\`a di Genova, Italy}

\begin{abstract}
We introduce a type and effect system, for an imperative object calculus, which
infers \emph{sharing} possibly introduced by the evaluation of an expression,
represented as an equivalence relation among its free variables. This direct
representation of sharing effects at the syntactic level allows us to express
in a natural way, and to generalize, widely-used notions in literature, notably
\emph{uniqueness} and \emph{borrowing}. Moreover, the calculus is
\emph{pure} in the sense that reduction is defined on language terms only,
since they directly encode store. The advantage of this non-standard execution
model with respect to a behaviourally equivalent standard model using a global
auxiliary structure is that reachability relations among references are partly
encoded by scoping.   
\end{abstract}


\end{frontmatter}



\section{Introduction}\label{sect:intro}
In the imperative programming paradigm, \emph{sharing} is the situation when a
portion of the store can be accessed through more than one reference, say $\x$
and $\y$, so that a change to $\x$ affects $\y$ as well. Unwanted sharing
relations are common bugs: unless sharing is carefully maintained, changes
through a reference might propagate unexpectedly, objects may be observed in an
inconsistent state, and conflicting constraints on shared data may
inadvertently invalidate invariants. Preventing such errors is even more
important in increasingly ubiquitous multi-core and many-core architectures. 

For this reason, there is a huge amount of literature on type systems for
controlling sharing and interference, notably using type annotations to
restrict the usage of references, see \refToSection{related} for a survey.
In particular, it is very useful for a programmer to be able to rely on the 
following properties of a reference $\x$.
\begin{myitemize}
  \item \emph{Capsule} reference: $\x$ denotes an isolated portion of store, 
    that is, the subgraph reachable from $\x$ cannot be reached through other 
    references. This allows programmers (and static analysis) to identify state 
    that can be safely handled by a thread.  In this paper we will use the 
    name \emph{capsule} for this property, to avoid confusion with many 
    variants in the literature 
    \cite{ClarkeWrigstad03,Almeida97,ServettoEtAl13a,Hogg91,DietlEtAl07,GordonEtAl12}.
  \item \emph{Lent} reference \cite{ServettoZucca15,GianniniEtAl16}, also 
    called \emph{borrowed} \cite{Boyland01,NadenEtAl12}: the subgraph reachable
    from $\x$ \EZ{can be manipulated by a client, but no sharing can be introduced through $\x$}. Typically,
    borrowing can be employed to ensure that the capsule guarantee is not broken. 
\end{myitemize}

In this paper, we propose a type and effect system which provides, in our
opinion, a very powerful, yet natural, way to express sharing. Notably, the two
above mentioned notions are smoothly included and generalized.

The distinguishing features are the following:
\begin{enumerate}
  \item Rather than declaring type annotations, the type system \emph{infers} 
    sharing possibly introduced by the evaluation of an expression.
  \item Sharing is \emph{directly represented at the syntactic level}, as an 
    equivalence relation among the free variables of the expression.
  \item The calculus is \emph{pure} in the sense that reduction is defined 
    on language terms only, rather than requiring an auxiliary structure.
\end{enumerate}
We now describe these three features in more detail.

Given an expression $\e$, the type system computes a
\emph{sharing relation} $\sharingRel$ which is an equivalence relation on \EZ{a} set
containing the free variables of $\e$ and an additional, distinguished variable $\resV$ denoting the result of
$\e$. That two variables, say $\x$ and $\y$, are in the same equivalence class
in $\sharingRel$ means \EZ{that} the evaluation of $\e$ can possibly introduce sharing
between $\x$ and $\y$, that is, connect their reachable object graphs, so that
a modification of (a subobject of) $\x$ could affect $\y$ as well, or
conversely. For instance, evaluating the expression
$\Sequence{\FieldAssign{\x}{\f}{\y}}{\FieldAccess{\z}{\f}}$
introduces connections
\begin{myitemizeA}
\item between $\x$ and $\y$,
\item between  {$\resV$ (the result)} and $\z$.
\end{myitemizeA}
The \emph{capsule} notion becomes just a special case: an expression is a 
\emph{capsule} iff its result will be disjoint from any free variable
(formally, $\resV$ is a singleton in $\sharingRel$).  For instance, the
expression
$\Sequence{\FieldAssign{\x}{\f}{\y}}{\FieldAccess{\ConstrCall{\C}{\ConstrCall{\D}{}}}{\f}}$
is a capsule, whereas the previous expression is not.\footnote{{Note that our
notion is related to the whole reachable object graph. For instance, a
doubly-linked list whose elements can be arbitrarily aliased can be
externally unique~\cite{ClarkeWrigstad03} and properly follow an owners as
dominator strategy~\cite{ZibinEtAl10}, but is not a capsule.}}

The \emph{lent} notion also becomes a special case: a variable $\x$ is used as
lent in an expression  if the evaluation of the expression will neither connect
$\x$ to any other variable, nor to the result (formally, $\x$ is a singleton in
$\sharingRel$). {In other words, the evaluation of the expression does not
introduce sharing between $\x$ and other variables (including $\resV$).} For
instance $\x$ is lent in
$\Sequence{\FieldAssign{\x}{\f_1}{\x.\f_2}}{\FieldAccess{\z}{\f}}$.  {In our
type system, this notion is generalized from singletons to arbitrary sets of
variables: for instance, in the previous example
$\Sequence{\FieldAssign{\x}{\f}{\y}}{\FieldAccess{\z}{\f}}$,  the set
$\{\x,\y\}$ is an equivalence class in $\sharingRel$ since the evaluation of
the expression does not introduce sharing between this set and other
variables (including $\resV$). 

Altogether, this direct representation at the syntactic level allows us to
express sharing in a natural way. Moreover, execution is modeled by a
  \emph{pure} calculus, where store is encoded directly in language terms,
rather than by an auxiliary structure. Formally, this is achieved by the block
construct, introducing local variable declarations, which play the role of
store when evaluated.  This operational semantics\footnote{Which is, of
course, expected to be behaviorally equivalent to the conventional semantics
where store is a global flat auxiliary structure, as we plan to formally state
and prove in further work.} will be informally introduced in
\refToSection{language}, and formalized in \refToSection{calculus}. 

A preliminary presentation of the approach presented in this paper has been given in \cite{GianniniSZ17,GianniniSZ17a}.

The rest of the paper is organized as follows: in \refToSection{language} we
provide syntax and an informal execution model, in \refToSection{types} the
type system, and in \refToSection{examples} some examples. The operational
semantics of the calculus is presented in \refToSection{calculus}, and the main
results and proofs in \refToSection{results}.  In \refToSection{related} we
discuss related work, and in \refToSection{conclu} we draw some conclusion and
highlight future work. \ref{app:derivation} contains a (rather complex) type
derivation. The proofs omitted from the main paper are in 
\ref{app:proofs}.


\section{Language}\label{sect:language}

The syntax of the language is given in \refToFigure{syntax}.  We assume sets
of \emph{variables} $\x, \y, \z$, \emph{class names} $\C, \D$, \emph{field
names} $\f$, and \emph{method names} $\m$.  We adopt the convention that a
metavariable which ends in \metavariable{s} is implicitly defined as a
(possibly empty) sequence, {for example}, $\decs$ is defined by
$\produzioneinline{\decs}{\epsilon\mid \dec\ \decs}$, where $\epsilon$ denotes
the empty sequence. 

\begin{figure}[t]
{
\begin{grammatica}
\produzione{\e}{\x\mid\FieldAccess{\e}{\f}\mid\FieldAssign{\e}{\f}{\e'}\mid\ConstrCall{\C}{\es}\mid\Block{\decs}{\e}{\mid\MethCall{\e}{\m}{\es}}}{expression}\\
\produzione{\dec}{\Dec{\T}{\x}{\e}}{declaration}\\
\produzione{\T}{\Type{\mu}{\C}\mid\intType}{declaration type}\\
\produzione{\mu}{\epsilon\mid\capsule}{optional modifier}\\
\end{grammatica}
}
\caption{Syntax}\label{fig:syntax}
\end{figure}

The calculus is designed with an object-oriented flavour, inspired {by}
Featherweight Java \cite{IgarashiEtAl01}. This is only a presentation choice:
all the ideas and results of the paper could be easily rephrased in a
different imperative calculus, e.g., supporting data type constructors and
reference types. For the same reason, we omit features such as inheritance and
late binding, which are orthogonal to our focus.

An expression can be a variable (including the special variable $\this$
denoting the receiver in a method body), a field access, a field assignment, a
constructor invocation, a block consisting of a sequence of declarations and a
body, {or a method invocation. In a block, a} declaration specifies a type, a
variable and an initialization expression. We assume that in well-formed blocks
there are no multiple declarations for the same variable{, that is, $\decs$ can
  be seen as a map from variables to expressions. 

A declaration type is a class name with an optional modifier $\capsule$,
{which, if present, indicates that} the variable is \emph{affine}.  We also
include $\intType$ as an example of primitive type, but we do not formally
model related operators used in the examples, such as integer constants and
sum. An affine variable can occur at most once in its scope, and should be
initialized with a \emph{capsule}, that is, an isolated portion of store. In
this way, it can be used as a temporary reference, to ``move'' a capsule to
another location in the store, without introducing sharing.  In the examples,
we generally omit the brackets of the outermost block, and abbreviate
$\Block{\Dec{\T}{\x}{\e}}{\e'}$ by $\Sequence{\e}{\e'}$ when $\x$ does not
occur free in $\e'$.

{We turn our attention to the operational semantics now.}
\refToFigure{examplered} shows an example of a reduction sequence in the
calculus.

\begin{figure}[t]
\begin{lstlisting}[basicstyle=\ttfamily\scriptsize,backgroundcolor=\color{white}]
$\store{D z=new D(0); C x=new C(z,z);}$ $\ul{C y=x;}$ D w=new D(y.f1.f+1); x.f2=w; x $\longrightarrow$
$\store{D z=new D(0); C x=new C(z,z);}$ D w=new D($\ul{x.f1}$.f+1); x.f2=w; x $\longrightarrow$
$\store{D z=new D(0); C x=new C(z,z);}$ D w=new D($\ul{z.f}$+1); x.f2=w; x $\longrightarrow$
$\store{D z=new D(0); C x=new C(z,z);}$ D w=new D($\ul{0+1}$); x.f2=w; x $\longrightarrow$
$\store{D z=new D(0); C x=new C(z,z); D w=new D(1);}$ $\ul{x.f2=w}$; x $\longrightarrow$
$\store{D z=new D(0); C x=new C(z,w); D w=new D(1);}$ x
\end{lstlisting} 
\caption{Example of reduction}\label{fig:examplered}
\end{figure}

The main idea is to use variable declarations to directly represent the store.
That is, a declared {(non affine)} variable is not replaced by its value, as in
standard \texttt{let}, but the association is kept and used when necessary{, as
it happens, with different aims and technical problems, in cyclic lambda
calculi \cite{AriolaFelleisen97,MaraistEtAl98}.}

In the figure, we emphasize at each step the declarations which can be seen
as the store (in grey) and the redex which is reduced (in a box).

Assuming a program  (class table) where class \lstinline{C} has two fields
\lstinline{f1} and \lstinline{f2} of type \lstinline{D}, and class
\lstinline{D} has an integer field \lstinline{f}, in the initial term in
\refToFigure{examplered} the first two declarations can be seen as a store
which associates to \lstinline{z} an object of class \lstinline{D} whose field
contains \lstinline{0}{}, and to \lstinline{x}{} an object of class
\lstinline{C}{} whose two fields contains (a reference to) the previous object.
The first reduction step eliminates an alias, by replacing occurrences of
\lstinline{y} by \lstinline{x}. The next three reduction steps compute
\lstinline{x.f1.f+1}{}, by performing two field accesses and one sum.  The last
step performs a field assignment, modifying the current store.  Finally, in
the last line we have a term which can no longer be reduced, consisting of a
store and the expression \texttt{x} which denotes a reference in the store,
taken as entry point. In other words, the final result of the evaluation is an
object of class \lstinline{C} whose fields contain (references to) two
objects of class \lstinline{D}, whose fields contain \lstinline{0} and
\lstinline{1}, respectively. 

As usual, references in the store can be mutually recursive\footnote{However,
mutual recursion is not allowed between declarations which are {\emph{not
evaluated}}, e.g., {\lstinline{B x= new B(y.f); B y= new B(x.f); y}} is
ill-formed.}, as in the following example, where we assume a class
\lstinline{B} with a field of type \lstinline{B}.
\begin{lstlisting}
B x= new B(y); B y= new B(x); y 
\end{lstlisting}
Again, this is a term which can no longer be reduced,  consisting of a
store and the reference \texttt{y} as entry point. In other words, this term
can be seen as an object of class \lstinline{B} whose field contains (a
reference to) another object of class \lstinline{B}, whose field contains (a
reference to) the original object.

In the examples until now, store is flat, as it usually happens in models of
imperative languages. However, in our calculus,  we are also able to represent
a hierarchical store, as shown in the example below, where we assume a class
\lstinline{A}{} with two fields of type \lstinline{B} and \lstinline{D},
respectively.
\begin{lstlisting}
D z= new D(0);  
A w= {
  B x= new B(y);   
  B y= new B(x);
  A u= new A(x,z);
  u};  
w
\end{lstlisting}

Here, the store associates to \lstinline{w}{} a block introducing local
declarations, that is, in turn a store.  The advantage of this representation
is that it models in a simple and natural way constraints about sharing among
objects, notably:
\vspace{3pt}
\begin{myitemize}
  \item the fact that an object is not referenced from outside some enclosing 
    object is directly modeled by the block construct: for instance, the 
    object denoted by \lstinline{y}{} can only be reached through \lstinline{w}{}
  \item conversely, the fact that an object does not refer to the outside is 
    modeled by the fact that the corresponding block is closed (that is, has no 
    free variables): for instance, the object denoted by \lstinline{w}{} is not 
    closed, since it refers to the external object \lstinline{z}{}.
\end{myitemize}
In other words, our calculus smoothly integrates memory representation with
shadowing and $\alpha$-conversion.  However, there is a problem which needs to
be handled to keep this representation correct: reading (or, symmetrically,
updating) a field could cause scope extrusion. For instance, the term 
\begin{lstlisting}
C y= {D z= new D(0); C x= new C(z,z); x};  y.f1
\end{lstlisting}
under a naive reduction strategy would reduce to the ill-formed term
\begin{lstlisting}
C y= {D z= new D(0); C x= new C(z,z); x};  z
\end{lstlisting}
To avoid this problem, the above reduction step is forbidden. However,
reduction is not stuck, since we can transform the above term into an
equivalent term where the inner block has been flattened, and get the following
correct reduction sequence:
\begin{lstlisting}
C y= {D z= new D(0); C x= new C(z,z); x}  y.f1 $\congruence{}{}$
D z= new D(0); C x= new C(z,z); C y= x; y.f1 $\longrightarrow$
D z= new D(0); C x= new C(z,z); x.f1 $\longrightarrow$
D z= new D(0); C x= new C(z,z); z $\congruence{}{}$
D z= new D(0); z
\end{lstlisting}
Formally, in addition to the reduction relation which models actual
computation, our operational semantics is defined by a \emph{congruence}
relation $\congruence{}{}$, which captures structural equivalence, as in
$\pi$-calculus \cite{Milner99}.  Note also that in the final term the
declaration of \lstinline{x}{} can be removed (again by congruence), since it
is useless.

Moving declarations from a block to the directly enclosing block is not always
safe. For instance, in the following variant of the previous example 
\begin{lstlisting}
C$^\capsule$ y= {D z= new D(0); C x= new C(z,z); x};  y.f1
\end{lstlisting}
the affine variable is required to be initialized with a capsule, and this is
the case indeed, since the right-hand side of the declaration is a closed
block. However, by flattening the term:
\begin{lstlisting}
D z= new D(0); C x= new C(z,z); C$^\capsule$ y= x; y.f1
\end{lstlisting}
this property would be lost, and we would get an ill-typed term. {Indeed, these
two terms are \emph{not} considered equivalent in our operational model.
Technically, this is obtained by detecting, during typechecking, which local
variables will be connected to the result of the block, as \lstinline{z} in
the example, and preventing to move such declarations from a block which is the
initialization expression of an affine variable. } 

In this case, reduction proceeds by replacing the (unique) occurrence of the
affine variable by its initialization expression, as shown below. 
\begin{lstlisting}
C$^\capsule$ y= {D z= new D(0); C x= new C(z,z); x}  y.f1 $\longrightarrow$
{D z= new D(0); C x= new C(z,z); x}.f1 $\congruence{}{}$
D z= new D(0); C x= new C(z,z); x.f1 $\longrightarrow$
D z= new D(0); C x= new C(z,z); z $\congruence{}{}$
D z= new D(0); z
\end{lstlisting}


\section{Type system}\label{sect:types}

In this section we introduce the {type and effect} system for the language.  We
use $\X,\Y$ to range over sets of variables.

A {\em sharing relation $\sharingRel$} on a set of variables $X$ is an
equivalence relation on $X$.\PGComm{Removed:, called {\em the domain of $\sharingRel$} and
dubbed $\dom{\sharingRel}$.}  As usual $\Closure{x}{\sharingRel}$ denotes the
{\em equivalence class of $\x$ in $\sharingRel$}.  We will call the
elements $\Pair{\x}{\y}$ of a sharing relation \emph{connections}, and say
that $\x$ and $\y$ are \emph{connected}. The intuitive meaning is that, if
$\x$ and $\y$ are connected, then their reachable graphs in the store are
possibly shared (that is, not disjoint), hence a modification of the
reachable graph of $\x$ could affect $\y$ as well, or conversely.

We use the following notations on sharing relations:
\begin{itemize}
  \item A sequence
    of subsets of $X$, say, $\X_1\ldots\X_n$, represents the smallest
    equivalence relation on $X$ containing the connections $\Pair{\x}{\y}$, for all $\x,\y$ belonging to the same $\X_i$. So, $\epsilon$
    represents the identity relation on any set of variables. \EZ{Note that this representation is deliberately ambiguous as to the
domain of the defined equivalence: any common superset of the $\X_i$
will do.}\EZComm{suggested by Tim}
  \item $\Sum{\sharingRel_1}{\sharingRel_2}$ is the 
    smallest equivalence relation containing $\sharingRel_1$ and $\sharingRel_2$. 
    It is easy to show that sum is commutative and associative. \PG{With $\Sum{\sharingRel}{\X}$ we denote the
    sum of $\sharingRel$ with the sharing relation containing the connections $\Pair{\x}{\y}$, for all $\x,\y\in\X$.}
  \item $\Remove{\sharingRel}{\X}$ is the sharing relation \PGComm{Removed: on 
    $\dom{\sharingRel}\setminus \X$} obtained by ``removing'' $\X$ from
    $\sharingRel$, that is, the smallest equivalence relation containing the
    connections $\Pair{\x}{\y}$, for all $\Pair{\x}{\y}\in\sharingRel$ such
    that $\x,\y\not\in\X$. $\Remove{\sharingRel}{\y}$ stands for
    $\Remove{\sharingRel}{\{\y\}}$. It is easy to see that 
    $\Remove{\sharingRel}{(\X\cup\Y)}=\Remove{(\Remove{\sharingRel}{\X})}{\Y}$.
\item \EZ{$\SubstEqRel{\sharingRel}{\y}{\x}$ is the sharing relation \PGComm{Removed: on $\dom{\sharingRel}\setminus \{\x\}$} obtained by ``replacing'' $\x$ by $\y$ in $\sharingRel$, that is, the smallest equivalence relation containing the  connections:\\
$\Pair{\z}{\z'}$, for all  $\Pair{\z}{\z'}\in\sharingRel$, $\z\neq\x, \z'\neq\x$\\
$\Pair{\y}{\z}$, for all $\Pair{\x}{\z}\in\sharingRel$.}
 \item {\em $\sharingRel_1$ has less (or equal) sharing effects
than $\sharingRel_2$}, dubbed $\Finer{\sharingRel_1}{\sharingRel_2}$, if, for all
$\x$, 
$\Closure{\x}{\sharingRel_1}\subseteq\Closure{\x}{\sharingRel_2}$.
\end{itemize}

The following proposition asserts some properties of sharing relations. 
\begin{proposition}\label{prop:lessSrRel}\
\begin{enumerate}
    \item \label{p1} Let $\x\neq\y$, $\Pair{\x}{\y}\in\sum\limits_{i=1}^{n}\sharingRel_i$ if and only if
    there are sequences $i_1\dots i_{k-1}$ ($1\leq i_h\leq n$ for all $h$) and $\z_1\dots\z_k$ ($k> 1$) such that $\x=\z_1$ and  $\y=\z_k$
    and $\Pair{z_j}{\z_{{j+1}}}\in\sharingRel_{i_{j}}$ and $i_j\neq i_{j+1}$ and $\z_j\neq \z_{j+1}$ for $1\leq j\leq (k-1)$.
  \item \label{p2} $\Finer{\sharingRel_1}{\sharingRel_2}$ implies 
    $\Finer{\Sum{\sharingRel}{\sharingRel_1}}{\Sum{\sharingRel}{\sharingRel_2}}$ for all $\sharingRel$.
  \item \label{p3} $\Finer{\sharingRel_1}{\sharingRel_2}$ implies $\Finer{\Remove{\sharingRel_1}{\X}}{\Remove{\sharingRel_2}{\X}}$ for all $\X$. 
  \item If \label{p4}$\Remove{\sharingRel_1}{\X}={\sharingRel_1}$, then
   $\Remove{(\Sum{\sharingRel_1}{\sharingRel_2})}{\X}=\Sum{\Remove{\sharingRel_1}{\X}}{\Remove{\sharingRel_2}{\X}}$.
 \item If \label{p5}$\y\in\Closure{\x}{\sharingRel}$, then $\sharingRel[\y/\x]=\Remove{\sharingRel}{\x}$.
\end{enumerate}
Since $\sharingRel+\epsilon=\sharingRel$ and $\Finer{\epsilon}{\sharingRel}$
for all $\sharingRel$, from 2. we have that
$\Finer{\sharingRel}{\Sum{\sharingRel}{\sharingRel'}}$ for all $\sharingRel$
and $\sharingRel'$.
\end{proposition}
\begin{proof}\
\begin{enumerate}
\item From the fact that $\sum\limits_{i=1}^{n}\sharingRel_i$ is the transitive closure of 
$\bigcup_{1\leq i\leq n}\{\Pair{\x}{\y}\ |\ \Pair{\x}{\y}\in\sharingRel_i\}$.
 \item From 1. and  the fact that for all $\z$ and $\z'$, if $\Pair{\z}{\z'}\in{\sharingRel_1}$ then $\Pair{\z}{\z'}\in{\sharingRel_2}$. 
 \item Let $\Pair{\z}{\z'}\in\Remove{\sharingRel_1}{\X}$ with $\z\neq\z'$. Then  $\Pair{\z}{\z'}\in{\sharingRel_1}$ and $\z,\z'\not\in\X$. From
 $\Finer{\sharingRel_1}{\sharingRel_2}$,   then $\Pair{\z}{\z'}\in{\sharingRel_2}$ and so also $\Pair{\z}{\z'}\in\Remove{\sharingRel_2}{\X}$. 
 \\
\PG{Let $\Pair{\z}{\z'}$ be such that  $\Pair{\z}{\z'}\in{\SubstEqRel{\sharingRel_1}{\y}{\x}}$ and $\z\neq\z'$.
Then $\z\neq\x$ and $\z'\neq\x$. If $\Pair{\z}{\z'}\in{{\sharingRel_1}}$, then 
 $\Pair{\z}{\z'}\in{{\sharingRel_2}}$ and so also $\Pair{\z}{\z'}\in{\SubstEqRel{\sharingRel_2}{\y}{\x}}$. 
 If $\Pair{\z}{\z'}\not\in{{\sharingRel_1}}$, then there are pairs $\Pair{\z}{\x}$ and $\Pair{\z'}{\y}$ 
 such that $\Pair{\z}{\x}\in{{\sharingRel_1}}$ and $\Pair{\y}{\z'}\in{{\sharingRel_1}}$. Therefore
 $\Pair{\z}{\x}\in{{\sharingRel_2}}$ and $\Pair{\y}{\z'}\in{{\sharingRel_2}}$ and so 
 $\Pair{\z}{\z'}\in{\SubstEqRel{\sharingRel_2}{\y}{\x}}$.}
 \item 
From  2. and 3. we have that $\Finer{\Remove{\sharingRel_1}{\X}}{\Remove{(\Sum{\sharingRel_1}{\sharingRel_2})}{\X}}$ and $\Finer{\Remove{\sharingRel_2}{\X}}{\Remove{(\Sum{\sharingRel_1}{\sharingRel_2})}{\X}}$. By definition of $+$ we
derive 
$\Finer{\Sum{\Remove{\sharingRel_1}{\X}}{\Remove{\sharingRel_2}{\X}}}{\Remove{(\Sum{\sharingRel_1}{\sharingRel_1})}{\X}}$.\\
We now prove that $\Finer{\Remove{(\Sum{\sharingRel_1}{\sharingRel_2})}{\X}}{\Sum{\Remove{\sharingRel_1}{\X}}{\Remove{\sharingRel_2}{\X}}}$.
Let $\Pair{\x}{\y}\in\Remove{(\Sum{\sharingRel}{\sharingRel'})}{\X}$ with $\x\neq\y$. Then $\Pair{\x}{\y}\in\Sum{\sharingRel}{\sharingRel'}$
and $\x,\y\not\in\X$.
By 1., there are sequences $i_1\dots i_{k-1}$ and $\z_1\dots\z_k$ ($k> 1$) such that $\x=\z_1$ and  $\y=\z_k$
and $\Pair{z_j}{\z_{{j+1}}}\in\sharingRel_{i_{j}}$ and $i_j\neq i_{j+1}$ and $\z_j\neq \z_{j+1}$ for $1\leq j\leq (k-1)$.
The fact $i_j\neq i_{j+1}$ implies that the sequence $i_1\dots i_{k-1}$ alternates between $1$ and $2$. So for 
any $j$, $2\leq j\leq (k-1)$,  $\Pair{z_{j-1}}{\z_{{j}}}\in\sharingRel_{1}$ or 
$\Pair{z_{j}}{\z_{{j+1}}}\in\sharingRel_{1}$. Since $\sharingRel_{1}=\Remove{\sharingRel_{1}}{\X}$, in either cases 
$\z_j \not\in\X$. So no element of $\z_1\dots\z_k$ is in $\X$, thus for any $j$, $1\leq j\leq (k-1)$,
$\Pair{z_{j}}{\z_{{j+1}}}\in\sharingRel_{i_j}$ implies $\Pair{z_{j}}{\z_{{j+1}}}\in\Remove{\sharingRel_{i_j}}{\X}$.
By 1. we have that  $\Pair{\x}{\y}\in\Sum{\Remove{\sharingRel_1}{\X}}{\Remove{\sharingRel_2}{\X}}$.
\item Let $\Pair{\z}{\z'}\in{\Remove{\sharingRel}{\x}}$ with $\z\neq\z'$, if $\z\neq\x$ and $\z'\neq\x$, then 
$\Pair{\z}{\z'}\in\sharingRel[\y/\x]$. So $\Remove{\sharingRel}{\x}\subseteq\sharingRel[\y/\x]$. \\
To show $\sharingRel[\y/\x]\subseteq\Remove{\sharingRel}{\x}$, first observe that there cannot be $\Pair{\z}{\z'}\in\sharingRel[\y/\x]$ such that $\z\neq\z'$ and
either $\z=\x$ or $\z'=\x$.\\ 
Let $\Pair{\z}{\z'}\in\sharingRel[\y/\x]$ with $\z\neq\z'$. If $\z\neq\y$ and $\z'\neq\y$, then, by definition of $\sharingRel[\y/\x]$
we get $\Pair{\z}{\z'}\in\Remove{\sharingRel}{\x}$. If $\z=\y$ there are
2 cases: either $\Pair{\y}{\z'}\in\sharingRel$, or $\Pair{\x}{\z'}\in\sharingRel$. In the first case $\Pair{\y}{\z'}\in\Remove\sharingRel\x$.
In the second, from $\y\in\Closure{\x}{\sharingRel}$ we get that $\Pair{\x}{\z'}\in\sharingRel$ implies 
$\Pair{\y}{\z'}\in\sharingRel$ and so also $\Pair{\y}{\z'}\in\Remove\sharingRel\x$. Similar if  $\z'=\y$. Therefore $\sharingRel[\y/\x]\subseteq\Remove{\sharingRel}{\x}$.
\end{enumerate}
\end{proof}

The typing judgment has shape 
\begin{center}
$\TypeCheckAnnotate{\Gamma}{\e}{\T}{\sharingRel}{\e'}$
\end{center}
where $\Gamma$ is a \emph{type environment}, that is, an assignment of types
to variables, written $\TypeDec{\x_1}{\T_1},\ldots,\TypeDec{\x_1}{\T_n}$, \EZ{$\T$ is a type\footnote{\EZ{Note that types of shape $\Type{\capsule}{\C}$ only occur as declaration types and in $\Gamma$, whereas they are never assigned to expressions. However, we use the same metavariable for simplicity.}}}, {$\sharingRel$ is a sharing relation, and $\e'$ is an \emph{annotated expression}.
\PG{The sharing relation $\sharingRel$ is defined on the variables in $\Gamma$ plus a distinguished variable $\resV$  for ``result''.} 
The intuitive
meaning is that $\sharingRel$ represents the connections among the free
variables of $\e$ possibly introduced by the evaluation of the expression, and
the variables in $\Closure{\resV}{\sharingRel}$ are the ones that {will be}
possibly connected to the result of the expression.  
 
We write $\IsCapsule{\sharingRel}$ if
$\Closure{\resV}{\sharingRel}=\{\resV\}$. If $\TypeCheckAnnotate{\Gamma}{\e}{\T}{\sharingRel}{\e'}$ with $\IsCapsule{\sharingRel}$, then the expression $\e$
denotes a \emph{capsule}, that is, reduces to an isolated portion of store\EZ{, as will be formally shown in \refToSection{results} (\refToTheorem{capsule})}. An
affine variable will never be connected to another, nor to $\resV$, since it is
initialized with a capsule and used only once. Analogously, a variable of a
primitive type will never be connected to another.

Moreover, during typechecking expressions are annotated. The syntax of {\em
annotated expressions} is given by:
\begin{center}
$
\begin{array}{lcl}
\e& ::= & \x\mid\FieldAccess{\e}{\f}\mid{\MethCall{{\e}}{\m}{{\e_1}, \ldots, {\e_n}}}\mid\FieldAssign{\e}{\f}{\e'}\mid\ConstrCall{\C}{\es}\mid\BlockLab{\decs}{\e}{\X}
\end{array}
$
\end{center}
\PG{where $\X\subseteq\dom{\decs}$.} We use the same metavariable of source expressions for simplicity.  As we
can see, the only difference is that blocks are annotated by {a} set $\X$ of variables. 
In an annotated block obtained as output of typechecking, $\X$ will be the local variables declared in the block
(possibly) connected with the result of the body, see rule \rn{t-block}.  Such
annotations, as we will see in the next section, are used to define
the congruence relation among terms. \EZ{Notably, we can move local store from a block to
the directly enclosing block, or conversely, as it happens with rules for
\emph{scope extension} in the $\pi$-calculus \cite{Milner99}. 
However, this is not allowed if such block initializes an affine variable declaration, and we would move outside
variables} possibly connected to the result of the block. Indeed, this would make the term ill-typed, as shown in the last
example of \refToSection{language}.

The class table is abstractly modeled by the following functions:
\begin{itemize}
  \item $\fields{\C}$ gives, for each declared class $\C$, the sequence 
    $\Field{\T_1}{\f_1}\ldots\Field{\T_n}{\f_n}$ of its fields 
    declarations, with $\T_i$ either class name or primitive type.\footnote{That 
    is, the $\capsule$ modifier, denoting a temporary reference, makes no sense for fields.}
  \item $\method{\C}{\m}$ gives, for each method $\m$ declared in class $\C$, the tuple \\
    $\Method{\ReturnTypeNew{\T}{\sharingRel}}{\mu}{\Param{\T_1}{\x_1}\ldots\Param{\T_n}{\x_n}}{\e}$ 
    consisting of its return type paired with the resulting sharing relation, optional 
    $\capsule$ modifier for $\this$, parameters, and body.
\end{itemize} 
We assume a well-typed class table, that is, method bodies are expected to be
well-typed with respect to method types. Formally, if \\
$\method{\C}{\m}=\Method{\ReturnTypeNew{\T}{\sharingRel}}{\mu}{\Param{\T_1}{\x_1}\ldots\Param{\T_n}{\x_n}}{\e}$,
then it should be 
\begin{itemize}
  \item $\TypeCheckAnnotate{\Gamma}{\e}{\T}{\sharingRel}{\e'}$, with
  \item  $\Gamma=\TypeDec{\this}{\Type{\mu}{\C},\TypeDec{\x_1}{\T_1},\ldots,\TypeDec{\x_n}{\T_n}}$.
\end{itemize} 

{Note that the $\sharingRel$ effects in the return type of a method can be
inferred by typechecking the body for a non-recursive method.  Recursion
could be handled by a global fixed-point inference to find effects across
methods. Alternatively, and also to support interfaces, (some) effect
annotations in method return types could be supplied by the programmer,
likely in a simpler form, e.g., using the \emph{capsule} modifier. In this
case, typechecking the body should check conformance to its declared
interface. Still, the fixed-point inference scheme would be useful in porting
over code-bases, and might help to identify how effective the type system is in
practice. We leave this matter to further work.

The typing rules are given in \refToFigure{typing}.

\begin{figure}[t]
\begin{small}
\begin{math}
\begin{array}{l}
\NamedRule{t-var}{}{\TypeCheckAnnotate{\Gamma}{\x}{\C}{\{\x,\resV\}}{\x}}{
\Gamma(\x)=\C
}\Space
\NamedRule{t-affine-var}{}{\TypeCheckAnnotate{\Gamma}{\x}{\C}{\epsilon}{\x}}{
\Gamma(\x)=\T\\
\T=\Type{\capsule}{\C}\mid\intType
}
\\[5ex]
\NamedRule{t-field-access}{\TypeCheckAnnotate{\Gamma}{\e}{\C}{\sharingRel}{\e'}}{\TypeCheckAnnotate{\Gamma}{\FieldAccess{\e}{\f_i}}{\T_i}{\sharingRel}{\FieldAccess{\e'}{\f_i}}}{
\fields{\C}=\Field{\T_1}{\f_1}\ldots\Field{\T_n}{\f_n}\\
i\in 1..n}
\\[4ex]
\NamedRule{t-field-assign}{
  \TypeCheckAnnotate{\Gamma}{\e_1}{\C}{\sharingRel_1}{\e'_1}
  \BigSpace
  \TypeCheckAnnotate{\Gamma}{\e_2}{\T_i}{\sharingRel_2}{\e'_2}
 }{
  \TypeCheckAnnotate{\Gamma}{\FieldAssign{\e_1}{\f_i}{\e_2}
  }{\MS{\T_i}}{\Sum{\sharingRel_1}{\sharingRel_2}}{\FieldAssign{\e'_1}{\f_i}{\e'_2}}
  }
  {\fields{\C}=\Field{\T_1}{\f_1}\ldots\Field{\T_n}{\f_n}\\
  i\in 1..n
    }
\\[4ex]
\NamedRule{t-new}{ \TypeCheckAnnotate{\Gamma}{\e_i}{\T_i}{\sharingRel_i}{\e'_i}
 \BigSpace
1{\leq}i{\leq}n}
{
\TypeCheckAnnotate{\Gamma}{\ConstrCall{\C}{\e_1,\ldots,\e_n}}
{\C}
{\sum\limits_{i=1}^{n}\sharingRel_i}
{\ConstrCall{\C}{\e'_1,\ldots,\e'_n}}
}
{
\fields{\C}{=}\Field{\T_1}{\f_1}\ldots\Field{\T_n}{\f_n}
}
\\[4ex]
\NamedRule{t-block}{
\begin{array}{l}
\TypeCheckAnnotate{\SubstFun{\Gamma}{\Gamma'}}{\e_i}{\T_i}{\sharingRel_i}{\e'_i}\Space 1{\leq}i{\leq}n\\
\TypeCheckAnnotate{\SubstFun{\Gamma}{\Gamma'}}{\e}{\T}{\sharingRel'}{\e'}
\end{array}
}
{
\begin{array}{l}
\TypeCheckAnnotate{\Gamma}{\Block{\DecP{\T_1}{\x_1}{\e_1}\ldots\DecP{\T_n}{\x_n}{\e_n}}{\e}}{\T}
{{\Remove{\sharingRel}{\dom{\Gamma'}}}}{}\\
\quad\quad\BlockLab{\DecP{\T_1}{\x_1}{\e'_1}\ldots\DecP{\T_n}{\x_n}{\e'_n}}{\e'}
{{\Closure{\resV}{\sharingRel}}\cap\dom{\Gamma'}}
\end{array}
}
{
\begin{array}{l}
\Gamma'=\TypeDec{\x_1}{\T_1},\ldots,\TypeDec{\x_n}{\T_n}\\
\forall {1{\leq}i{\leq}n}\ \  \T_i{=}\Type{\capsule}{\C_i}\Longrightarrow\\
\ \ \ {\IsCapsule{\sharingRel_i}}\wedge \x_i\ \mbox{affine}\\
\EZ{\sharingRel'_i=\SubstEqRel{\sharingRel_i}{\x_i}{\resV}}\\
\EZ{\sharingRel=\Sum{\sum\limits_{i=1}^{n}\sharingRel'_i}{\sharingRel'}}
\end{array}
}
\\[9ex]
\NamedRule{{t-invk}}{
\begin{array}{l}
\TypeCheckAnnotate{\Gamma}{\e_0}{\C}{\sharingRel_0}{\e'_0}\\
\TypeCheckAnnotate{\Gamma}{\e_i}{\T_i}{\sharingRel_i}{\e'_i}\BigSpace
0{\leq}i{\leq}n
\end{array}
}
{\begin{array}{l}
\TypeCheckAnnotate{\Gamma}{\MethCall{\e_0}{\m}{\e_1,\ldots,\e_n}}{\T}{\sharingRel}{}\\
\quad\quad\quad\quad\quad\quad{\MethCall{{\e'_0}}{\m}{{\e'_1},\ldots,{\e'_n}}}
\end{array}}
{\begin{array}{l}
\method{\C}{\m}{=}\Method{\ReturnTypeNew{\T}{\sharingRel'}}{\mu}{\Param{\T_1}{\x_1}\ldots\Param{\T_n}{\x_n}}{\e}\\
\mu=\capsule\Longrightarrow\IsCapsule{\sharingRel_0}\\
\forall {1{\leq}i{\leq}n}\ \T_i{=}\Type{\capsule}{\C_i}\Longrightarrow{\IsCapsule{\sharingRel_i}}\\
\EZ{\sharingRel'_0=\SubstEqRel{\sharingRel_0}{\this}{\resV}}\BigSpace\EZ{\sharingRel'_i=\SubstEqRel{\sharingRel_i}{\x_i}{\resV}}\\
\sharingRel=\Remove{(\Sum{\sum\limits_{i={0}}^{n}\sharingRel'_i}{\sharingRel'})}{\{\this,\x_1,\ldots,\x_n\}}
\end{array}}
\\[11ex]
\end{array}
\end{math}
\end{small}
\caption{Typing rules}\label{fig:typing}
\end{figure}

In rule \rn{t-var}, the evaluation of a variable (if neither affine nor of
a primitive type)  connects the result of the expression with the variable
itself.  In rule {\rn{t-affine-var}}, the evaluation of an affine variable does
not introduce any connection, so the resulting sharing relation is the identity
relation. Indeed, affine variables are temporary references and will be
substituted with capsules. The same happens for variables of primitive
types.\\
In rule \rn{t-field-access}, the connections introduced by a field access are
those introduced by the evaluation of the receiver.\\
In rule \rn{t-field-assign}, the connections introduced by a
field assignment are those introduced by the evaluation of the two expressions
($\sharingRel_1$ and $\sharingRel_2$). Since both $\sharingRel_1$ and
$\sharingRel_2$ contain the variable $\resV$, the equivalence class of this
variable in the resulting sharing relation is, as expected, the (transitive
closure of the) union of the two equivalence classes. For instance, given the
assignment $\FieldAssign{\e}{\f}{\e'}$, if the evaluation of $\e$ connects $\y$
with $\z$ and $\x$ with its result, and the evaluation of $\e'$ connects $\y'$
with $\z'$ and $\x'$ with its result, then the evaluation of the field
assignment connects $\y$ with $\z$, $\y'$ with $\z'$,  and {both $\x$ and $\x'$
with the result}. \\
In rule \rn{t-new}, the connections  introduced by a constructor
invocation are those introduced by the evaluation of the arguments. As for
field assignment, the equivalence class of $\resV$ in the resulting sharing
relation is, as expected, the (transitive closure of the) union of the
equivalence classes of $\resV$ in the sharing relations of the arguments of
the constructor. \\
In rule \rn{t-block}, the initialization expressions and the body of the block
are typechecked in the current type {environment}, enriched by the association
to local variables of their declaration types.  We denote by
$\SubstFun{\Gamma}{\Gamma'}$ the type environment which to a variable $\x$ assigns 
  $\Gamma'(\x)$ if this is defined, and $\Gamma(\x)$ otherwise. If a local variable is
affine, then its initialization expression is required to denote a capsule. Moreover, the variable can
occur at most once in its scope, as abbreviated by the side condition ``$\x_i$
affine''.\footnote{{In our case the affinity requirement can be simply
expressed as syntactic well-formedness condition, rather than by context rules,
as in \EZ{affine} type systems.}} The connections 
introduced by a block are obtained modifying those introduced by the evaluation
of the initialization expressions {($\sharingRel_i, 1{\leq}i{\leq}n$)} plus
those introduced by the evaluation of the body {$\sharingRel'$}. More
precisely\EZ{, for each declared variable, the connections of the result of the initialization expression are transformed in connections to the variable itself. Finally, we remove from the resulting sharing relation the local variables.}
The block is annotated with the subset of local variables which are in the
sharing relation $\sharingRel$ with the result of the block.\\
In rule \rn{t-invk}, the typing of $\MethCall{\e_0}{\m}{\e_1,\ldots,\e_n}$ is
similar to the typing of the block
$\Block{\Dec{\EZ{\Type{\mu}{\C}}}{\this}{\e_0}\,\Dec{\T_1}{\x_1}{\e_1}\ldots\Dec{\T_n}{\x_n}{\e_n}}{\e}$.
For
instance, assume that method $\m$ has parameters $\x$ and $\y$, and the
evaluation of its body possibly connects $\x$ with $\this$, and $\y$ with the
result, i.e., the sharing relation associated to the method is
$\sharingRel'=\{\x,\this\}\{\y,\resV\}$. Then, the evaluation of the method
call $\MethCall{\z}{\m}{\x',\y'}$,  possibly connects $\x'$ with $\z$, and
$\y'$ with the result of the expression, i.e., has sharing effects
$\{\x',\z\}\{\y',\resV\}$.

Finally, note that primitive types are used in the standard way. For
instance, in the premise of rule \rn{t-new} the types of constructor
arguments could be primitive types, whereas in rule \rn{t-meth-call} the type
of receiver could not. 

\PG{The following proposition formalizes some properties of the \EZ{typing judgment}. Notably, 
if two different variables are 
in sharing relation, then they must have a reference type and cannot be affine. This is
true also for variables in sharing relation with the result of an expression.}
So affine
variables are always singletons in the sharing relation. In the following proposition we omit the annotations of  terms, 
which are irrelevant.

\PGComm{If we have time could be simplified!}
\begin{proposition}\label{prop:invTyping1}
Let $\TypeCheck{\Gamma}{\e}{\D}{\sharingRel}$. If $\Pair{\x}{\y}\in\sharingRel$ and $\x\neq\y$, then 
\begin{itemize}
\item if $\x\neq\resV$ and $\y\neq\resV$, then $\Gamma(\x)=\PG{\C}$ and $\Gamma(\y)=\PG{\C'}$ (for some $\C$ and $\C'$) and $\EZ{\x,\y\in}\FV{\e}$.
\item if $\y=\resV$ or $\x=\resV$, then $\Gamma(\x)={\C}$ or $\Gamma(\y)={\C}$ (for some $\C$) and $\x\in\FV{\e}$ or $\y\in\FV{\e}$.
\end{itemize}
\end{proposition}
\begin{proof}
The proof is by induction on the type derivation $\TypeCheck{\Gamma}{\e}{\EZ{\D}}{\sharingRel}$. 
Consider the last \EZ{typing} rule used in the type derivation.\\
\underline{Rule \rn{T-Var}}. In this case $\e=\x$, $\Gamma(\x)=\D$, and $\FV{\e}=\{\x\}$. The only non trivial equivalence 
class is $\{\x,\resV\}$. Therefore the result holds. \\
\underline{Rule \rn{T-Affine-Var}}. In this case $\e=\x$ and $\sharingRel$ is the identity. Therefore there is no $\Pair{\x}{\y}\in\sharingRel$ such that $\x\neq\y$ and the result  holds trivially. \\
\underline{Rule \rn{T-Field-Access}}. In this case $\e=\FieldAccess{\e_1}{\f}$ and the result derives by induction hypothesis on $\e_1$.\\
\underline{Rule \rn{T-Field-Assign}}. In this case $\e=\FieldAssign{\e_1}{\f_i}{\e_2}$ and $\sharingRel=\sharingRel_1+\sharingRel_2$
and $\TypeCheck{\Gamma}{\e_i}{\C}{\sharingRel_i}$ for $i=1,2$. By induction hypotheses, we have that 
if $\Pair{\z}{\z'}\in\sharingRel_i$ and $\z\neq\z'$
 \begin{enumerate} [(1)]
      \item if $\z\neq\resV$ and $\z'\neq\resV$, then $\Gamma(\z)=\C$ and $\Gamma(\z')=\C'$ (for some $\C$ and $\C'$) and $\{\z,\z'\}\subseteq\FV{\e_h}$
  for $1{\leq}h{\leq}2$
      \item if either $\z=\resV$ or $\z'=\resV$, then $\Gamma(\z')={\C}$ or $\Gamma(\z)={\C}$ (for some $\C$) and  $\z'\in\FV{\e_h}$ 
      or $\z\in\FV{\e_i}$
       for $1{\leq}h{\leq}2$
\end{enumerate}
If $\Pair{\x}{\y}\in\sharingRel$ and $\x\neq\y$, then, by \refToProp{lessSrRel}.\ref{p1}, there are sequences $i_1\dots i_{k-1}$ and $\z_1\dots\z_k$ ($k> 1$) such that $\x=\z_1$ and  $\y=\z_k$
and $\Pair{z_j}{\z_{{j+1}}}\in\sharingRel_{i_{j}}$ and $i_j\neq i_{j+1}$ and $\z_j\neq \z_{j+1}$ for $1\leq j\leq (k-1)$.
The fact $i_j\neq i_{j+1}$ implies that the sequence $i_1\dots i_{k-1}$ alternates between $1$ and $2$. 
So for 
any $j$, $1\leq j\leq (k-1)$,  $\Pair{z_j}{\z_{{j+1}}}\in\sharingRel_{1}$ or 
$1\leq j\leq (k-1)$,  $\Pair{z_j}{\z_{{j+1}}}\in\sharingRel_{2}$. 
In both cases, by inductive hypotheses (1) and (2) on $\sharingRel_1$ or  $\sharingRel_2$
we have that for all $i$, $1\leq i\leq (k-1)$
\begin{enumerate} [(a)]\addtocounter{enumi}{2}
  \item if $\z_i\neq\resV$ and $\z_{i+1}\neq\resV$, then $\Gamma(\z_i)=\C$ and $\Gamma(\z_{i+1})=\C'$ (for some $\C$ and $\C'$) and $\{\z_i,\z_{i+1}\}\subseteq\FV{\e_h}$ ($h=1$ or $h=2$)
  \item if $\z_{i+1}=\resV$ or $\z_{i}=\resV$, then $\Gamma(\z_i)=\C$ or $\Gamma(\z_{i+1})=\C$ (for some $\C$) and $\z_i\in\FV{e_h}$
  or $\z_{i+1}\in\FV{e_h}$ ($h=1$ or $h=2$)
\end{enumerate} 
By transitivity of equality we have that, for all $i$, $1\leq i\leq k$,
if $\z_i\neq\resV$, then $\Gamma(\z_i)=\C$ and if there is $j$, $1\leq j\leq k$,
such that $\z_j=\resV$ also $\C=\D$. Morever $\{\z_i\ |\ \z_i\neq\resV\ \ 1\leq i\leq k\}\subseteq\FV{\e_1}\cup\FV{\e_2}=\FV{\e}$. 
Therefore the result holds.\\
\PG{\underline{Rule \rn{T-Block}}. In this case $\e=\Block{\DecP{\T_1}{\x_1}{\e_1}\ldots\DecP{\T_n}{\x_n}{\e_n}}{\e_0}$. \\
Let
$\Gamma'=\SubstFun{\Gamma}{\TypeDec{\x_1}{\T_1},\ldots,\TypeDec{\x_n}{\T_n}}$ we have that
 \begin{enumerate} [(a)]
      \item $\sharingRel=\Remove{({\sum\limits_{i=0}^{n}\sharingRel'_i})}{\X}$   where $\X=\dom{\Gamma'}$  
       \item $\sharingRel'_i=\SubstEqRel{\sharingRel_i}{\x_i}{\resV}$ ($1{\leq}i{\leq}n$)
      \item $\TypeCheck{\Gamma'}{\e_i}{\T_i}{\sharingRel_i}$ ($1\leq i\leq n$) 
      \item $\TypeCheck{\Gamma'}{\e_0}{\T}{\sharingRel'_0}$             
      \item if $\T_i{=}\Type{\capsule}{\C_i}$, then  $\Closure{\resV}{\sharingRel_i}=\{\resV\}$ ($1{\leq}i{\leq}n$) 
\end{enumerate}
By induction hypotheses on (c), we have that for all $i$, $1\leq i\leq n$,
if $\Pair{\z}{\z'}\in\sharingRel_i$ and $\z\neq\z'$ 
 \begin{enumerate} [(1)]
       \item if $\z\neq\resV$ and $\z'\neq\resV$, then $\Gamma'(\z)=\C$ and $\Gamma'(\z')=\C'$ (for some $\C$ and $\C'$) and $\{\z,\z'\}\subseteq\FV{\e_i}$
      \item if either $\z=\resV$ or $\z'=\resV$, then $\Gamma'(\z')={\C}$ or $\Gamma'(\z)={\C}$ (for some $\C$) and  $\z'\in\FV{\e_i}$ 
      or $\z\in\FV{\e_i}$
\end{enumerate}
By induction hypotheses on (d), we have that 
if $\Pair{\z}{\z'}\in\sharingRel'_0$ and $\z\neq\z'$ 
 \begin{enumerate} [(1)]\addtocounter{enumi}{2}
       \item if $\z\neq\resV$ and $\z'\neq\resV$, then $\Gamma'(\z)=\C$ and $\Gamma'(\z')=\C'$ (for some $\C$ and $\C'$) and $\{\z,\z'\}\subseteq\FV{\e_0}$
      \item if either $\z=\resV$ or $\z'=\resV$, then $\Gamma'(\z')={\C}$ or $\Gamma'(\z)={\C}$ (for some $\C$) and  $\z'\in\FV{\e_0}$ 
      or $\z\in\FV{\e_0}$
\end{enumerate}
Observe that if for all $i$, $1\leq i\leq n$ if $\z\neq\z'$,  if $\z\neq\z'$ and $\Pair{\z}{\z'}\in\sharingRel'_i$, then $\z,\z'\neq\resV$ and
either $\Pair{\z}{\z'}\in\sharingRel_i$ or $\Pair{\z}{\x_i}\in\sharingRel_i$ and $\Pair{\z'}{\resV}\in\sharingRel_i$.
Therefore by (1) and (2) we have that
\begin{enumerate} [(1)]\addtocounter{enumi}{4}
\item if $\z\neq\z'$ and $\Pair{\z}{\z'}\in\sharingRel'_i$, then $\z,\z'\neq\resV$ and $\Gamma'(\z)=\C$ and
$\Gamma'(\z)=\C'$ (for some $\C$ and $\C'$) and $\{\z,\z'\}\subseteq\FV{\e_i}$
\end{enumerate}
 Let $\x\neq\y$ and $\Pair{\x}{\y}\in\sharingRel$, then ${\Pair{\x}{\y}\in\sum\limits_{i=0}^{n}\sharingRel'_i}$ and
 $\x,\y\not\in\X$.
By \refToProp{lessSrRel}.\ref{p1}, 
there are sequences $i_1\dots i_{k-1}$ \PG{($0\leq i_h\leq n$ for all $h$)} and $\z_1\dots\z_k$ ($k> 1$) such that $\x=\z_1$ and  $\y=\z_k$
 \begin{enumerate} [(A)]\addtocounter{enumi}{1}
\item $\Pair{z_j}{\z_{{j+1}}}\in\sharingRel'_{i_{j}}$ and $i_j\neq i_{j+1}$ and $\z_j\neq \z_{j+1}$ for $1\leq j\leq (k-1)$.
 \end{enumerate}
 By (5) and (3) (the induction hypothesis on $\sharingRel'_{0}$) for all $i$, $1\leq i\leq k-1$ 
 \begin{enumerate} [(A)]
  \item if $\z_i\neq\resV$ and $\z_{i+1}\neq\resV$, then $\Gamma(\z_i)=\C$ and $\Gamma(\z_{i+1})=\C'$ (for some $\C$ and $\C'$) and $\{\z_i,\z_{i+1}\}\subseteq\FV{\e_{i_j}}$   
  \end{enumerate} 
By (5) and (4) (the induction hypothesis on $\sharingRel'_{0}$) for all $i$, $1\leq i\leq k-1$  
\begin{enumerate} [(A)]\addtocounter{enumi}{1}
\item if $\z_{i+1}=\resV$ or $\z_{i}=\resV$, then $\Gamma(\z_i)=\C$ or $\Gamma(\z_{i+1})=\C$ (for some $\C$) and $\z_i\in\FV{e_{i_j}}$
  or $\z_{i+1}\in\FV{e_{i_j}}$. 
\end{enumerate} 
Finally, by transitivity of equality we have that, for all $i$, $1\leq i\leq k$,
if $\z_i\neq\resV$, then $\Gamma(\z_i)=\C$ and if there is $j$, $1\leq j\leq k$,
such that $\z_j=\resV$ also $\C=\D$. Morever $\{\z_i\ |\ \z_i\neq\resV\ \ 1\leq i\leq k\}\subseteq\bigcup_{1\leq i\leq n}(\FV{\e_i}\setminus\X)=\FV{\e}$. 
Therefore, if $\Pair{\x}{\y}\in\sharingRel$ and $\x\neq\y$ we get
\begin{itemize}
\item if $\x\neq\resV$ and $\y\neq\resV$, then $\Gamma(\x)={\C}$ and $\Gamma(\y)=\PG{\C'}$ (for some $\C$ and $\C'$) and $\x,\y\in\FV{\e}$,
\item if $\y=\resV$ or $\x=\resV$, then $\Gamma(\x)={\C}$ or $\Gamma(\y)={\C}$ (for some $\C$) and $\x\in\FV{\e}$ or $\y\in\FV{\e}$.
\end{itemize}
.\\
}
The proofs for \underline{rules \rn{T-Invk} and \rn{T-New}} are similar.
\end{proof}


\section{Examples}\label{sect:examples}

In this section we illustrate the expressiveness of the type system by
programming examples, and we show a type derivation.
\begin{myexample}\label{ex:One}
Assume we have a class \lstinline{D} with a field \lstinline{f} of type \lstinline{D}, and
a class \lstinline{C} with two fields of type {\lstinline{D}}.  Consider the
following closed expression $\ett$:
\begin{lstlisting}
D y= new D(y);
D x= new D(x);
C$^\capsule$ z= {D z2= new D(z2); D z1= (y.f= x); new C(z2,z2)};
z
\end{lstlisting}

The inner block (right-hand side of the declaration of \lstinline{z}{}) refers
to the external variables \lstinline{x} and \lstinline{y}, that is, they occur
free in the block. In particular, the execution of the block has the sharing
effect of connecting \lstinline{x} and \lstinline{y}. However, these variables
will \emph{not} be connected to the final result of the block, since the result
of the assignment will be only connected to a local variable which is not used
to build the final result, as more clearly shown by using the sequence
abbreviation: \lstinline@{D z2= new D(z2); y.f= x; new C(z2,z2)}@.  Indeed, as
will be shown in the next section, the block reduces to 
\lstinline@{D z2= new D(z2); new C(z2,z2)}@ which is a {closed block}.  In 
existing type systems {supporting the capsule notion} this example is either 
ill-typed \cite{GordonEtAl12}, or can be typed by means of a rather tricky 
\emph{swap} typing rule \cite{ServettoZucca15,GianniniEtAl16} which, roughly 
speaking, temporarily changes, in a subterm, the set of variables which can 
be freely used.
\end{myexample}
\vspace{-10pt}                                                                                                                                                                                                                                                                   
\begin{myexample}
As a counterexample, consider the following {ill-typed} term
\begin{lstlisting}
D y= new D(y);
D x= new D(x);
C$^\capsule$ z= {D z1= (y.f= x); new C(z1,z1)};
z
\end{lstlisting}
Here the inner block is not a capsule, since the local variable \lstinline{z1}
is initialized as an \emph{alias} of \lstinline{x}, hence connected to both
\lstinline{x} and \lstinline{y}{}.  Indeed, the block reduces to 
\lstinline@new C(x,x)@ which is not {closed}.
\end{myexample}
\begin{figure*}[ht]
{\small
\begin{center}
\begin{math}
\begin{array}{c}
\deriv_1:\BigSpace{\prooftree
{\prooftree
\begin{array}{l}
\TypeCheck{\SubstFun{\Gamma_1}{\Gamma_2}}{\xtt}{\Dtt}{\{\xtt,\resV\}}\ \\
\TypeCheck{\SubstFun{\Gamma_1}{\Gamma_2}}{\ytt}{\Dtt}{\{\ytt,\resV\}}\ 
\end{array}
\justifies
\TypeCheck{\SubstFun{\Gamma_1}{\Gamma_2}}{\FieldAssign{\ytt}{\ftt}{\xtt}}{\Dtt}{\{\xtt,\ytt,\resV\}}
\endprooftree}
\justifies
\TypeCheck{\SubstFun{\Gamma_1}{\Gamma_2}}{{\ConstrCall{\Dtt}{\FieldAssign{\ytt}{\ftt}{\xtt}}}}{\Dtt}{\{\xtt,\ytt,\resV\}}
\endprooftree}
\\[10ex]
\deriv_2:\Space{\prooftree
\prooftree
\TypeCheck{\SubstFun{\Gamma_1}{\Gamma_2}}{\zD}{\Dtt}{\{\zD,\resV\}}\ 
\justifies
\TypeCheck{\SubstFun{\Gamma_1}{\Gamma_2}}{{\ConstrCall{\Dtt}{\zD}}}{\Dtt}{\{\zD,\resV\}}
\endprooftree
\Space
\deriv_1
\Space
\prooftree
\TypeCheck{\SubstFun{\Gamma_1}{\Gamma_2}}{\zD}{\Dtt}{\{\zD,\resV\}}\ 
\justifies
\TypeCheck{\SubstFun{\Gamma_1}{\Gamma_2}}{{\ConstrCall{\Ctt}{\zD,\zD}}}{\Ctt}{\{\zD,\resV\}}
\endprooftree
\justifies
\TypeCheck{{\Gamma_1}}{\Block{\Dec{\Dtt}{\zD}{\ConstrCall{\Dtt}{\zD}}\,\Dec{\Dtt}{\zU}{\ConstrCall{\Dtt}{\FieldAssign{\ytt}{\ftt}{\xtt}}}}{\ConstrCall{\Ctt}{\zD,\zD}}}{\Ctt}{\{\xtt,\ytt\}}
\endprooftree}
\\[10ex]
\deriv:\Space\prooftree
\prooftree
\TypeCheck{\Gamma_1}{\ytt}{\Dtt}{\{\ytt,\resV\}}\ 
\justifies
\TypeCheck{\Gamma_1}{\ConstrCall{\Dtt}{\ytt}}{\Dtt}{\{\ytt,\resV\}}
\endprooftree
\Space
\prooftree
\TypeCheck{{\Gamma_1}}{\xtt}{\Dtt}{\{\xtt,\resV\}}\ 
\justifies
\TypeCheck{{\Gamma_1}}{{\ConstrCall{\Dtt}{\xtt}}}{\Dtt}{\{\xtt,\resV\}}
\endprooftree
\BigSpace
\deriv_2
\BigSpace
\TypeCheck{\Gamma_1}{\ztt}{\Ctt}{\epsilon}\ 
\justifies
\TypeCheck{}{\Block{\Dec{{\Dtt}}{\ytt}{\ConstrCall{\Dtt}{\ytt}}\,\Dec{{\Dtt}}{\xtt}{\ConstrCall{\Dtt}{\xtt}}\,\Dec{\Type{\capsule}{\Ctt}}{\ztt}{\eU}}{\ztt}}{\Ctt}{\epsilon}
\endprooftree\\ \\
\end{array}
\end{math}
\end{center}
}
\hrulefill
{\small
\begin{itemize}
\item $\deriv_2$ yields $\eUA=\BlockLab{\Dec{\Dtt}{\zD}{\ConstrCall{\Dtt}{\zD}}\,\Dec{\Dtt}{\zU}{\ConstrCall{\Dtt}{\FieldAssign{\ytt}{\ftt}{\xtt}}}}{\ConstrCall{\Ctt}{\zD,\zD}}{\{\zD\}}$
\item $\deriv$ yields 
${\ett'=}\BlockLab{\Dec{{\Dtt}}{\ytt}{\ConstrCall{\Dtt}{\ytt}}\,\Dec{{\Dtt}}{\xtt}{\ConstrCall{\Dtt}{\xtt}}\,\Dec{\Type{\capsule}{\Ctt}}{\ztt}{\eUA}}{\ztt}{\emptyset}$
\end{itemize}
}
\caption{Type derivation for \refToExample{One}}\label{fig:TypingOne}
\end{figure*}

\noindent
\emph{Type derivation for \refToExample{One}.} Let {$\Gamma_1=\TypeDec{\ytt}{\Dtt},\TypeDec{\xtt}{\Dtt},\TypeDec{\ztt}{\Type{\capsule}{\Ctt}}$, and $\Gamma_2=\TypeDec{\zD}{\Dtt},\TypeDec{\zU}{\Dtt}$}.
In \refToFigure{TypingOne} we give the type derivation that shows that
expression $\ett$ {of Example~\ref{ex:One}} is well-typed.  To save space we
omit the annotated expression produced by the derivation, and show the
annotations for the blocks at the bottom of the figure.

Consider the type derivation $\deriv_2$ {for the expression $\eU$ which
initializes $\ztt$ (inner block)}. The effect of the declaration of $\zD$ is
the sharing relation $\{\zD,\resV\}$ and the effect of the declaration of $\zU$
is the sharing relation $\{\zU,\xtt,\ytt,\resV\}$. Before joining these sharing
relations, we remove $\resV$ from their domains, since the results of the two
expressions are not connected with each other.  So the resulting sharing
relation is represented by $\{\zU,\xtt,\ytt\}$ ($\zD$ is only connected with
itself). The effect of the evaluation of the body is the sharing relation
represented by $\{\zD,\resV\}$. Therefore, before removing the local variables
$\zU$ and $\zD$ from the sharing relations 
we have the sharing relation $\{\zU,\xtt,\ytt\}\,\{\zD,\resV\}$.  The block is
annotated with $\{\zD\}$ (the local variable in the equivalence class of
$\resV$).  Since, after removing the local variables,
$\Closure{\resV}{}=\{\resV\}$, $\eU$ denotes a capsule, and may be used to
initialize an affine variable.\\
{$\deriv$ is the type} derivation for the {whole} expression $\ett$. Note that,
since $\ztt$ is an affine variable, the sharing relation of $\ztt$, which is
the body of the block, is the identity. In particular,
$\Closure{\resV}{}=\{\resV\}$, so the annotation of the block $\ett$ is
$\emptyset$. 

\begin{myexample}
We provide now a more realistic programming example, assuming a syntax enriched
by usual programming constructs. See the Conclusion for a discussion on how
to extend the type system to such constructs.

The class \Q@CustomerReader@ below models reading information about customers
out of a text file formatted as shown in the example:
\begin{lstlisting}
Bob
1 500 2 1300
Mark
42 8 99 100
\end{lstlisting}
In {odd} lines we have customer names, in {even}  lines we have a shop history:
a sequence of product codes.  The method \Q@CustomerReader.read@ takes a
\Q@Scanner@, assumed to be a class similar to the one in Java, for reading a
file and extracting different kinds of data.
\vspace{-5pt}
\begin{lstlisting}
class CustomerReader {
  static Customer read(Scanner s)/*$\sharingRel = \epsilon$*/{ 
    Customer c=new Customer(s.nextLine())
    while(s.hasNextNum()){
      c.addShopHistory(s.nextNum())
    }
    return c
  }
}
class Scanner {
  String nextLine()/*$\sharingRel = \epsilon$*/{...}
  boolean hasNextNum()/*$\sharingRel = \epsilon$*/{...}
  int nextNum()/*$\sharingRel = \epsilon$*/{...}
}
\end{lstlisting}
Here and in the following, we insert after method headers, as comments, their
sharing effects.  In a real language, a library should declare sharing effects
of methods by some concrete syntax, as part of the type information available
to clients. In this example, \lstinline{CustomerReader.read}{} uses some
methods of class \lstinline{Scanner}{}. Such methods have no sharing effects,
as represented by the annotation $\sharingRel=\epsilon$.  Note that for the
last two methods this is necessarily the case since they have no explicit
parameters and a primitive return type. For the first method, the only possible
effect could have been to mix $\this$ with the result.

A \Q@Customer@ object is read from the file, and then its shop history is
added.  Since methods invoked on the scanner introduce no sharing, we can
infer that the same holds for method \Q@CustomerReader.read@. In other words,
we can statically ensure that the data of the scanner are not mixed with the
result. 

\label{open-capsule}
The following method \Q@update@ illustrates how we can ``open'' capsules,
modify their values and then recover the original capsule guarantee. The method
takes a customer, which is required to be a capsule by the fact that the
corresponding parameter is affine, and a scanner as before.
\vspace{-5pt}
\begin{lstlisting}
class CustomerReader {...//as before
  static Customer update(Customer$^\capsule$ old,Scanner s)/*$\sharingRel = \epsilon$*/{
    Customer c=old//we open the capsule `old'
    while(s.hasNextNum()){
      c.addShopHistory(s.nextNum())
    }
    return c
  }
}
\end{lstlisting}
A method which has no sharing effects can use the pattern illustrated above:
one (or many) affine parameters are opened  (that is, assigned to local
variables) and, in the end, the result is guaranteed to be a capsule again.
This mechanism is not possible in~\cite{Almeida97,ClarkeWrigstad03,DietlEtAl07}
and relies on destructive reads in~\cite{GordonEtAl12}.

A less restrictive version of method \lstinline{update} could take a non affine
\lstinline{Customer old}{} parameter, that is, not require that the old
customer is a capsule.  In this case, the sharing effects would be $\sharingRel
=\{{\small \resV,\texttt{old} }\}$.  Hence, in a call
\lstinline{Customer.update(c1,s)}{}, the connections of \lstinline{c1}{} would
be propagated to the result of the call. In other words, the method type is
``polymorphic'' with respect to sharing effects.  Notably, the method will
return a capsule if invoked on a parameter which is a capsule.
\end{myexample}
\vspace{-10pt}                                                                                                                                                                                                                                                                   
\begin{myexample}\label{ex:Four}
The following method takes two teams \lstinline{t1}, \lstinline{t2}. Both teams
want to add a reserve player from their respective lists \lstinline{p1} and
\lstinline{p2}, assumed to be sorted with best players first.  However, to keep
the game fair, the two reserve players can only be added if they have the same
skill level.
\vspace{-5pt}
\begin{lstlisting}
static void addPlayer(Team t1, Team t2, Players p1, Players p2)
/*$\sharingRel = \{\texttt{t1},\texttt{p1}\}, \{\texttt{t2},\texttt{p2}\}$*/{
  while(true){//could use recursion instead
    if(p1.isEmpty()||p2.isEmpty()) {/*error*/}
    if(p1.top().skill==p2.top().skill){
      t1.add(p1.top());
      t2.add(p2.top());
      return;
      }
    else{
      removeMoreSkilled(p1,p2);
      }
  }
\end{lstlisting}
\vspace{-5pt}
The sharing effects express the fact that each team is only mixed with its list
of reserve players. 
\end{myexample}
\begin{myexample}\label{ex:Five}
Finally, we provide a more involved example which illustrates the expressive
power of our approach.  Assume we have a class $\Ctt$ as follows:
\vspace{-5pt}
\begin{lstlisting}
class C {
   C f;
   C clone()/*$\sharingRel = \epsilon$*/{...}
   C mix(C x)/*$\sharingRel = \{\xtt,\this,\resV\}$*/{...}
}
\end{lstlisting}
\vspace{-5pt}

The method $\clone$ is expected to return a deep copy of the receiver. Indeed,
this method has no parameters, apart from the implicit non-affine\footnote{{We
assume $\this$ to be non-affine unless explicitly indicated, e.g., by inserting
$\capsule$ as first element of the list of parameters.} } parameter $\this$,
and returns an object of class $\Ctt$ which is \emph{not} connected to the
receiver, as specified by the fact that the sharing relation is the identity,
represented by $\epsilon$, where $\resV$ is not connected to $\this$. Note that
a shallow clone method would be typed
\lstinline{C clone()/*$\sharingRel = \{\this,\resV\}$*/}{}.

The method $\mix$ is expected to return a ``mix'' of the receiver with the
argument. Indeed, this method has, besides $\this$, a parameter $\xtt$ of class
$\Ctt$, both non affine, returns an object of class $\Ctt$ and its effects are
connecting $\xtt$ with $\this$, and both with the result.  
 
Consider now the following closed expression $\ett$:
\vspace{-5pt}
\begin{lstlisting}
C c1= new C(c1);
C outer= {
  C c2= new C(c2);
  C$^\capsule$ inner= {
    C c3= new C(c3);
    C r= c2.mix(c1).clone()
    r.mix(c3)};
  inner.mix(c2)};
outer
\end{lstlisting}
\vspace{-5pt}
The key line in this example is \lstinline{C r=c2.mix(c1).clone()}{}.\\
Thanks to the fact that \lstinline{clone}{} returns a capsule, we know that
\lstinline{r}{} will not be connected to the external variables
\lstinline{c1}{} and \lstinline{c2}{}, hence also the result of the block will
not be connected to \lstinline{c1}{} or \lstinline{c2}{}. However, the sharing
between \lstinline{c2}{} and \lstinline{c1}{} introduced by the
\lstinline{mix}{} call is traced, and prevents the \emph{outer} block {from
being} a capsule. The reader can check that, by replacing the declaration of
\lstinline{r}{} with \lstinline{C r=c2.mix(c2).clone()}{}, also the outer block
is a capsule, hence variable {\lstinline{outer}} could be declared affine.
{This example is particularly challenging to test the capability of a type
system to detect the capsule property.} For instance,  the type system in
\cite{GordonEtAl12} does not discriminate these two cases, those in
\cite{ServettoZucca15,GianniniEtAl16} require rather tricky and non
syntax-directed rules.
The type derivation for \refToExample{Five} can be found in the Appendix.
\end{myexample}


\section{The calculus}\label{sect:calculus}

The calculus has a simplified syntax, defined in \refToFigure{calculus},  where
we assume that, except from right-hand sides of declarations and bodies of
blocks, subterms of a compound expression are only values. This simplification
can be easily obtained by a (type-driven) translation of the syntax of
\refToFigure{syntax} generating for each subterm {which is not a value} a local
declaration of the appropriate type. Moreover we omit primitive types.
\EZ{Finally, the syntax describes runtime terms, where blocks are annotated as described in the previous section.}

\begin{figure}[t]
\begin{grammatica}
\produzione{\e}{\x\,{\mid}\,\FieldAccess{\val}{\f}\,{\mid}\,\MethCall{\val}{\m}{\vs}\,{\mid}\,\FieldAssign{\val}{\f}{\val}\,{\mid}\,\ConstrCall{\C}{\vs}}{expression}\\
\seguitoproduzione{\,{\mid}\,\EZ{\BlockLab{\decs}{\e}{\X}}}{}\\
\produzione{\dec}{\Dec{\T}{\x}{\e}}{declaration}\\ \\
\produzione{\T}{\Type{\mu}{\C}}{declaration type}\\
\produzione{\val}{{\x\mid\BlockLab{\dvs}{{\val}}{\X}\mid\ConstrCall{\C}{\vs}}}{value}\\
\produzione{\dv}{\Dec{\C}{\x}{\ConstrCall{\C}{\vs}}}{evaluated declaration}\\
\end{grammatica}
\caption{Syntax of calculus, values, and evaluated declarations}\label{fig:calculus}
\end{figure}

A {\em value} is the result of the reduction of an expression, and is either a
variable (a reference to an object), or a block where the declarations are
evaluated (hence, correspond to a local store) and the body is in turn a
value\EZ{, or a constructor call where argument are evaluated.}

A sequence $\dvs$ of \emph{evaluated declarations} plays the role of the store
in conventional models of imperative languages, that is, each $\dv$ can be seen
as an association of \EZ{a right-value to a reference.}

As anticipated in \refToSection{language}, mutual recursion among evaluated
declarations is allowed, whereas we do not allow references to variables on the
left-hand side of forward unevaluated declarations. E.g, 
\lstinline{C x= new C(y); C y= new C(x);  x}{} is allowed, whereas 
\lstinline{C x= new C(y); C y= x.f;   x}{} is not. 

That is, our calculus supports \emph{recursive object initialization},
whereas,  e.g., in Java, we cannot directly create two mutually referring
objects as in the allowed example above, but we need to first initialize their
fields to \texttt{null}. However, to make recursive object initialization safe,
we should prevent access to an object which is not fully initialized yet, as in
the non-allowed example. Here we take a simplifying assumption, just requiring
initialization expressions to be values. More permissive assumptions can be
expressed by a sophisticated type system as in \cite{ServettoEtAl13}. 

Semantics is defined by a \emph{congruence} relation, which captures structural
equivalence, and a \emph{reduction} relation, which models actual computation,
similarly to what happens, e.g., in $\pi$-calculus \cite{Milner99}.

The congruence relation, denoted by $\congruence{}{}$, is defined as the
smallest congruence satisfying the axioms in \refToFigure{congruence}.  We
write $\FV{\decs}$ and $\FV{\e}$ for the free variables of a sequence of
declarations and of an expression, respectively, and $\Subst{\X}{\y}{\x}$,
$\Subst{\decs}{\y}{\x}$, and $\Subst{\e}{\y}{\x}$ for the capture-avoiding
variable substitution on a set of variables, a sequence of declarations, and an
expression, respectively, all defined in the standard way. 

\begin{figure}[!htbp]
{
\begin{math}
\begin{array}{l}
{\NamedRuleOL{alpha}{{\congruence{\BlockLab{\decs\ \Dec{\T}{\x}{\e}\ \decs'}{\e'}{\X}}{\BlockLab{\Subst{\decs}{\y}{\x}\ \Dec{\T}{\y}{\Subst{\e}{\y}{\x}}\ \Subst{\decs'}{\y}{\x}}{\Subst{\e'}{\y}{\x}}{{\Subst{\X}{\y}{\x}}}}}}{}}
\\[3ex]
\NamedRuleOL{reorder}{\congruence{
\BlockLab{
\decs\ \Dec{\C}{\x}{{\ConstrCall{\C}{\EZ{\vs}}}}\ \decs'
}{\e}{\X}}{
\BlockLab{
\Dec{\C}{\x}{{\ConstrCall{\C}{\EZ{\vs}}}}\ \decs\ \decs'\ }{\e}{\X}}}{}
\\[3ex]
\NamedRuleOL{new}{\congruence{
\ConstrCall{\C}{\EZ{vs}}
}{
\BlockLab{\Dec{\C}{\x}{\ConstrCall{\C}{\EZ{\vs}}}}{\x}{{\{\x\}}}}
}{}
\\[3ex]
\NamedRuleOL{{block-elim}}{\congruence{{\BlockLab{}{\e}{\emptyset}}}{\e}}{}
\\[2ex]
\NamedRuleOL{{{dec}}}
{
\PG{
\begin{array}{l}
\BlockLab{\dvs\ \Dec{\Type{\mu}{\C}}{\x}{\BlockLab{\dvs_1\ \decs_2}{\e}{\X}}\ \decs'}{\e'}{\Y}\cong \\
\BigSpace\BlockLab{\dvs\ \dvs_1\ \Dec{\Type{\mu}{\C}}{\x}{\BlockLab{\decs_2}{\e}{\X'}}\ \decs'}{\e'}{{\Y'}}
\end{array}
}
}
{
\!\!\!\!\!\!\!\!\!\!\begin{array}{l}
\FV{\dvs_1}\cap\dom{\decs_2}=\emptyset\\
\FV{\dvs\,\decs'\,{\e'}}{\cap}\dom{\dvs_1}{=}\emptyset\\ 
\mu=\capsule{\implies}\dom{\dvs_1}{\cap}\X{=}\emptyset
\end{array}}
\\[4ex]
\NamedRuleOL{{body}}
{\begin{array}{l}
{\BlockLab{\EZ{\dvs}}{\BlockLab{\EZ{\dvs_1}\ \decs_2}{\e}{{\X}}}{{\Y}}}\cong\\
\BigSpace{\BlockLab{\EZ{\dvs}\ \EZ{\dvs_1}}{
\BlockLab{\decs_2}{\e}{{\X'}}}{{\Y'}}}
\end{array}
}
{
\begin{array}{l}
\FV{\EZ{\dvs_1}}\cap\dom{\decs_2}=\emptyset\\
\FV{\EZ{\dvs}}\cap\dom{\EZ{\dvs_1}}=\emptyset
\end{array}}
\\[4ex]
\NamedRuleOL{val-ctx}{
\begin{array}{l}
{\ValCtx{\BlockLab{{\dvs_1\,\dvs_2}}{\val}{\X}}}\cong\\
\BigSpace{\BlockLab{{\dvs_1}}{\ValCtx{\BlockLab{{\dvs_2}}{\val}{\X'}}}{\Y}}
\end{array}}
{
\begin{array}{l}
\FV{{\dvs_1}}\cap\dom{{\dvs_2}}=\emptyset\\
\FV{\valCtx}\cap\dom{{\dvs_1}}=\emptyset
\end{array}}
\end{array}
\end{math}
}
\caption{Congruence rules}
\label{fig:congruence}
\end{figure}

Rule \rn{alpha} is the usual $\alpha$-conversion. The condition
$\x,\y\not\in\dom{\decs\,\decs'}$ is implicit by well-formedness of blocks. \\
Rule \rn{reorder} states that we can move evaluated declarations in an
arbitrary order.  Note that, instead, $\decs$ and $\decs'$ cannot be swapped,
because this could change the order of side effects. \\
In rule \rn{new}, a constructor invocation can be seen as an elementary block
where a new object is allocated.\\
Rule \rn{{block-elim}} states  that a block with no
declarations is equivalent to its body.\\
With the remaining rules we can move a sequence of declarations from a block to
the directly enclosing block, or conversely, as it happens with rules for
\emph{scope extension} in the $\pi$-calculus \cite{Milner99}.\\ 
In rules \rn{dec} and \rn{body}, the inner block is the body, or the right-hand
side of a declaration, respectively, of the enclosing block.} The first two
side conditions ensure that moving the declarations outside the block does  cause
neither scope extrusion nor capture of free variables. More precisely: the first
prevents moving outside declarations which depend on local
variables of the inner block. The second prevents capturing free
variables of the enclosing block.  Note that the second condition can be
obtained by $\alpha$-conversion of the inner block, but the first cannot.
Finally, the third side condition of rule \rn{{dec}} prevents, in case the
block {initializes} an affine variable, to move outside declarations of
variables that {will be possibly} connected to the result of the block.
Indeed, in this case we would get an ill-typed term. In case of a non-affine
declaration, instead, this is not a problem.\\ 
Rule \rn{val-ctx} handles the cases when the inner block is a  subterm of a
field access, method invocation, field assignment or constructor invocation.
Note that in this case the inner block is necessarily a (block) value. To
express all such cases in a compact way, we define {\em value contexts}
$\valCtx$ in the following way:
\begin{center}
$\valCtx::=\emptyctx\mid\FieldAccess{\valCtx}{\f}\mid\FieldAssign{\valCtx}{\f}{\val}\mid\FieldAssign{\val}{\f}{\valCtx}\mid\ConstrCall{\C}{\vs,\valCtx,\vs'}$
\end{center}
For instance, if $\valCtx=\ConstrCall{\C}{\vs,\emptyctx,\vs'}$, we get
{\small \begin{center}
${\ConstrCall{\C}{\vs,\BlockLab{{\dvs_1\,\dvs_2}}{\val}{\X},\vs'}\cong 
\BlockLab{\dvs_1}{{\ConstrCall{\C}{\vs,\BlockLab{\dvs_2}{\val}{\X'},\vs'}}}{\Y}}$
\end{center}}
\vspace{-5pt}
As for rules \rn{dec} and \rn{body}, the first side condition prevents moving
outside a declaration in $\dvs_1$ which depends on local variables of the inner
block, and the second side condition prevents capturing free variables of
$\valCtx$\EZ{, defined in the standard way.}

The following definition introduces a simplified syntactical form for values
and evaluated declarations. \EZ{In this canonical form, a sequence of evaluated declarations (recall that its order is immaterial) can be seen as a \emph{store} which associates to references \emph{object states} of shape $\ConstrCall{\C}{\xs}$, where fields contain in turn references, and a value is a
variable (a reference to an object) possibly enclosed in a local store. }

\begin{mydefinition}\label{def:canonicalVal}\
\begin{enumerate}
  \item A sequence of evaluated declarations $\dvs$ is \emph{in canonical form} 
    if, for all $\dv$ in $\dvs$, $\dv=\Dec{\C}{\x}{\ConstrCall{\C}{\xs}}$ for 
    some $\C$, $\x$ and $\xs$.
  \item A value $\val$ is \emph{in canonical form} if either $\val=\x$ for some 
    $\x$, or $\val=\BlockLab{\dvs}{\x}{\X}$ for some $\X$, $\dvs$, and $\x$, 
    with $\dvs\neq\epsilon$ in canonical form.
\end{enumerate}
\end{mydefinition}
The following proposition shows that congruence allows us to assume that values
are in canonical form.
\begin{myproposition}\label{prop:value}
If $\val$ is a value, then there exists $\val'$ such that $\congruence{\val}{\val'}$ and $\val'$ is in 
canonical form.
\end{myproposition}
\begin{proof}
By structural induction on values, and for blocks by induction on the number of 
declarations that are not in canonical form. The full proof is in 
\ref{app:proofs}.
\end{proof}

\EZ{From now on}, unless otherwise stated, we assume that values and evaluated
declarations are in canonical form.
\EZ{Moreover, we also need to characterize values which are garbage-free, in the sense that they do not contain useless store. 
To this end, we first inductively define $\connected{\decs}{\X}{\x}$, meaning that $\x$ is (transitively)
used by $\X$ through $\decs=\Dec{\T_1}{\x_1}{\e_1}\ldots\Dec{\T_n}{\x_n}{\e_n}$, by:
\begin{quote}
$\connected{\decs}{\X}{\x}$ if $\x\in\X$\\
$\connected{\decs}{\X}{\x}$ if $\x\in\FV{\e_i}$, for some $i\in 1..n$, and $\connected{\decs}{\X}{\x_i}$.
\end{quote}
Then, we write $\Reduct{\decs}{\X}$ for the subsequence of $\decs$ (transitively) used by $\X$, defined by: for all $i\in 1..n$,
$\Dec{\T_i}{\x_i}{\e_i}\in\Reduct{\decs}{\X}$ if
$\connected{\decs}{\X}{\x_i}$.\\
Finally, we define }
$\remGarbage(\BlockLab{\dvs}{\x}{\X})=\BlockLab{\Reduct{\dvs}{\x}}{\x}{{\X}\cap{\dom{\Reduct{\dvs}{\x}}}}$, and we say that a value $\val$ is \emph{garbage-free}} if either $\val=\x$ or $\val=\remGarbage(\val)$.}\\

{\em Evaluation contexts}, defined below, express standard left-to-right 
evaluation.
\begin{center}
$\ctx::=\emptyctx\mid\BlockLab{\dvs\ \Dec{\T}{\x}{\ctx}\ \decs}{\e}{\X}\mid \BlockLab{\dvs}{\ctx}{\X}$
\end{center}
In the evaluation context $\BlockLab{\dvs\ \Dec{\T}{\x}{\ctx}\ \decs}{\e}{\X}$
we assume that no declaration in $\decs$ is evaluated. \PG{This can always be
achieved by the congruence rule \rn{reorder}.}

We introduce now some notations which will be used in reduction rules.  We
write $\dvs(\x)$ for {\em the declaration of $\x$ in $\dvs$}, if any (recall
that in well-formed blocks there are no multiple declarations for the same
variable).  We write $\HB{\ctx}$ for the \emph{hole binders} of $\ctx$, that
is, the variables declared in blocks enclosing the context hole, defined by: \\
\indent$\bullet$  if $\ctx=\Block{\dvs\ \Dec{\T}{\x}{\ctxP}\ \decs}{\e}$, 
then $\HB{\ctx}=\dom{\dvs}\cup\{\x\}\cup\HB{\ctxP}\cup\dom{\decs}$\\
\indent $\bullet$ if $\ctx=\Block{\dvs}{\ctxP}$, 
then $\HB{\ctx}=\dom{\dvs}\cup\HB{\ctxP}$\\
We write $\decCtx{\ctx}{\x}$ and $\extractDec{\ctx}{\x}$ for the {\em
sub-context {declaring} $\x$} and the {\em evaluated declaration of $x$} extracted from
$\ctx$, defined as follows:\\
\indent$\bullet$ let $\ctx=\Block{\dvs\ \Dec{\T}{\y}{\ctxP}\ \decs}{\e}$\\
\indent\indent- if {$\dvs(\x)=\dv$ and $\x\not\in\HB{\ctxP}$}, then 
$\decCtx{\ctx}{\x}=\Block{\dvs\ \Dec{\T}{\y}{\emptyctx}\ \decs}{\e}$ and\\
\indent\indent\ \ $\extractDec{\ctx}{\x}={\dv}$\\
\indent\indent-  else 
  $\decCtx{\ctx}{\x}=\Block{\dvs\ \Dec{\T}{\y}{\decCtx{\ctxP}{\x}}\ \decs}{\e}$
  and $\extractDec{\ctx}{\x}=\extractDec{\ctxP}{\x}$\\
\indent$\bullet$  let $\ctx=\Block{\dvs}{\ctxP}$\\
\indent\indent-  if {$\dvs(\x)=\dv$ and $\x\not\in\HB{\ctxP}$}, then
$\decCtx{\ctx}{\x}=\Block{\dvs}{\emptyctx}$ and
  $\extractDec{\ctx}{\x}={\dv}$\\
\indent\indent-  else $\decCtx{\ctx}{\x}=\Block{\dvs}{\decCtx{\ctxP}{\x}}$, and
 $\extractDec{\ctx}{\x}=\extractDec{\ctxP}{\x}$.\\
Note that $\decCtx{\ctx}{\x}$ and $\extractDec{\ctx}{\x}$ are not defined if \EZ{there is no evaluated declaration for
$\x$} in some block enclosing the context hole.
 
Reduction rules are given in \refToFigure{reduction}. 
\begin{figure}[!htbp]
\begin{small}
\begin{math}
\begin{array}{l}
\PG{\NamedRule{congr}{\reduce{\e'}{\e''}
}{
  \reduce{{\MS{\e}}}{\e''}
}{  \begin{array}{l}
\congruence{\e}{\e'}
\end{array}
}}
\\[4ex]
{\NamedRuleOL{field-access}{\reduce{\Ctx{\FieldAccess{\x}{\f_i}}}{\Ctx{\x_i}}}{
\begin{array}{l}
\extractDec{\ctx}{\x}=\Dec{\C}{\x}{\ConstrCall{\C}{\x_1,\ldots,\x_n}}\\
\fields{\C}=\Field{\C_1}{\f_1}\ldots\Field{\C_n}{\f_n}\\
i\in 1..n\\
{\ctx=\DecCtx{\ctx}{\x}{\ctxP}\ \wedge\ {\x_i\not\in\HB{\ctxP}}}
\end{array}
}}
\\[4ex]
{\NamedRuleOL{invk}{
\begin{array}{l}
{\Ctx{\MethCall{\val}{\m}{\val_1,..,\val_n}}}\longrightarrow\\
{\Ctx{\Block{\Dec{\Type{\mu}{\C}}{\this}{\val}\, \Dec{\T_1}{\x_1}{\val_1}..\Dec{\T_n}{\x_n}{\val_n}}{\e}}}
\end{array}
}{
\begin{footnotesize}
\begin{array}{l}
\!\!\class{\ctx}{\val}=\C\\
\!\!\method{\C}{\m}{=}\FourTuple{\D}{\mu}{\Param{\T_1}{\x_1}..\Param{\T_n}{\x_n}}{\e}
\end{array}
\end{footnotesize}
}}
\\[4ex]
{\NamedRuleOL{field-assign}{
\reduce{\Ctx{\FieldAssign{\x}{\f_i}{\y}}}{{\UpdateCtx{\y}{\x}{i}{\y}}}
}{
\begin{array}{l}
\extractDec{\ctx}{\x}=\Dec{\C}{\x}{\ConstrCall{\C}{\x_1,\ldots,\x_n}}\\
\fields{\C}=\Field{\C_1}{\f_1}\ldots\Field{\C_n}{\f_n}\\
i\in 1..n\\
{\ctx=\DecCtx{\ctx}{\x}{\ctxP}\ \wedge\ {\y\not\in\HB{\ctxP}}}\\
\end{array}}}
\\[6ex]
{\NamedRuleOL{alias-elim}{\reduce{\Ctx{\BlockLab{\dvs\ \Dec{\C}{\x}{\y}\ \decs }{\e}{\X}}}{\Ctx{{\BlockLab{\dvs\ \Subst{\decs}{\y}{\x}}{\Subst{\e}{\y}{\x}}{\X\setminus\{\x\}}}}}}{}}
\\[4ex]
{\NamedRuleOL{affine-elim}{\reduce{\Ctx{\BlockLab{\dvs\ \Dec{\Type{\capsule}{\C}}{\x}{\val}\ \decs }{\e}{X}}}{\Ctx{{\BlockLab{\dvs\ \SubstVal\decs{\val'}{\x}}{\SubstVal\e{\val'}{\x}}{\PG{\X}}}}}}{\!\!\!\!\!\!\PG{\val'=\remGarbage(\val)}}}
\end{array}
\end{math}
\end{small}
\caption{{Reduction rules}}
\label{fig:reduction}
\end{figure}

Rule \rn{congr}  can be used  to reduce a term which
otherwise would be stuck, as it happens for the $\alpha$-rule in lambda calculus.

In rule \rn{field-access}, given a field access of shape
$\FieldAccess{\x}{\f}$, the first enclosing declaration for $\x$ is found
(through the auxiliary function \textsf{dec}).  The fields of the class $\C$ of
$\x$ are retrieved from the class table.  If $\f$ is actually the name of a
field of $\C$, say, the $i$-th, then the field access is reduced to the
reference $\x_i$ stored in this field.  In the last side condition,
$\decCtx{\ctx}{\x}$ is the (necessarily defined) sub-context containing the
first enclosing declaration for $\x$, and the condition $\x_i\not\in\HB{\ctxP}$
ensures that there are no declarations for $\x_i$ in inner blocks (otherwise
$\x_i$ would be erroneously bound).  This can always be obtained by rule
\rn{alpha} of \refToFigure{congruence}.

For instance, assuming a class table where class \lstinline{A} has an
\lstinline{int} field, and class \lstinline{B} has an \lstinline{A} field \lstinline{f},
without this side condition, the term (without annotations):
\begin{lstlisting}
A a= new A(0); B b= new B(a); {A a= new A(1); b.f}
\end{lstlisting}
{would reduce to}
\begin{lstlisting}
A a= new A(0); B b= new B(a); {A a= new A(1); a}
\end{lstlisting}
{whereas this reduction is forbidden, and by rule \rn{alpha} the term is 
instead reduced to}
\begin{lstlisting}
A a= new A(0); B b= new B(a); {A a1= new A(1); a}
\end{lstlisting}
For this example:
$\decCtx{\ctx}{\texttt{b}}=\Block{\texttt{A a = new A(0); B b = new B(a)}}{\ctxP}$ and
$\ctxP=\Block{\texttt{A a = new A(1)}}{\emptyctx}$.\\
In rule \rn{{invk}}, the class $\C$ of the receiver $\val$ is found through the
auxiliary function \textit{class} defined by
\begin{quote}
$\class{\ctx}{\x}=\C$ if $\extractDec{\ctx}{\x}=\Dec{\Type{}{\C}}{\x}{\_}$\\
$\class{\ctx}{\BlockLab{\dvs}{\x}{\X}}=\C$ if $\dvs(\x)=\Dec{\Type{}{\C}}{\x}{\_}$
\end{quote}
and method $\m$ of $\C$, if any, is retrieved from the class table. The call is
reduced to a block where declarations of the appropriate type for $\this$ and the
parameters are initialized with the receiver and the arguments, respectively, and the body is the method body. 

In rule \rn{field-assign}, given a field assignment of shape
$\FieldAssign{\x}{\f}{\y}$,  the first enclosing declaration for $\x$ is found
(through the auxiliary function \textsf{dec}).  If $\f$ is actually the name of
a field of $\C$, say, the $i$-th, then this first enclosing declaration is
updated, by replacing the $i$-th constructor argument by $\y$, obtaining
$\Dec{\C}{\x}{\ConstrCall{\C}{\x_1,\x_{i-1},\y,\x_{i+1},\ldots,\x_n}}$, as
expressed by the notation $\updateCtx{\y}{\x}{i}$ (the obvious formal
definition of which is omitted).  As for rule \rn{field-access} we have the side
condition that  $\y\not\in\HB{\ctxP}$. This side condition, requiring that there
are no inner declarations for the reference $\y$, prevents scope extrusion,
since if $\y\in\HB{\ctxP}$, $\updateCtx{\valPrime}{\x}{i}$ would take $\y$
outside the scope of its definition.  {The congruence {rules \rn{dec} and
\rn{body}}  of \refToFigure{congruence} can be used to {\em correctly} move the
declaration of $\y$ outside its declaration block, as previously described.}
For example, without the side condition, the term (without annotations)
\begin{lstlisting}
A a= new A(0); B b= new B(a); {A a1= new A(1); b.f=a1}
\end{lstlisting}
would reduce to 
\begin{lstlisting}
A a= new A(0); B b= new B(a1); {A a1= new A(1); a1}
\end{lstlisting}
The previous term is congruent to
\begin{lstlisting}
A a= new A(0); B b= new B(a); A a1= new A(1); b.f=a1
\end{lstlisting}
by applying rule \rn{body}, and then \rn{{block-elim}}. This term reduces 
correctly to
\begin{lstlisting}
A a= new A(0); B b= new B(a1); A a1= new A(1); a1
\end{lstlisting}
The last two rules eliminate evaluated declarations from a block.\\
In rule \rn{alias-elim}, a (non-affine) variable $\x$ which is
initialized as an alias of another reference $\y$ is eliminated by replacing
all its occurrences. In rule \rn{affine-elim}, an affine variable is eliminated
by replacing its unique occurrence with the value associated to its
declaration \PG{from which we remove garbage. } \EZComm{explain why?}  
%

We conclude this section by briefly discussing how the reduction relation \EZ{could} actually be computed. Indeed, the definition in \refToFigure{reduction} is not fully algorithmic, since rule \rn{congr} can always be applied in a non-deterministic way, analogously to what happens, e.g., with $\alpha$-conversion in lambda calculus or structural rules in $\pi$-calculus. 
\EZ{However, a}gain as is usually done for analogous cases, congruence is expected to be applied only when needed (since otherwise reduction would be stuck). All our congruence rules except for  \rn{alpha} and \rn{reorder} are meant to be applied 
from left to right (as a reduction). This is witnessed  by the fact that values in canonical form do not match the left-hand side
of any of the previously mentioned congruence rules and \refToProp{value} ensures that any value may be reduced in this form.


\section{Results}\label{sect:results}

In this section we present the main formal results on our calculus. First, we
show a canonical form theorem describing constraints on free variables of
well-typed \EZ{garbage-free} values. Then we prove subject reduction, stating that reduction
preserves the type, and may reduce the sharing effects. In addition, reduction
preserves an invariant on the store that allows us to prove that lent and
capsule references have the expected behaviour. 

Finally, we prove progress, i.e., that well-typed expressions do not get
``stuck''. 

First of all we extend the typing judgment to annotated expressions, and to
(annotated) sequences of declarations, as follows. 

\begin{definition}\label{def:typeBlock}\ 
\begin{itemize}
\item \EZ{Given an annotated expression $\e$, $\erase{\e}$ is the expression
obtained by erasing the annotations from $\e$, and} $\TypeCheck{\Gamma}{\e}{\C}{\sharingRel}$ if 
    $\TypeCheckAnnotate{\Gamma}{\erase{\e}}{\C}{\sharingRel}{\e}$.
  \item \PG{Given the annotated expression $\e_1$ and $\e_2$, we say that {\em $\e_1$ and $\e_2$ are equal up to annotations}, dubbed $\e_1\variant\e_2$ if $\erase{\e_1}=\erase{\e_2}$.}
      \item Given
  $\decs=\Dec{\Type{\mu_1}{\C_1}}{\x_1}{\e_1}\ldots\Dec{\Type{\mu_n}{\C_n}}{\x_n}{\e_n}$, 
    an (annotated) sequence of declarations, $\TypeCheckDecs{\Gamma}{\decs}{\sharingRel}$ if
    \begin{itemize}
      \item $\TypeCheck{\Gamma}{\e_i}{\C_i}{\sharingRel_i}$, for some 
        $\sharingRel_i$  ($1\leq i\leq n$),
      \item \PG{$\sharingRel=\sum\limits_{i=1}^{n}(\SubstEqRel{\sharingRel_i}{\x_i}{\resV})$}, and 
      \item if $\mu_i{=}\capsule$ then $\IsCapsule{\sharingRel_i}$ ($1\leq i\leq n$).
    \end{itemize}
\end{itemize}
\end{definition}
Note that in the $\sharingRel$ derived for a sequence of declarations the equivalence class 
of $\resV$ is a singleton, according to the fact that a sequence of declarations has no ``result''.
\PG{As we can see from the typing rules of \refToFigure{typing} if 
the non-annotated expression $\e$ is typable, then there is a unique annotated $\e'$ such that
$\TypeCheck{\Gamma}{\e'}{\C}{\sharingRel}$ for some $\C$ and $\sharingRel$.}

\paragraph{Canonical form theorem}
In this subsection \EZComm{not really a subsection} we state a theorem describing constraints on free variables
of well-typed \EZ{\nogarbage\ values\footnote{Recall that values are assumed to be in canonical
form.}. Notably, such free variables are  either affine or connected to its result.} Therefore, a \nogarbage\ capsule
value can contain only affine free variables. 

In the following we use the underscore $\_$ for a type, when the specific type
is irrelevant. Moreover, we we will say that \emph{$\x$ is affine/non-affine}
in $\Gamma$ if $\Gamma(\y)=\Type{\capsule}{\C}$ or $\Gamma(\y)=\Type{}{\C}$,
respectively. 
\begin{theorem}\label{theo:freevars}
\PG{If $\TypeCheck{\Gamma}{\val}{\C}{\sharingRel}$ where $\val$ is \nogarbage\ and $\y\in\FV{\val}$, then:}
\begin{enumerate}
  \item if {$\y$ is non-affine in $\Gamma$}, then $\Pair{\y}{\resV}\in\sharingRel$
  \item if $\IsCapsule{\sharingRel}$, then {$\y$ is affine in $\Gamma$}.
\end{enumerate}
\end{theorem}
\begin{proof}
\PG{1. By cases on the shape of canonical values. 
 \begin{itemize}
    \item If \underline{$\val=\x$}, then the only free variable in $\val$ 
      is $\x$ itself. Since $\x$ is non-affine in $\Gamma$, then the judgment  $\TypeCheck{\Gamma}{\x}{\C}{\sharingRel}$ 
     is derived by rule \rn{t-var}, hence $\Pair{\x}{\resV}\in\sharingRel$. 
    \item If \underline{$\val=\BlockLab{\dvs}{\x}{\X}$}, since $\dvs$ is in 
      canonical form and \nogarbage\ we have 
      $\dvs=\DecP{\C_1}{\x_1}{\ConstrCall{\C_1}{\xs_1}}\ldots\DecP{\C_n}{\x_n}{\ConstrCall{\C_n}{\xs_n}}$
      and $\connected{\dvs}{\x}{\x_i}$ for all $i\in 1..n$.  The judgment 
      $\TypeCheck{\Gamma}{\val}{\_}{\sharingRel}$ is derived 
      by rule \rn{t-block} with premises derived by \rn{t-new} for each declaration (which in turn 
      have premises that are derived with rules \rn{t-var} or \rn{t-affine-var})
      and rule \rn{t-var} to derive a type for the body ($\x=\x_i$ for some $1{\leq}i{\leq}n$). 
      Hence, letting $\Gamma'=\TypeDec{\x_1}{\Type{}{\C_1}},\ldots,\TypeDec{\x_n}{\Type{}{\C_n}}$, we have
\begin{itemize}
\item $\TypeCheck{\SubstFun{\Gamma}{\Gamma'}}{\ConstrCall{\C_i}{\MS{\xs_i}}}{\C_i}{\{\xs'_i,\resV\}}$ $(1\leq i\leq n)$ where $\xs'_i$ are the non affine variables in $\xs_i$
\item $\TypeCheck{\SubstFun{\Gamma}{\Gamma'}}{\x}{\C}{\{\x,\resV\}}$
\item $\sharingRel_i=\{\xs'_i,\x_i\}$
\item $\sharingRel=\Remove{\sharingRel'}{\dom{\Gamma'}}$, with $\sharingRel'=\Sum{\sum\limits_{i=1}^{n}\sharingRel_i}{\{\x,\resV\}}$
\end{itemize}
From \refToProp{lessSrRel}.\ref{p1} and $\connected{\dvs}{\x}{\x_i}$ for all $i\in 1..n$, we have 
that $\sharingRel'$ has a unique equivalence class
$\bigcup_{1\leq i\leq n}\{\xs'_i,\x_i\}\cup\{\resV\}$.
If $\y\in\FV{e}$, then $\y\in\xs_j$ for some $j\in\{1,\dots,n\}$ and $\y\not\in\{\x_1,\dots,\x_n\}$. From
the fact that $\y$ is not affine we have that $\y\in\xs'_j$. Therefore $\Pair{\y}{\resV}\in\sharingRel$.
 \end{itemize}
 }
2. If $\y$ is free in $\val$ and $\y$ is non-affine in $\Gamma$, then by the previous point we would have
$\Pair{\y}{\resV}\in\sharingRel$, contradicting  $\IsCapsule{\sharingRel}$. Hence, 
 $\Gamma(\y)=\Type{\capsule}{\C}$.
\end{proof}

\EZComm{moved here and changed the proof}
\EZ{The following lemma is a corollary of the canonical form theorem.}

\begin{lemma}\label{lemma:sharingCapsule}
If $\TypeCheck{\Gamma}{\val}{\C}{\sharingRel}$ where $\val$ is \nogarbage, and
$\IsCapsule{\sharingRel}$, then $\sharingRel=\epsilon$.
\end{lemma}
\begin{proof}
\EZ{Assume $\sharingRel\neq\epsilon$. Then, from \refToProp{invTyping1}, there would be a free variable in $\val$ non-affine in $\Gamma$, but this is impossible by \refToTheorem{freevars}.2. }
\end{proof}

\paragraph{Subject reduction} To show subject reduction, we need some preliminary lemmas. 

The following lemma states that typing essentially depends only on the free
variables of the expression. We denote by $\Remove{\Gamma}{\x}$ the type environment obtained by removing the
type association for $\x$ from $\Gamma$, if any.
\begin{lemma}{\rm (Weakening)}\label{lemma:weakening}
Let $\TypeCheck{\Gamma}{\e}{\C}{\sharingRel}$. If $\x\not\in\FV{\e}$, then 
\begin{enumerate}
\item $\TypeCheck{\SubstFun{\Gamma}{\x{:}\T}}{\e}{\C}{\Sum{\sharingRel}{\{\x\}}}$ for all $\T$, and
\item $\TypeCheck{\Remove{\Gamma}{\x}}{\e}{\C}{\Remove{\sharingRel}{\x}}$.
\end{enumerate}
\end{lemma}
\begin{proof}
By induction on derivations.  
\end{proof}
%
The following lemma states the dependency between the type and sharing relation derived for
a block \EZ{and} the ones derived for its declarations and body. 
\begin{lemma}{\rm (Inversion for blocks)}\label{lemma:invBlock} 
If $\TypeCheck{\Gamma}{\BlockLab{\decs}{\e}{\X}}{\C}{\sharingRel}$, then 
\begin{itemize}
  \item $\TypeCheckDecs{\Gamma}{\decs}{\sharingRel_{\decs}}$ for some $\sharingRel_{\decs}$ 
  \item $\TypeCheck{\Gamma}{{\e}}{\C}{\sharingRel_{\e}}$ for some $\sharingRel_{\e}$ such that
  \item $\sharingRel=\Remove{(\Sum{\sharingRel_{\decs}}{\sharingRel_{\e}})}{\dom{\decs}}$ and $\X=\Closure{\resV}{(\Sum{\sharingRel_{\decs}}{\sharingRel_{\e}})}\cap\dom{\decs}$.
\end{itemize}
\end{lemma}
\begin{proof}
By rule \rn{T-Block} and definition of the type judgement for declarations. 
\end{proof}

The following lemma asserts that congruent expressions have the same type and
sharing effects. Regarding annotations, which are uniquely determined by the
type derivation, if one of the two expressions is well-typed, hence its annotations are those derived
from its typing, then the annotations of the other are also uniquely determined.
\begin{lemma}{\rm (Congruence preserves types)}\label{lemma:congruence}
Let $\e_1$ and $\e_2$ be annotated expressions. 
If $\TypeCheck{\Gamma}{\e_1}{\C}{\sharingRel}$
and $\congruence{\e_1}{\e_2}$, then $\TypeCheck{\Gamma}{\e'_2}{\C}{\sharingRel}$ for some $\e'_2$ such that $\e_2'\variant\e_2$.\end{lemma}
\begin{proof}
The proof is in~\ref{app:proofs}.  
\end{proof}
\PG{In the following when
we have $\TypeCheck{\Gamma}{\e_1}{\C}{\sharingRel}$ and $\congruence{\e_1}{\e_2}$ we assume  
$\TypeCheck{\Gamma}{\e_2}{\C}{\sharingRel}$, that is, we picked the term with the right annotations.}

\bigskip
The {\em type environment extracted from $\decs$}, denoted $\TypeEnv{\decs}$, 
is defined by: 
\begin{center}
$\Gamma_{\decs}=\TypeDec{\x_1}{\T_1},\ldots,\TypeDec{\x_n}{\T_n}$ if $\decs=\Dec{\T_1}{\x_1}{\e_1}\ldots\Dec{\T_n}{\x_n}{\e_n}$.
\end{center}
Given an evaluation context $\ctx$, the {\em type 
environment extracted from $\ctx$}, denoted $\Gamma_{\ctx}$, is defined by:
\begin{itemize}
  \item $\Gamma_{\emptyctx}$ is the empty type environment,
  \item $\Gamma_{\Block{\dvs\,\Dec{\T}{\x}{\ctx}\ \decs}{\e}}=\SubstFun{(\Gamma_{\dvs\, \decs})[\TypeDec{\x}{\T}]}{\Gamma_{\ctx}}$ and
  \item $\Gamma_{{\Block{\dvs}{\ctx}}}=\SubstFun{\Gamma_{\dvs}}{\Gamma_{\ctx}}$.
\end{itemize}
The following lemma asserts that subexpressions of typable expressions are
themselves typable, and may be replaced with expressions that have the same
type and the same or possibly less sharing effects. \PG{Annotations may change
by effect of the reduced sharing effects since the equivalence class of
$\resV$ in the reduced sharing relations may contain less variables.}

\begin{lemma}{\rm (Context)}\label{lemma:context}
Let $\TypeCheck{\Gamma}{\Ctx{\e}}{\C}{\sharingRel}$, then
\begin{enumerate}
  \item $\TypeCheck{\Gamma[\Gamma_{\ctx}]}{{\e}}{\D}{\sharingRel_1}$ for some 
    $\D$ and $\sharingRel_1$,
  \item if $\TypeCheck{\Gamma[\Gamma_{\ctx}]}{{\e'}}{\D}{\sharingRel_2}$ where  
    $\Finer{\sharingRel_2}{\sharingRel_1}$ (${\sharingRel_2}={\sharingRel_1}$), 
    then $\TypeCheck{\Gamma}{\CtxP{\e'}}{\C}{\sharingRel'}$ for some $\ctxP$ such that
    $\CtxP{\e'}\variant\Ctx{\e'}$ and
    $\Finer{\sharingRel'}{\sharingRel}$ (${\sharingRel'}={\sharingRel}$). 
\end{enumerate}
\end{lemma}
\begin{proof}
The proof is in~\ref{app:proofs}.  
\end{proof}

The following lemma is used to prove that the elimination rules, namely \rn{alias-elim} and \rn{\EZ{affine-elim}},
do not  introduce sharing. In particular, for \rn{alias-elim} a non-affine variable $\x$ is
substituted with a non-affine variable $\y$ which is  in the equivalence class of $\x$ in the sharing relation
$\sharingRel$, so that there is no newly produced connection. For \rn{\EZ{affine-elim}}, an affine
variable is substituted with a capsule value, so also in this case there is no newly produced connection.
\begin{lemma} {\rm (Substitution)}\label{lemma:substitution}
  \begin{enumerate}
    \item 
      If $\TypeCheck{\Gamma,\x{:}\D,\y{:}\D}{\e}{\C}{\sharingRel}$, then 
      $\TypeCheck{\EZ{\Gamma}}{\Subst{\e}{\y}{\x}}{\C}{\Remove{\sharingRel}{\x}}$. 
    \item Let $\TypeCheck{\Gamma,\x{:}\Type{\capsule}{\D}}{\e}{\C}{\sharingRel}$. If
      $\TypeCheck{\Gamma}{\val}{\D}{\epsilon}$, then 
      $\TypeCheck{\EZ{\Gamma}}{\SubstVal{\e}{\val}{\x}}{\C}{\sharingRel}$. 
  \end{enumerate}
\end{lemma}
\begin{proof}
By induction on type derivation. For point 1. we use \refToProp{lessSrRel}.\ref{p5}.
\end{proof}
The previous lemma can be easily extended to the type checking judgement for declarations
$\TypeCheckDecs{\EZ{\Gamma}}{\decs}{\sharingRel}$.

The following lemma asserts that the sharing relation of a subexpression is finer than the
sharing relation of the expression that contains it.
\begin{lemma}\label{lemma:monotoneSharing}
Let $\TypeCheck{\Gamma}{\Ctx{e}}{\C}{\sharingRel}$ and $\TypeCheck{\Gamma}{\e}{\D}{\sharingRel'}$. 
If $\Pair{\x}{\y}\in\sharingRel'$ with $\x,\y\not\in\HB{\ctx}$ and $\x,\y\neq\resV$, 
then $\Pair{\x}{\y}\in\sharingRel$.
\end{lemma}
\begin{proof}
The proof is in~\ref{app:proofs}.  
\end{proof}
The following lemma {states} a technical property needed to prove that sharing
relations are preserved when we reduce a field access redex $\FieldAccess{\x}{\f}$ \EZ{to the reference $\y$ stored in the field. Recall that a set of variables $\X$ stands for the sharing relation where $\X$ is an equivalence class and the others are singletons.}
\begin{lemma}\label{lemma:fieldAcc}
If $\TypeCheck{\Gamma}{\Ctx{e_1}}{\C}{\sharingRel_1}$, $\TypeCheck{\Gamma}{\Ctx{e_2}}{\C}{\sharingRel_2}$,
$\TypeCheck{\Gamma}{\e_1}{\D}{\{\x,\resV\}}$ and $\TypeCheck{\Gamma}{\e_2}{\D}{\{\y,\resV\}}$ with 
$\{\x,\y\}\cap\HB{\ctx}=\emptyset$. Then $\sharingRel_1+\{\x,\y\}=\sharingRel_2+\{\x,\y\}$.
\end{lemma}
\begin{proof}
The proof is in~\ref{app:proofs}.  
\end{proof}

\EZ{As already mentioned, the subject reduction theorem states that in a reduction step $\reduce{\e_1}{\e_2}$:
\begin{enumerate}[(1)]
\item $\e_2$ preserves the type of $\e_1$, and has less or equal sharing effects.
\item For each variable declaration $\Dec{\x}{\_}{\e}$ occurring in $\e_1$ and reduced to $\Dec{\x}{\_}{\e'}$ in $\e_2$, $\e'$ preserves the type of $\e$.
Moreover, ``$\e'$ inside $\e_2$'' has less or equal sharing effects than ``$\e$ inside $\e_1$'', where such sharing effects are those of the initialization expression, plus the connections existing in the store (sequence of evaluated declarations) currently available in the enclosing expression.
\end{enumerate}
Invariant (2) corresponds, in a sense, to the invariant on store which we would have in a conventional calculus, and allows us to express and prove the expected properties of lent and capsule references. }
Note that there is no
guarantee that the sole sharing effects of the initialization expression are reduced, since a
new free variable could be introduced in the expression as an effect of field
access. 

\EZ{To formally express invariant (2), we need some notations and definitions. First of all, we need to trace the reduction of right-hand sides of declarations. To simplify the notation, we assume in the following that expressions contain at most one declaration for a variable (no shadowing, as can be always obtained by alpha-conversion). We define \emph{declaration contexts $\decctx{\mux}$} by:}
\begin{quote}
\begin{grammatica}
\produzione{\decctx{\EZ{\mux}}}
{\BlockLab{\dvs\ \Dec{\Type{\mu}{\C}}{\x}{\emptyctx}\ \decs}{\e}{\X}
\mid\BlockLab{\dvs\ \Dec{\T}{\y}{\decctx{\EZ{\mux}}}\ \decs}{\e}{\X}
\mid\BlockLab{\dvs}{\decctx{\EZ{\mux}}}{\X}
}{}
\end{grammatica}
\end{quote}
That is, in $\Decctx{\EZ{\mux}}{\e}$ the expression $\e$ occurs as the right-hand side
of the (unique) declaration for reference $\x$\EZ{, which has qualifier $\mu$}.  \EZ{Since declaration contexts are a subset of evaluation contexts, the same assumptions and definitions hold.} 

To lighten the notation we
write simply $\decctx{\x}$ when the modifier is not relevant, and $\decctx{}$ when not even the variable is relevant.

\EZ{We define now the sharing relation $\induced{\decctx{\x}}$ induced by the store (sequence of evaluated declarations) enclosing the hole in $\decctx{\x}$. To this end, we first inductively define $\extractAllDec{\decctx{\x}}$:}
\begin{itemize}
\item $\extractAllDec{\BlockLab{\dvs\ \Dec{\Type{\mu}{\C}}{\x}{\emptyctx}\ \decs}{\e}{\X}}=\dvs$
\item $\extractAllDec{\BlockLab{\dvs}{\decctx{\x}}{\X}}=\extractAllDec{\BlockLab{\dvs\ \Dec{\T}{\y}{\decctx{\x}}\ \decs}{\e}{\X}}=\dvs\ \extractAllDec{\decctx{\x}}$
\end{itemize}
\EZ{Then, if $\extractAllDec{\decctx{\x}}=\Dec{\C_1}{\x_1}{\ConstrCall{\C_1}{\xs_1}}\ \cdots \Dec{\C_n}{\x_n}{\ConstrCall{\C_n}{\xs_n}}$,
we define
$\induced{\decctx{\x}}=\X_1+\cdots+X_n$ where $\X_i=\{x_i,\xs_i\}$ ($1\leq i\leq n$).}

To prove \EZ{invariant (2) in the subject reduction} we first need to show that it holds for	
the congruence relation.

\begin{lemma}\label{lemma:congrSR}
If $\TypeCheck{\Gamma}{\Decctx{\x}{\e}}{\C}{\sharingRel}$ and $\TypeCheck{\Gamma[\TypeEnv{\decctx{\x}}]}{\e}{\D}{\sharingRel_x}$
and $\congruence{\Decctx{\x}{\e}}{\e_1}$ \EZComm{do we need to mention $\e_1$?}\PGComm{Nella prova} where $\e_1=\DecctxP{\x}{\e'}$ for some $\decctxP{\x}$ and $\e'$, then 
$\TypeCheck{\Gamma[\TypeEnv{\decctxP{\x}}]}{\e'}{\D}{\sharingRel'_x}$ 
and $\induced{\decctx{\x}}+\sharingRel_{x}=\induced{\decctxP{\x}}+\sharingRel'_{x}$.
\end{lemma}
\begin{proof}
By induction on $\decctx{\x}$.\\
\underline{Case $\BlockLab{\dvs\ \Dec{\T}{\x}{\emptyctx}\ \decs}{\e_b}{\X}$}.
Then $\TypeCheck{\Gamma}{\BlockLab{\dvs\ \Dec{\T}{\x}{\e}\ \decs}{\e_b}{\X}}{\C}{\sharingRel}$ and 
$\TypeEnv{\decctx{\x}}=\Gamma_{\dvs},\TypeDec{\x}{\T},\Gamma_{\decs'}$. Let 
$\congruence{\Decctx{\x}{\e}}{\e_1}$, by \refToLemma{congruence} we have
$\TypeCheck{\Gamma}{\e_1}{\C}{\sharingRel}$.\\
Consider how $\e_1$ could be obtained from  $\Decctx{\x}{\e}$ by applying congruence rules to its subexpressions.\\
Rule \rn{reorder}  applied to the block $\BlockLab{\dvs\ \Dec{\T}{\x}{e}\ \decs}{\e_b}{\X}$
does not modify $\e$ or $\sharingRel_{\dvs}$. In particular, since no declaration following $\emptyctx$ can be evaluated, no 
declaration of $\dvs$ can be moved after the one of $\x$, and no declaration in $\decs$ can be moved. So 
$\e_1=\BlockLab{\dvs'\ \Dec{\T}{\x}{e}\ \decs}{\e_b}{\X}$ where $\dvs'$ is a reordering of $\dvs$. So the result is obvious.\\
Rule \rn{body} applied to the block $\BlockLab{\dvs\ \Dec{\T}{\x}{e}\ \decs}{\e_b}{\X}$ can only
modify the structure of the block by ``inserting curly brackets'' between declaration. For example
$\e_1$ could be $\BlockLab{\dvs_1} {\BlockLab{\dvs_2\  \Dec{\T}{\x}{e}\ \decs}{\e_b}{\Y}}{\X}$
where $\dvs=\dvs_1\,\dvs_2$ and $\FV{\dvs_1}\cap\dom{\dvs_2}=\emptyset$. Again in this case 
$\e_1=\DecctxP{\x}{\e}$ for some $\decctxP{\x}$ such that $\induced{\decctx{\x}}=\induced{\decctxP{\x}}$ and therefore the result holds.\\
If $\e_1$ is obtained from  $\Decctx{\x}{\e}$ by applying the congruence rules in $\decs$ or $\e_b$ or $\e$, then
$\e_1=\BlockLab{\dvs\ \Dec{\T}{\x}{e'}\ \decs'}{\e'_b}{\X}$, so $\e_1=\DecctxP{\x}{\e'}=\BlockLab{\dvs\ \Dec{\T}{\x}{e'}\ \decs'}{\e'_b}{\X}$ where $\congruence{\e}{\e'}$ . By \refToLemma{congruence} and
$\TypeCheck{\Gamma[\TypeEnv{\decctx{\x}}]}{\e}{\D}{\sharingRel_x}$ we have $\TypeCheck{\Gamma[\TypeEnv{\decctx{\x}}]}{\e'}{\D}{\sharingRel_x}$. It is easy to see that $\FV{\e}=\FV{\e'}$. (Congruent expressions have the same set of free variables.) 
So again the result holds.\\
Since the congruence has to produce an expression of the shape of $\decctxP{\x}$ no rule can be applied to $\dvs$.\\
Finally, we may have that $\DecctxP{\x}{\e'}$ is obtained by application of rule \rn{dec} to the declaration of $\x$. There are two
cases
\begin{enumerate}[(1)]
\item $\Decctx{\x}{\e}=\BlockLab{\dvs\ \Dec{\T}{\x}{\BlockLab{\dvs_1\ \decs_2}{\MS{\e_b}}{\X}}\ \MS{\decs'}}{\MS{\e'_b}}{\Y}$ and \\
 $\DecctxP{\x}{\e'}=\BlockLab{\dvs\ \dvs_1\ \Dec{\T}{\x}{\BlockLab{\decs_2}{\MS{\e_b}}{\X'}}\ \decs'}{\e'_b}{\Y'}$ 
 or
\item $\Decctx{\x}{\e}=\BlockLab{\dvs\ \dvs_1\ \Dec{\T}{\x}{\BlockLab{\decs_2}{\MS{\e_b}}{\X'}}\ \decs'}{\e'_b}{\Y'}$ and
\\
$\DecctxP{\x}{\e'}=\BlockLab{\dvs\ \Dec{\T}{\x}{\BlockLab{\dvs_1\ \decs_2}{\MS{\e_b}}{\X}}\ \MS{\decs'}}{\MS{\e'_b}}{\Y}$ 
\end{enumerate}
where in both cases 
\begin{enumerate}[(1)]\addtocounter{enumi}{2}
\item $\FV{\dvs_1}\cap\dom{\decs_2}=\emptyset$ and $\FV{\dvs\,\decs'\,{\MS{\e'_b}}}{\cap}\dom{\dvs_1}{=}\emptyset$ and
\item $\mu=\capsule$ implies $\dom{\dvs_1}{\cap}\X{=}\emptyset$
\end{enumerate}
From (3) we have that, if $\Gamma'=\Gamma[\Gamma_{\dvs},\Gamma_{\MS{\decs'}}][\Gamma_{\dvs_1},\Gamma_{\decs_2}]$
and $\Gamma''=\Gamma[\Gamma_{\dvs},\Gamma_{\MS{\decs'}},\Gamma_{\dvs_1}][\Gamma_{\decs_2}]$, then $\Gamma'=\Gamma''$.\\
From $\TypeCheck{\Gamma}{\Decctx{\x}{\e}}{\C}{\sharingRel}$ and \refToLemma{invBlock} we have that
\begin{enumerate}[(a)]
\item $\TypeCheck{\Gamma[\Gamma_{\dvs},\Gamma_{\MS{\decs'}}]}{\BlockLab{\dvs_1\ \decs_2}{\MS{\e_b}}{\X}}{\MS{\T}}{\sharingRel_x}$ 
and $\TypeCheckDecs{\Gamma[\Gamma_{\dvs},\Gamma_{\decs}]}{\dvs}{\sharingRel_d}$ where
\item $\sharingRel_x=\sharingRel_1+\sharingRel_2+\sharingRel_e$
\item $\TypeCheckDecs{\Gamma'}{\dvs_1}{\sharingRel_1}$ and $\TypeCheckDecs{\Gamma'}{\decs_2}{\sharingRel_2}$ and $\TypeCheck{\Gamma'}{\e_b}{\_}{\sharingRel_e}$
\end{enumerate}
From $\congruence{\Decctx{\x}{\e}}{\DecctxP{\x}{\e'}}$ and \refToLemma{congruence} we have that 
$\TypeCheck{\Gamma}{\DecctxP{\x}{\e'}}{\C}{\sharingRel}$. 
From \refToLemma{invBlock}
\begin{enumerate}[(a)]\addtocounter{enumi}{3}
\item $\TypeCheck{\Gamma[\Gamma_{\dvs},\Gamma_{\MS{\decs'}},\Gamma_{\dvs_1}]}{\BlockLab{\decs_2}{\MS{\e_b}}{\X'}}{\D}{\sharingRel'_x}$ and\\ 
 $\TypeCheckDecs{\Gamma[\Gamma_{\dvs},\Gamma_{\MS{\decs'}},\Gamma_{\dvs_1}]}{\dvs}{\sharingRel'_d}$ and $\TypeCheckDecs{\Gamma[\Gamma_{\dvs},\Gamma_{\MS{\decs'}},\Gamma_{\dvs_1}]}{\dvs_1}{\sharingRel'_1}$  where
\item $\sharingRel'_x=\sharingRel'_2+\sharingRel'_e$ 
\item $\TypeCheckDecs{\Gamma''}{\decs_2}{\sharingRel'_2}$ and 
$\TypeCheck{\Gamma''}{\e'_b}{\_}{\sharingRel'_e}$
\end{enumerate}
By \refToLemma{weakening}.1 and 2 we have that $\sharingRel'_1=\sharingRel_1$ and
$\sharingRel'_2=\sharingRel_2$ and $\MS{\sharingRel'_e}=\MS{\sharingRel_e}$ and $\sharingRel'_d=\sharingRel_d$.\\
From the fact that forward references can be done only to evaluated declarations and there are none in $\decs$
we have that in $\dvs$ and $\dvs_1$ there cannot free variable in  $\dom{\decs}$. Therefore 
$\sharingRel_d=\sharingRel_{\dvs}$. Moreover, from (3) in $\dvs_1$ there cannot free variable in $\dom{\decs_2}$ and
$\sharingRel_1=\sharingRel_{\dvs_1}$. So we have that\\
\centerline{
$
\begin{array}{lcll}
\induced{\decctx{\x}}+\sharingRel_{x}&=&\sharingRel_{\dvs}+(\sharingRel_1+\sharingRel_2+\sharingRel_e)&\text{by definition of $\induced{\decctx{\x}}$ and $\sharingRel_{x}$}\\
&=&\sharingRel_{\dvs}+\sharingRel_{\dvs_1}+(\sharingRel_2+\sharingRel_e)\\
&=&\induced{\decctxP{\x}}+(\sharingRel_2+\sharingRel_e)&\text{by definition of $\induced{\decctxP{\x}}$}\\
&=&\induced{\decctxP{\x}}+\sharingRel'_{x}&\text{by definition of $\sharingRel'_x$}
\end{array}
$
}
This proves the result for both cases (1) and (2)

\medskip\noindent
The \underline{cases $\BlockLab{\dvs\ \Dec{\T}{\y}{\decctx{\x}}\ \decs}{\e_b}{\X}$ and $\BlockLab{\dvs}{\decctx{\x}}{\X}$} 
are proved using the inductive hypothesis and a case analysis on the congruence used for the block as for the previous case. \\
\end{proof}

\begin{theorem}{\rm (Subject reduction)}\label{theo:subred}
If $\TypeCheck{\Gamma}{\e_1}{\C}{\sharingRel}$ and $\reduce{\e_1}{\e_2}$, then
\begin{enumerate}
  \item  \PG{$\TypeCheck{\Gamma}{\e'_2}{\C}{\sharingRel'}$ where $\e_2'\variant\e_2$ 
    $\Finer{\sharingRel'}{\sharingRel}$, and}
  {\item if $\e_1=\Decctx{\x}{\e}$, $\e_2=\DecctxP{\x}{\e'}$, 
    and $\TypeCheck{\TypeEnv{\decctx{\x}}}{\e}{\D}{\sharingRel_x}$ we have 
    that: $\TypeCheck{\TypeEnv{\decctxP{\x}}}{\e''}{\D}{\sharingRel'_x}$ where $\e''\variant\e'$ and 
    $\Finer{(\induced{\decctxP{\x}}+\sharingRel'_{x})}{(\induced{\decctx{\x}}+\sharingRel_{x})}$.}
\end{enumerate}
\end{theorem}
\begin{proof}
By induction on the derivation of $\reduce{\e_1}{\e_2}$ with a case analysis on the last rule of \refToFigure{reduction} used for $\reduce{\Ctx{\redex}}{\CtxP{\e'}}$.
We show the two most interesting cases, which are \rn{congr}, \rn{field-access} and
\rn{field-assign}, and just hint the one for \rn{alias-elim} and
\rn{\EZ{affine-elim}}. The proof of the remaining cases is in~\ref{app:proofs}. 

\underline{Rule \rn{congr}}.
In this case 
    \begin{itemize}
      \item  $\congruence{\e_1}{\e'_1}$ 
      \item $\reduce{\e'_1}{\e_2}$
    \end{itemize}
    If $\TypeCheck{\Gamma}{\e_1}{\C}{\sharingRel}$ and 
$\congruence{\e_1}{\e'_1}$, from \refToLemma{congruence},
$\TypeCheck{\Gamma}{\e'_1}{\C}{\sharingRel}$.
\begin{enumerate}
\item   
By induction hypothesis on $\e'_1$ we have that
$\TypeCheck{\Gamma}{\e'_2}{\C}{\sharingRel'}$ where  $\e_2'\variant\e_2$ and $\Finer{\sharingRel'}{\sharingRel}$.

\item If  $\e_1=\Decctx{\y}{\e}$  and $\congruence{\e_1}{\e'_1}$, then $\e'_1=\DecctxS{\y}{\e''}$
for some $\decctxS{\y}$. Let  $\e_2=\DecctxP{\y}{\e'}$, 
 from $\TypeCheck{\TypeEnv{\decctx{\y}}}{\e}{\D}{\sharingRel_y}$ and
  \refToLemma{congrSR}, we have that
  $\TypeCheck{\TypeEnv{\decctxS{\y}}}{\e''}{\D}{\sharingRel''_y}$ and  
  ${(\induced{\decctx{\y}}+\sharingRel_{y})}={(\induced{\decctxS{\y}}+\sharingRel''_{y})}$.
  By induction hypothesis $\TypeCheck{\TypeEnv{\decctxP{\x}}}{\e''}{\D}{\sharingRel'_y}$ where $\e''\variant\e'$ and
  $\Finer{(\induced{\decctxP{\y}}+\sharingRel'_{y})}{(\induced{\decctxS{\y}}+\sharingRel''_{y})}$.
  Therefore
  $\Finer{(\induced{\decctxP{\y}}+\sharingRel'_{y})}{(\induced{\decctx{\y}}+\sharingRel_{y})}$.
\end{enumerate}

\medskip
\underline{Rule \rn{field-access}}. 
\begin{enumerate}
  \item In this case 
    \begin{itemize}
			\item \MS{$\e_1 = \ctx[\rho]$, and $\e_2 = \ctx'[\e']$}
      \item $\ctxP=\ctx=\DecCtx{\ctx}{\x}{\ctx_1}$ for some $\ctx_1$
      \item $\extractDec{\ctx}{\x}=\dv_x=\Dec{\D}{\x}{\ConstrCall{\D}{\x_1,\ldots,\x_n}}$
      \item $\redex=\FieldAccess{\x}{\f_i}$ and $\e'=\x_i$ with $\x_i\not\in\HB{\ctx_1}$.
    \end{itemize}
    By definition of $\decCtx{\ctx}{\x}$, for some $\ctx_2$ 
    \begin{enumerate}[(1)]
      \item either $\decCtx{\ctx}{\x}=\ctx_2[\Block{\decs'}{\e_b}]$ with 
        $\decs'=\dvs\ \dv_x\ \Dec{\T}{\z}{\ctx_1[\FieldAccess{\x}{\f_i}]}\ \decs$
      \item or $\decCtx{\ctx}{\x}=\ctx_2[\Block{\dvs\ \dv_x}{\ctx_1}]$.
    \end{enumerate}
    and $\x\not\in\HB{\ctx_1}$. \\
    We consider only case (1) since the proof for the other case is similar and simpler.\\
    From \refToLemma{context}.1 we have that
    $\TypeCheck{\Gamma[\Gamma_{\ctx_2}]}{\Block{\decs'}{\e_b}}{\C'}{\sharingRel_1}$
    for some $\C'$ and $\sharingRel_1$. From typing rule \rn{T-block},
    $\TypeCheck{\Gamma[\Gamma_{\ctx_2}][\Gamma_{\decs'}]}{\ctx_1[\FieldAccess{\x}{\f_i}]}{\D'}{\sharingRel_2}$
    for some $\D'$ and $\sharingRel_2$. From \refToLemma{context}.1 and typing
    rules \rn{field-access} and \rn{var} we have that
    $\TypeCheck{\Gamma[\Gamma_{\ctx_2}][\Gamma_{\decs'}][\Gamma_{\ctx_1}]}{\FieldAccess{\x}{\f_i}}{\D_i}{\{\x,\resV\}}$ and
    $\TypeCheck{\Gamma[\Gamma_{\ctx_2}][\Gamma_{\decs'}][\Gamma_{\ctx_1}]}{\x_i}{\D_i}{\{{x_i},\resV\}}$.
    (Note that neither $\x$ nor $\x_i$ can be forward references to non
    evaluated declarations and therefore they must be defined without the
    $\capsule$ modifier.) From \refToLemma{fieldAcc} we have that: if
    $\TypeCheck{\Gamma[\Gamma_{\ctx_2}][\Gamma_{\decs'}]}{\ctx_1[\x_i]}{\D'}{\sharingRel'_2}$,
    then $\sharingRel_2+\{\x,\x_i\}=\sharingRel'_2+\{\x,\x_i\}$. Since
    $\TypeCheckDecs{\Gamma[\Gamma_{\ctx_2}][\Gamma_{\decs'}]}{\dv_x}{\{\x,\x_1,\ldots,\x_n,\resV\}}$,
    we have that:\\ 
    $\TypeCheckDecs{\Gamma[\Gamma_{\ctx_2}][\Gamma_{\decs'}]}{\Dec{\T}{\z}{\ctx_1[\FieldAccess{\x}{\f_i}]}\ \dv_x}{\sharingRel_3}$ implies
    $\TypeCheckDecs{\Gamma[\Gamma_{\ctx_2}][\Gamma_{\decs''}]}{\Dec{\T}{\z}{\ctx_1[\x_i]}\ \dv_x}{\sharingRel_3}$ where
    $\decs''=\dvs\ \dv_x\ \Dec{\T}{\z}{\ctx_1[x_i]}\ \decs$. Therefore
    \begin{center}
      $\TypeCheck{\Gamma[\Gamma_{\ctx_2}]}{\Block{\dvs\ \dv_x\ \Dec{\T}{\z}{\ctx_1[\x_i]}\ \decs}{\e_b}}{\C'}{\sharingRel_1}$
    \end{center}
    and from \refToLemma{context}.2 we derive \PG{$\TypeCheck{\Gamma}{\CtxP{\e''}}{\C}{\sharingRel}$ where $\Ctx{\e'}\variant\CtxP{\e''}$.}
  \item Let  $\e_1=\Decctx{\y}{\e}$, $\e_2=\DecctxP{\y}{\e'}$, and 
    $\reduce{\e_1=\Ctx{\FieldAccess{\x}{\f_i}}}{\e_2=\CtxP{\x_i}}$. If the
    redex $\FieldAccess{\x}{\f_i}$ is not a subexpression of $\e$ then
    $\e=\e'$, and since $\TypeEnv{\decctx{\y}}=\TypeEnv{\decctxP{\y}}$, the
    result is obvious. If $\FieldAccess{\x}{\f_i}$ is a subexpression of $\e$,
    then from \refToLemma{decomposition} for some $\ctxS$ we have that
    $\e=\ctxS[\FieldAccess{\x}{\f_i}]$ and $\e'=\ctxS[{\x_i}]$. From
    $\TypeCheck{\TypeEnv{\decctx{\x}}}{\ctxS[\FieldAccess{\x}{\f_i}]}{\D}{\sharingRel_x}$
    and $\TypeCheck{\TypeEnv{\decctx{\x}}}{\ctxS[{\x_i}]}{\D}{\sharingRel'_x}$,
    and \refToLemma{context}.1 we have that
    $\TypeCheck{\Gamma[\TypeEnv{\decctx{\x}}][\Gamma_{\ctxS}]}{\FieldAccess{\x}{\f_i}}{\D_i}{\{\x,\resV\}}$
    and
    $\TypeCheck{\Gamma[\TypeEnv{\decctx{\x}}][\Gamma_{\ctxS}]}{\x_i}{\D_i}{\{{x_i},\resV\}}$.
    If $\dv_x\in\extractAllDec{\ctxS}$ then $\sharingRel_x=\sharingRel'_x$, otherwise
    $\dv_x\in\extractAllDec{\decctx{\y}}$. In both cases  
    $\induced{\decctx{\y}}+\sharingRel_x=\induced{\decctx{\y}}+\sharingRel'_x$.
\end{enumerate}

\underline{Rule \rn{field-assign}}. 
\begin{enumerate}
  \item In this case 
    \begin{itemize}
      \item $\ctx=\DecCtx{\ctx}{\x}{\ctx_1}$ for some $\ctx_1$, 
      \item $\ctx'=\UpdateCtxX{\y}{\x}{i}{\ctx_1}$ since $\x\not\in\HB{\ctx_1}$
      \item $\extractDec{\ctx}{\x}=\dv_x=\Dec{\D}{\x}{\ConstrCall{\D}{\x_1,\ldots,\x_n}}$
      \item $\redex=\FieldAssign{\x}{\f_i}{\y}$ and $\e'=\y$ with $\y\not\in\HB{\ctx_1}$.
    \end{itemize}
    As for the case of field update $\decCtx{\ctx}{\x}$ has either shape (1) or (2) with 
    $\y\not\in\HB{\ctx_1}$. We consider only case (1).\\
    From \refToLemma{context}.1 we have that
    $\TypeCheck{\Gamma[\Gamma_{\ctx_2}]}{\Block{\decs'}{\e_b}}{\C'}{\sharingRel_1}$
    for some $\C'$ and $\sharingRel_1$. From typing rule \rn{T-block} we have that
    $\TypeCheck{\Gamma[\Gamma_{\ctx_2}][\Gamma_{\decs'}]}{\ctx_1[\FieldAssign{\x}{\f_i}{\y}]}{\D'}{\sharingRel_2}$
    for some $\D'$ and $\sharingRel_2$. From \refToLemma{context}.1, typing
    rule \rn{T-Field-assign} and \rn{t-var}  we have that
    $\TypeCheck{\Gamma[\Gamma_{\ctx_2}][\Gamma_{\decs'}][\Gamma_{\ctx_1}]}{\x}{\D}{\{\x,\resV\}}$,
    $\TypeCheck{\Gamma[\Gamma_{\ctx_2}][\Gamma_{\decs'}][\Gamma_{\ctx_1}]}{\y}{\D_i}{\{\y,\resV\}}$, and
    $\TypeCheck{\Gamma[\Gamma_{\ctx_2}][\Gamma_{\decs'}][\Gamma_{\ctx_1}]}{{\FieldAssign{\x}{\f_i}{\y}}}{\D_i}{\{\x,\y,\resV\}}$
    where $\fields{\D}{=}\Field{\D_1}{\f_1}\ldots\Field{\D_n}{\f_n}$ and
    $\Finer{\{\y,\resV\}}{\{\x,\y,\resV\}}$.\\
    From $\TypeCheck{\Gamma[\Gamma_{\ctx_2}]}{\Block{\decs'}{\e_b}}{\C'}{\sharingRel_1}$
    and \refToLemma{invBlock}, we have that
    \begin{enumerate}[(a)]
      \item $\TypeCheckDecs{\Gamma[\Gamma_{\ctx_2}][\Gamma_{\decs'}]}{\Dec{\T}{\z}{\ctx_1[\FieldAssign{\x}{\f_i}{\y}]}}{\sharingRel_3}$, 
      \item $\TypeCheckDecs{\Gamma[\Gamma_{\ctx_2}][\Gamma_{\decs'}]}{\dv_x}{\sharingRel_x}$ 
        where $\sharingRel_x=\{\x,\x_1,\ldots,\x_n\}$,
      \item $\TypeCheckDecs{\Gamma[\Gamma_{\ctx_2}][\Gamma_{\decs'}]}{\dvs\ \decs}{\sharingRel_4}$, and 
      \item $\TypeCheck{\Gamma[\Gamma_{\ctx_2}][\Gamma_{\decs'}]}{\e_b}{\C'}{\sharingRel_5}$. 
    \end{enumerate}
    Let $\dv'_x=\Dec{\D}{\x}{\ConstrCall{\D}{\x_1,\ldots,\x_{i-1},\y,\x_{i+1},\ldots,\x_n}}$ 
    and $\decs''=\dvs\ \dv'_x\ \Dec{\T}{\z}{\ctx_1[\y]}\ \decs$. We have that 
    $\Gamma_{\decs'}=\Gamma_{\decs''}$. Therefore
    \begin{enumerate}[(A)]
      \item $\TypeCheckDecs{\Gamma[\Gamma_{\ctx_2}][\Gamma_{\decs''}]}{\Dec{\T}{\z}{\ctx_1[\y]}}{\sharingRel'_3}$, 
      \item $\TypeCheckDecs{\Gamma[\Gamma_{\ctx_2}][\Gamma_{\decs''}]}{\dv'_x}{\sharingRel_{x'}}$ 
        where $\sharingRel_{x'}=\{\x,\x_1,\ldots,\x_{i-1},\y,\x_{i+1},\ldots,\x_n\}$,
      \item $\TypeCheckDecs{\Gamma[\Gamma_{\ctx_2}][\Gamma_{\decs''}]}{\dvs\ \decs}{\sharingRel_4}$, and 
      \item $\TypeCheck{\Gamma[\Gamma_{\ctx_2}][\Gamma_{\decs''}]}{\e_b}{}{\sharingRel_5}$. 
    \end{enumerate}
    From typing rules \rn{T-field-assign} and \rn{t-var}, and
    \refToLemma{fieldAcc} we get
    $\sharingRel_3+\{\x,\y\}=\sharingRel'_3+\{\x,\y\}$. Moreover, from (a),
    typing rule \rn{T-field-assign}, and \refToLemma{monotoneSharing} we have
    that $\x$ and $\y$ are in the same equivalence class in $\sharingRel_3$,
    i.e., $\Closure{\x}{\sharingRel_3}=\Closure{\y}{\sharingRel_3}$.
    Therefore,
    $\Finer{\sharingRel'_3+\sharingRel_{x'}}{\sharingRel_3+\sharingRel_{x}}$.
    From (A)$\div$(D), typing rule \rn{T-block} and  \refToLemma{context}.2 we get 
    \begin{center}
      $\TypeCheck{\Gamma[\Gamma_{\ctx_2}]}{\Block{\dvs\ \dv'_x\ \Dec{\T}{\z}{\ctx_1[\y]}\ \decs}{\e_b}}{\C'}{\sharingRel'_1}$
    \end{center}
    where $\Finer{\sharingRel'_1}{\sharingRel_1}$. From \refToLemma{context}.2
    we derive that $\TypeCheck{\Gamma}{\ctx_2[\Block{\dvs\ \dv'_x\
    \Dec{\T}{\z}{\ctx_1[\y]}\ \decs}{\e_b}]}{\C}{\sharingRel'}$ where
    $\Finer{\sharingRel'}{\sharingRel}$.  Consider
    $\UpdateCtxX{\y}{\x}{i}{\ctx_1}=\ctx_2[\Block{\dvs\ \dv'\
    \Dec{\T}{\z}{\ctx_1}\ \decs}{\e_b}]$, we have that
    $\TypeCheck{\Gamma}{\UpdateCtxX{\y}{\x}{i}{\ctx_1[\y]}}{\C}{\sharingRel'}$,
    which proves the result.
  \item Let  $\e_1=\Decctx{\y}{\e}$, $\e_2=\DecctxP{\y}{\e'}$, and 
    $\reduce{\e_1=\Ctx{\FieldAssign{\x}{\f_i}{\y}}}{\e_2=\CtxP{\y}}$. If the
    redex is not a subexpression of $\e$ then $\e=\e'$, and since
    $\TypeEnv{\decctx{\y}}=\TypeEnv{\decctxP{\y}}$, the result is obvious. If
    $\FieldAssign{\x}{\f_i}{\y}$ is a subexpression of $\e$, then from
    \refToLemma{decomposition} for some $\ctxS$ we have that
    $\e=\ctxS[\FieldAssign{\x}{\f_i}{\y}]$ and $\e'=\ctxS[{\y}]$. From
    $\TypeCheck{\TypeEnv{\decctx{\x}}}{\ctxS[\FieldAssign{\x}{\f_i}{\y}]}{\D}{\sharingRel_x}$,
    $\TypeCheck{\TypeEnv{\decctx{\x}}}{\ctxS[{\y}]}{\D}{\sharingRel'_x}$, and
    \refToLemma{context}.1 we have that
    $\TypeCheck{\Gamma[\TypeEnv{\decctx{\x}}][\Gamma_{\ctxS}]}{\FieldAssign{\x}{\f_i}{\y}}{\D_i}{\{\x,\y,\resV\}}$ and
    $\TypeCheck{\Gamma[\TypeEnv{\decctx{\x}}][\Gamma_{\ctxS}]}{\y}{\D_i}{\{\y,\resV\}}$.
    From $\Finer{\{\y,\resV\}}{\{\x,\y,\resV\}}$, and \refToLemma{context}.2 we
    derive that $\Finer{\sharingRel'_x}{\sharingRel_x}$.  Therefore, from
    \refToProp{lessSrRel}.\ref{p2} $\Finer{\induced{\decctx{\y}}+\sharingRel'_x}{\induced{\decctx{\y}}+\sharingRel_x}$
\end{enumerate}

\underline{Rule \rn{alias-elim}}. 
Clause 1. is proved  using \refToLemma{substitution}.1. Clause 2. is proved as
in the case of \rn{field-assign}.

\underline{Rule \rn{affine-elim}}.
Clause 1. is proved  using \refToLemma{substitution}.2  and
\ref{lemma:sharingCapsule}. Clause 2. is proved as in the case of
\rn{field-assign}.
\end{proof}

\EZComm{rewritten from here}
\paragraph{Lent and capsule properties}
We can now formally express the lent and capsule notions, informally
described in the Introduction.

Informally, a reference $\x$ is (used as) lent if no sharing can be introduced through $\x$. Formally, if $\TypeCheck{\Gamma}{\e}{\_}{\sharingRel}$, then $\x$ is \emph{lent in $\e$} if $\Closure{\x}{\sharingRel}=\{\x\}$. The notion can be generalized to a set of references $\xs$, that is, $\xs$ is \emph{lent in $\e$} if, for each $\x\in\xs$, $\Closure{\x}{\sharingRel}\subseteq \xs$.

Consider now an expression in a \EZ{declaration} context $\Decctx{\x}{\e}$, and a reference $\y$. The portion of store connected to $\y$ before the evaluation of $\e$ is $\Closure{\y}{\induced{\decctx{\x}}}$.
Now, if  the expression $\e$ can access such portion of store only through lent references, then the two following properties are ensured by the evaluation of $\e$:

\begin{enumerate}
\item such portion of store remains isolated from others
\item  such portion of store cannot be part of the final result of $\e$.
\end{enumerate}

This is formally expressed by the following theorem.

\begin{theorem}{\rm (Lent)}\label{theo:lent}
If $\IsWellTyped{\Decctx{}{\e}}$,
$\TypeCheck{\TypeEnv{\decctx{}}}{\e}{\_}{\sharingRel}$, and $\ys=\Closure{\y}{\induced{\decctx{}}}$ is lent in $\e$, then:
\begin{enumerate}
\item if $\Decctx{}{\e}\longrightarrow\DecctxP{}{\e'}$ then $\TypeCheck{\TypeEnv{\decctxP{}}}{\PG{\e''}}{\_}{\sharingRel'}$ where 
\PG{$\e''\variant\e'$} and $\Closure{\y}{(\induced{\decctxP{}}+\sharingRel')}\subseteq\ys$
\item if $\Decctx{}{\e}\longrightarrow^\star\DecctxP{}{\val}$, then $\y\not\in\FV{\remGarbage(\val)}$.
\end{enumerate}
\end{theorem}
\begin{proof}\
\begin{enumerate}
\item Since $\ys$ is lent in $\e$, $\Closure{\y}{\sharingRel}\subseteq\ys$, hence $\Closure{\y}{(\induced{\decctx{}}+\sharingRel)}=\Closure{\y}{\induced{\decctx{}}}$. 

From
\refToTheorem{subred}.2, since
$\TypeCheck{\TypeEnv{\decctx{}}}{\e}{\_}{\sharingRel}$, we have that
$\TypeCheck{\TypeEnv{\decctxP{}}}{\e''}{\_}{\sharingRel'}$ where 
\PG{$\e''\variant\e'$} and $\Finer{(\induced{\decctxP{}}+\sharingRel')}{(\induced{\decctx{}}+\sharingRel)}$, hence 
$\Closure{\y}{(\induced{\decctxP{}}+\sharingRel')}\subseteq\Closure{\y}{(\induced{\decctx{}}+\sharingRel)}=\Closure{\y}{\induced{\decctx{}}}\EZ{=\ys}$. 
\item 
By induction on the number $n$ of steps of the reduction
$\Decctx{}{\e}\longrightarrow^\star\DecctxP{}{\val}$. \\
For $n=0$, we have $\TypeCheck{\TypeEnv{\decctx{}}}{\val}{\_}{\sharingRel}$. If
we had $\y\in\FV{\EZ{\remGarbage(\val)}}$, then, from \refToTheorem{freevars}.1, it
should be either $\Pair{\y}{\resV}\in\sharingRel$ or $\y$ affine in $\TypeEnv{\decctx{}}$. However, $\Pair{\y}{\resV}\in\sharingRel$ contradicts the hypothesis that $\ys$ is lent in $\e$, and it is easy to see from the definition that in $\decctx{}$ there are no affine variable declarations.\\
\PG{Let $\Decctx{}{\e} \reduce{} \DecctxP{}{\e'}$ and
$\DecctxP{}{\e''}\longrightarrow^\star\DecctxS{}{\val}$,  where 
$\e''\variant\e'$}, in $n$ steps.  By point 1. of this theorem we have $\TypeCheck{\TypeEnv{\decctxP{}}}{\e''}{\_}{\sharingRel'}$ and \EZ{$\Closure{\y}{(\induced{\decctxP{}}+\sharingRel')}\subseteq\ys$, hence $\Closure{\y}{\sharingRel'}\subseteq\ys$, that is, $\ys$ is lent in $\e''$ (note that to be lent does not depend on the annotations)}, and by inductive
hypothesis on $\TypeCheck{\TypeEnv{\decctxP{}}}{\e''}{\_}{\sharingRel'}$ we
derive the thesis.
\end{enumerate}
\end{proof}

Informally, a capsule is a reachable object graph which is an isolated portion
of store, that is, it does not contain nodes reachable from the outside.  In
our calculus, a reachable object subgraph is a value $\val$, nodes reachable
from the outside are free variables,  hence the condition to be a capsule can
be formally expressed by requiring that $\val$ has no free variables. 
The following theorem states that the right-hand side of a capsule declaration actually reduces to a closed portion of store. 

\PGComm{Review proof with new Theorem 22.2!}
\begin{theorem}{\rm (Capsule)}\label{theo:capsule}
If $\IsWellTyped{\Decctx{\ax}{\e}}$,
$\TypeCheck{\TypeEnv{\decctx{\ax}}}{\e}{\_}{\sharingRel}$ with
$\IsCapsule{\sharingRel}$, and 
$\Decctx{\ax}{\e}\longrightarrow^\star\DecctxP{\ax}{\val}$, then 
$\FV{\remGarbage(\val)}=\emptyset$.
\end{theorem}
{\begin{proof}
By induction on the number $n$ of steps of the reduction
$\Decctx{\ax}{\e}\longrightarrow^\star\DecctxP{\ax}{\val}$. \\
For $n=0$, we have $\TypeCheck{\TypeEnv{\decctx{}}}{\val}{\_}{\sharingRel}$.
From \refToTheorem{freevars}.2, $\IsCapsule{\sharingRel}$ implies $\y$ affine in $\TypeEnv{\decctx{}}$ for each $\y\in\FV{\val}$. However, 
there are no affine variable declarations in $\decctx{\ax}$.\\
Let $\Decctx{\ax}{\e} \reduce{} \DecctxP{\ax}{\e'}$ and
\PG{$\DecctxP{\ax}{\e''}\longrightarrow^\star\DecctxS{\ax}{\val}$,  where 
$\e''\variant\e'$, in $n$ steps}.  From
\refToTheorem{subred}.2, since
$\TypeCheck{\TypeEnv{\decctx{}}}{\e}{\_}{\sharingRel}$, and
$\Closure{\resV}{\sharingRel}=\emptyset$, we have
$\TypeCheck{\TypeEnv{\decctxP{}}}{\e''}{\_}{\sharingRel'}$, and
$\Closure{\resV}{\sharingRel'}{=\emptyset}$. By induction hypothesis on
$\TypeCheck{\TypeEnv{\decctxP{}}}{\e''}{\_}{\sharingRel'}$ we derive the
result.
\end{proof}}

\paragraph{Progress} Closed expressions are not ``stuck'', that is, they 
either are values or can be reduced.

To prove the theorem we introduce the set of \emph{redexes},
and we show that expressions can be decomposed  in an evaluation
context filled with a redex. Therefore an expression either is a value or it
matches the left-hand side of exactly one reduction rule.
\begin{definition}
Redexes, $\redex$, are defined by:
{\small \begin{center}
$
\begin{array}{c}
\redex ::=\FieldAccess{\x}{\EZ{\f}}\ \mid\ \MethCall{\val}{\m}{\vs}\ \mid\ 
\FieldAssign{\x}{\f}{\y}
\mid\ \BlockLab{\dvs\ \Dec{\C}{\x}{\y}\ \decs }{\e}{\X}\ \mid\ \BlockLab{\dvs\ \Dec{\Type{\capsule}{\C}}{\x}{\val}\ \decs }{\e}{\X}
\end{array}
$
\end{center}
}
\end{definition}

\begin{lemma}{\rm (Decomposition)}\label{lemma:decomposition}
If 
$\e$ is not a value, then there are 
$\ctx$ and $\redex$ such that $\congruence{\e}{\Ctx{\redex}}$. 
\end{lemma}
\begin{proof}
The proof is in~\ref{app:proofs}.  
\end{proof}

{We write $\reduce{\e}{}$ for $\reduce{\e}{\e'}$ for some $\e'$,
$\TypeCheckGround{\e}{\C}{\sharingRel}$ for
$\TypeCheck{\emptyset}{\e}{\C}{\epsilon}$, and $\IsWellTyped{\e}$ for
$\TypeCheckGround{\e}{\C}{\epsilon}$ for some $\C$ (note that an expression
with no free variables has the identity as sharing effects).}
\begin{theorem}{\rm (Progress)}\label{theo:progress}
If $\IsWellTyped{\e}$, and $\e$ is not a value, then 
$\reduce{\e}{}$.
\end{theorem}
\begin{proof}

By \refToLemma{decomposition}, if $\e$ is not a value, then
$\congruence{\e}{\Ctx{\redex}}$.  {By rule \rn{congr}, it is enough to show the
thesis for $\Ctx{\redex}$.  For all $\redex$, except field access and field
update, we have that the corresponding reduction rule is applicable, since
either there are no side conditions (cases \rn{alias-elim} and
\rn{affine-elim}), or the side conditions can be easily shown to hold (case
\rn{invk}).}

In the proof for field access and field update, we use the following auxiliary
notation. Given an evaluation context $\ctx$, the context $\HoleCtx{\ctx}$,
the outermost block of the evaluation context, is defined by
\begin{itemize}
\item $\HoleCtx{\ctx}=\Block{\dvs\ \Dec{\T}{\y}{\emptyctx}\ \decs}{\e}$ if $\ctx=\Block{\dvs\ \Dec{\T}{\y}{\ctxP}\ \decs}{\e}$ and
\item $\HoleCtx{\ctx}=\Block{\dvs}{\emptyctx}$ if  $\ctx=\Block{\dvs}{\ctxP}$.
\end{itemize}

If \underline{$\redex$ is $\FieldAccess{\x}{\f_i}$}, from
\refToLemma{context}.1, rule \rn{t-field-access}, and rule \rn{t-var} of
\refToFigure{typing}, we have that
$\TypeCheck{\Gamma_{\ctx}}{{\FieldAccess{\x}{\f_i}}}{\C_i}{\EZ{\{\x,\resV\}}}$ where
$\fields{\C}=\Field{\C_1}{\f_1}\ldots\Field{\C_n}{\f_n}$. So, we have that \EZ{$\extractDec{\ctx}{\x}=\Dec{\C}{\x}{\ConstrCall{\C}{\x_1,\ldots,\x_n}}$,}
$\decCtx{\ctx}{\x}$ is defined, and, for some $\ctx'$,
$\ctx=\DecCtx{\ctx}{\x}{\ctxP}$.  If $\x_i\not\in\HB{\ctxP}$, then rule
\rn{Field-Access} is applicable. \\
Otherwise, since $\x_i\in\HB{\ctxP}$, $\decCtx{\ctxP}{\x_i}$ is
defined, and
$\CtxP{\FieldAccess{\x}{\f_i}}=\ctx_1[\HoleCtx{{\decCtx{\ctx'}{\x_i}}}[\ctx_2[\FieldAccess{\x}{\f_i}]]]$
for some $\ctx_1$ and $\ctx_2$ such that
${\HoleCtx{{\decCtx{\ctx'}{\x_i}}}[\ctx_2[\FieldAccess{\x}{\f_i}]]}=\BlockLab{\dvs\,\Dec{\EZ{\C_i}}{\x_i}{\val}\,\decs}{\e}{\X}$.
Using rule \rn{alpha} of \refToFigure{congruence} we have that
\begin{quote}
 $\congruence{\BlockLab{\dvs\,\Dec{\EZ{\C_i}}{\x_i}{\val}\,\decs}{\e}{\X}}{\BlockLab{\Subst{\dvs}{\y}{\x_i}\ \Dec{\T}{\y}{\Subst{\val}{\y}{\x_i}}\ \Subst{\decs'}{\y}{\x_i}}{\Subst{\e}{\y}{\x_i}}{{\Subst{\X}{\y}{\x_i}}}}$
 \end{quote}
where $\y$ can be chosen such that $\y\not\in\HB{\ctx_2}$. Therefore 
$\congruence{\DecCtx{\ctx}{\x}{\CtxP{\FieldAccess{\x}{\f_i}}}}{\DecCtx{\ctx}{\x}{\ctx_3[\FieldAccess{\x}{\f_i}]}}$ 
where $\x_i\not\in\HB{\ctx_3}$.  So 
$\reduce{\DecCtx{\ctx}{\x}{\ctx_3[{\FieldAccess{\x}{\f_i}}]}}{\e_2}$ by applying rule \rn{field-access}.

\bigskip

If \underline{$\redex$ is $\FieldAssign{\x}{\f_i}{\y}$}, from
\refToLemma{context}.1, rule \rn{t-field-assign}, and rule \rn{t-var} of
\refToFigure{typing}, we have that
$\TypeCheck{\Gamma_{\ctx}}{{\FieldAssign{\x}{\f_i}{\y}}}{\C_i}{\{\x,\y,\resV\}}$
where $\fields{\C}=\Field{\C_1}{\f_1}\ldots\Field{\C_n}{\f_n}$. So, we have
that $\decCtx{\ctx}{\x}$ is defined, and $\ctx=\DecCtx{\ctx}{\x}{\ctxP}$ for
some $\ctxP$. Therefore, for some $\ctx_1'$,
$\ctx=\ctx_1'[\HoleCtx{\decCtx{\ctx}{\x}}[{\ctxP}]]$.  If
$\y\not\in\HB{\ctxP}$, then rule \rn{Field-Assign} is applicable. \\
Otherwise, since $\y\in\HB{\ctxP}$, $\decCtx{\ctxP}{\y}$ is defined, and
$\CtxP{\FieldAssign{\x}{\f_i}{\y}}=\ctx_1[\HoleCtx{{\decCtx{\ctx'}{\y}}}[\ctx_2[\FieldAssign{\x}{\f_i}{\y}]]]$
for some $\ctx_1$ and $\ctx_2$ such that
$\HoleCtx{{\decCtx{\ctx'}{\y}}}[\ctx_2[\FieldAssign{\x}{\f_i}{\y}]]=\BlockLab{\EZ{\dvs}\,\Dec{\EZ{\C_i}}{\y}{\val}\,\decs}{\e}{}$
where \EZ{$\dvs=\Reduct{(\dvs\, \decs)}{\FV{\val}}$ are all the declarations connected to the free variables of $\val$, hence to be extruded together with the declaration of $\y$}.\\
By induction on the
number $n>0$ of blocks from which we have to extrude the declaration of $\y$.
Let $\dvs_1=\EZ{\dvs}\,\Dec{\EZ{\C_i}}{\y}{\val}$. 
If $n>1$, then for some $\ctx''\neq\emptyctx$, either
\begin{enumerate}[(a)]
  \item  $\HoleCtx{\decCtx{\ctx}{\x}}[{\ctxP[\redex]}]=\BlockLab{\dvs'\ \Dec{\EZ{\C}}{\x}{\val'}}{\ctx''[\BlockLab{\dvs_1\,\decs}{\e}{}]}{}$   or
  \item 
    $\HoleCtx{\decCtx{\ctx}{\x}}[{\ctxP[\redex]}]=\BlockLab{\dvs'\ \Dec{\EZ{\C}}{\x}{\val'}\ \Dec{\T}{\z}{\ctx''[\BlockLab{\dvs_1\,\decs}{\e}{\X}]}\ \decs'}{\e'}{}$,
\end{enumerate}
For (a),  by induction hypothesis we have that
\begin{quote}
$\congruence{\BlockLab{\dvs'\ \Dec{\EZ{\C}}{\x}{\val'}}{\ctx''[\BlockLab{\dvs_1\,\decs}{\e}{}]}{}}{\BlockLab{\dvs'\ \Dec{\EZ{\C}}{\x}{\val'}}{\BlockLab{\dvs'_1\,\decs'}{\e'}{}}{}}$
\end{quote} 
for some $\decs'$, $\e'$, and \PG{$\dvs'_1$ containing $\Dec{\EZ{\C_i}}{\y}{\val}$, which are the declarations that have been extruded}. 
Let $\decs'=\dvs_2\,\decs_2$ where $\decs_2$ are not evaluated declarations and let $\decs_2=\dvs'_2\,\decs'_2 $
where $\dvs'_2=\Reduct{\dvs_2}{\FV{\dvs'_1}}$.
Since we cannot have forward reference to unevaluated variables $\FV{\dvs'_2\,\dvs_1'}\cap\dom{\decs'_2 }=\emptyset$. Therefore we can apply rule \rn{Body}  of \refToFigure{congruence}.
Applying rule \rn{body}  of \refToFigure{congruence} we have that 
\begin{quote}
$\BlockLab{\dvs'\ \Dec{\EZ{\C}}{\x}{\val'}}{\BlockLab{\dvs_1\,\decs'}{\e'}{}}{}\cong\BlockLab{\dvs'\ \Dec{\EZ{\C}}{\x}{\val'}\ \dvs'_1\ \dvs'_2}{\BlockLab{\decs'_2}{\e'}{}}{}$. 
\end{quote}
For (b),  by induction hypothesis we have that
\begin{quote}
\begin{tabular}{l}
$\BlockLab{\dvs'\ \Dec{\EZ{\C}}{\x}{\val'}\ \Dec{\T}{\z}{\ctx''[\BlockLab{\dvs_1\,\decs}{\e}{\X}]}\ \decs'}{\e'}{}\cong$\\
\BigSpace$\BlockLab{\dvs'\ \Dec{\EZ{\C}}{\x}{\val'}\ \Dec{\T}{\z}{\BlockLab{\dvs'_1\,\decs'}{\e'}{\Y}}\ \decs''}{\e''}{}$
\end{tabular}
\end{quote} 
for some $\decs'$, $\decs''$, $\e'$, $\e''$, $\Y$ and \PG{$\dvs'_1$ containing $\Dec{\EZ{\C_i}}{\y}{\val}$, which are the declarations that have been extruded}. 
Let $\decs'=\dvs_2\,\decs_2$ where $\decs_2$ are not evaluated declarations and let $\decs_2=\dvs'_2\,\decs'_2 $
where $\dvs'_2=\Reduct{\dvs_2}{\FV{{\dvs'_1}}}$.
Since we cannot have forward reference to unevaluated variables $\FV{\dvs'_2\,\dvs_1'}\cap\dom{\decs'_2 }=\emptyset$. \\
From the fact that the term is well typed, we have that $\Y=\Closure{\resV}{\sharingRel}\cap(\dom{\dvs'_1}\cup\dom{\decs'})$ 
for some $\sharingRel$ such that $(\dom{\dvs'_2}\cup\dom{\dvs'_1})\subseteq\Closure{\y}{\sharingRel}$ and from \refToLemma{monotoneSharing}
and $\TypeCheck{\Gamma_{\ctx}}{{\FieldAssign{\x}{\f_i}{\y}}}{\C_i}{\{\x,\y,\resV\}}$ also $\x\in\Closure{\y}{\sharingRel}$.
Moreover, let $\Z=\dom{\dvs'_1}\cup\dom{\decs'}$, $\Remove{\sharingRel}{\Z}$ is the sharing relation associated to the inner block.\\
In order to apply congruence rule \rn{dec} we have to prove that if $\T=\Type{\mu}{\D}$ and $\mu=\capsule$, then $(\dom{\dvs'_2}\cup\dom{\dvs_1'})\cap\Y=\emptyset$. \\
If $\y\not\in\Closure{\resV}{\sharingRel}$, then $\Closure{\resV}{\sharingRel}\cap(\dom{\dvs'_2}\cup\dom{\dvs_1'})=\emptyset$.\\
If $\y\in\Closure{\resV}{\sharingRel}$, then, since $\x\in\Closure{\y}{\sharingRel}$ we have that  $\x\in\Closure{\resV}{\sharingRel}$.
Therefore $\Closure{\resV}{\sharingRel}\setminus\Z\not=\emptyset$ and $\mu$ cannot be $\capsule$. \\
Applying rule \rn{Dec} of
\refToFigure{congruence} we get 
\begin{quote}
\begin{tabular}{l}
$\BlockLab{\dvs'\ \Dec{\T}{\x}{\val'}\ \Dec{\T}{\z}{\BlockLab{\dvs'_1\,\decs'}{\e'}{\Y}}\ \decs''}{\e''}{}\cong$\\
\BigSpace$\BlockLab{\dvs'\ \Dec{\T}{\x}{\val'}\ \dvs'_1\ \dvs'_2\ \Dec{\T}{\z}{\BlockLab{\decs'_2}{\e'}{\Y'}}\ \decs''}{\e''}{}$.
\end{tabular}
\end{quote} 
Therefore 
$\congruence{\DecCtx{\ctx}{\x}{\CtxP{\FieldAssign{\x}{\f_i}{\y}}}}{\DecCtx{\ctx}{\x}{\ctx_3[\FieldAssign{\x}{\f_i}{\y}]}}$ 
for some $\ctx_3$ such that $\y\not\in\HB{\ctx_3}$.  So 
$\reduce{\ctx_1'[\HoleCtx{\decCtx{\ctx}{\x}}[{\ctx_3[{\FieldAssign{\x}{\f_i}{\y}}]}]]}{\e_2}$
applying rule \rn{field-assign}. 
\end{proof}


\section{Related work}\label{sect:related}

\emph{Capsule and lent notions.} As mentioned, the capsule property has many
variants in the literature, such as \emph{isolated} \cite{GordonEtAl12},
\emph{uniqueness} \cite{Boyland10} and \emph{external
uniqueness}~\cite{ClarkeWrigstad03}, \emph{balloon}
\cite{Almeida97,ServettoEtAl13a}, \emph{island} \cite{DietlEtAl07}. The fact
that aliasing can be controlled by using \emph{lent} (\emph{borrowed})
references is well-known~\cite{Boyland01,NadenEtAl12}.  However, before the
work in \cite{GordonEtAl12}, the capsule property was only detected in simple
situations, such as using a primitive deep clone operator, or composing
subexpressions with the same property.

The important novelty of the type system in \cite{GordonEtAl12} has been
\emph{recovery}, that is, the ability to detect properties (e.g., capsule or
immutability) by keeping into account not only the expression itself but the
way the surrounding context is used. In \cite{GordonEtAl12} an expression which
does not use external mutable references is recognized to be a capsule.
However, expressions which \emph{do} use  external mutable references, but do
not introduce sharing between them and the final result, are not recognized to
be capsules. For instance, \refToExample{One} and \refToExample{Five} in
\refToSection{examples} would be ill-typed in \cite{GordonEtAl12}. Our type
system generalizes recovery by precisely computing sharing effects.

\noindent 
\emph{Capabilities.}
In other proposals
{\cite{HallerOdersky10,HallerLoiko16,ClebschEtAl15,CastegrenWrigstad16,CastegrenW17}}
types are compositions of one or more \emph{capabilities}, and expose the union
of their operations. The modes of the capabilities in a type control how
resources of that type can be aliased. 
By using capabilities it is possible to obtain an expressivity which looks
similar to our type system, even though with different sharing notions and
syntactic constructs. For instance, the \emph{full encapsulation} notion in
\cite{HallerOdersky10}\footnote{{This paper includes a very good survey of work
in this area, notably explaining the difference between \emph{external
uniqueness}~\cite{ClarkeWrigstad03} and \emph{full encapsulation}.}}, apart
from the fact that sharing of immutable objects is not allowed, is equivalent
to the guarantee of our $\capsule$ modifier.  Their model has a higher
syntactic/logic overhead to explicitly  track regions.  As for all work
before~\cite{GordonEtAl12}, objects need to be born \Q|@unique| (that is, with the capsule property) and the type
system permits to manipulate data preserving their uniqueness. With
recovery~\cite{GordonEtAl12} (as with our type and effect system), instead, we can forget
about uniqueness, use normal code designed to work on conventional shared
data, and then recover the aliasing encapsulation property.

\noindent 
\emph{Ownership.}
A closed stream of research is that on \emph{ownership} (see an overview
in~\cite{ClarkeEtAl13}) which, however, offers an opposite approach. The ownership invariant, which can be
  expressed and proved, is, however, expected to be guaranteed by defensive   
  cloning.
In our approach, instead, the capsule concept models an efficient
\emph{ownership transfer}. In other words, when an object $\x$ is ``owned'' by
another object $\y$, it remains always true that $\y$ can only be accessed
through $\x$, whereas the capsule notion is more dynamic: a capsule can be
``opened'', that is, assigned to a standard reference and modified, and then we
can recover the original capsule guarantee.  Cloning, if needed, becomes a
responsibility of the client.  Other work in the literature supports ownership
transfer, see for example~\cite{MullerRudich07, ClarkeWrigstad03}.  In the
literature it is however applied to uniqueness/external uniqueness, thus not
{the whole} reachable object graph is transferred.

\noindent
\emph{Full, deep and shallow interpretation.}
The literature on sharing control distinguishes three interpretations for properties of objects.
  
\begin{myitemize}
\item Full: the whole reachable object graph respects that property.
\item Shallow: only the object itself respects that property.
\item Deep: the reachable object graph is divided in 2 parts:
The first part is the one that is logically ``owned'' by the
object itself, while the second part is just ``referenced''.
\end{myitemize}
In our approach, {as in  \cite{Almeida97,GianniniEtAl16,GordonEtAl12,
ServettoEtAl13a,ServettoZucca15}}, properties have the \emph{full}
interpretation, in the sense that they are propagated to the whole reachable
object graph. In a deep interpretation, instead, {as in
\cite{Boyland10,NadenEtAl12,Reynolds02}, it is possible, for instance, to
reach a mutable object from an immutable object.  In this sense, approaches
based on ownership, or where it is somehow possible to have any form of
``internal mutation'' are (only) deep, as in
\cite{CastegrenWrigstad16,Hogg91, BocchinoADAHKOSSV09,Turon17}.  This also
includes \cite{ClarkeWrigstad03}, where an unique object can point to
arbitrarily shared objects, if they do not, in turn, point back to the unique
object itself.}

The advantage of the full interpretation is that libraries can declare strong
intentions in a coherent and uniform way, independently of the concrete
representation of the user input (that, with the use of interfaces, could be
unknown to the library). On the other side, providing (only) full modifiers
means that we do not offer any language support for (as an example) an
immutable list of mutable objects.

\noindent
\emph{Destructive reads.}
Uniqueness can be enforced by destructive reads, i.e., assigning a copy of the
unique reference to a variable and destroying the original reference, see
\cite{GordonEtAl12,Boyland10}. Traditionally, borrowing/fractional
permissions~\cite{NadenEtAl12} are related to uniqueness  in the opposite way:
a unique reference can be borrowed, it is possible to track when all borrowed
aliases are buried~\cite{Boyland01}, and then uniqueness can be recovered.
These techniques offer a sophisticate alternative to destructive reads.  We
also wish to avoid destructive reads. In our work, we ensure uniqueness by
linearity, that is, by allowing at most one use of an $\capsule$ reference.

\noindent
\emph{Alias analysis.} 
Alias analysis is a fundamental static analysis, used in compilers and code
analysers. Algorithms such as Steensgaard's algorithm, \cite{Steensgaard96},
infer equivalence classes that may alias. In \cite{DeD12} is presented a
refined version of such algorithm, performing a uniqueness analysis, similar to
our detection of ``capsule'' values. However, the aim of our work is to design
a language in which annotations, such as the affine modifier, can be used by
the user to enforce properties of its code. Then the inference system checks
that such annotations are correctly used.

\noindent
\emph{Calculus.}
Finally, an important distinguishing feature of our work is that sharing can be
directly represented at the syntactic level as a relation among free variables,
thanks to the fact that the calculus is pure.  Models of the imperative
paradigm as pure calculi have been firstly proposed in
\cite{ServettoLindsay13,CapriccioliEtAl15}.


\section{Conclusion}\label{sect:conclu}
We have presented a type and effect system which \emph{infers} sharing possibly
introduced by the evaluation of an expression.  Sharing is directly represented
at the syntactic level as a relation among free variables. This representation
allows us to express in a natural way, and to generalize, widely-used notions
in literature, providing great expressivity, as shown by the examples of
\refToSection{examples}. 

We adopt a non standard operational model, where store is encoded directly
  in the language terms. In this model, sharing properties are directly
expressed at the syntactic level. Notably, in a subterm $\e$ of a program,
objects reachable from other parts of the program are simply those denoted by
free variables in $\e$, whereas local objects are those denoted by local
variables declared in $\e$. In our opinion, this offers a very intuitive and
simple understanding of the portion of memory only reachable from $\e$, since
it is encoded in $\e$ itself. Another advantage is that, the store being
encoded in terms, most proofs can be done by structural induction on
terms, as we have exploited in the Coq implementation mentioned
below.

On the other hand, a disadvantage of our model may be that it is
possibly more esoteric for people used to the other one. Moreover, since
isolation is encoded by scoping, some care is needed to avoid scope extrusion
during reduction. For this reason, reduction is defined on typechecked terms,
where blocks have been annotated with the information of which local variables
will be connected to the result. In this way, it is possible to define a rather
sophisticated notion of congruence among terms, which allows to move
declarations out of a block only when this preserves well-typedness.

Besides standard type soundness, we proved that uniqueness and lent properties
inferred by the type system actually ensure the expected behaviour. 

To focus on novelties and allow complete formal proofs, we illustrate our
type and effect system on a minimal language, only supporting the features
which are strictly necessary for significance. However, we expect that the
approach could be smoothly extended to typical constructs of
imperative/object-oriented languages, in the same way as other type systems
or type and effect systems (e.g., when effects are possibly thrown
exceptions \cite{AnconaLZ01}). We briefly discuss two key features below.
\begin{description}
  \item[Inheritance] For a method redefined in a heir class, the returned 
    sharing relation should be \emph{contained} in that of the parent, that is, 
    \emph{less} connections should be possibly caused by the method, 
    analogously to the requirement that possibly thrown exceptions should be a 
    subset (modulo subtyping) of those of the parent.
  \item[Control flow] As it is customary in type systems, control flow 
    constructs are handled by taking the ``best common approximation'' of the 
    types obtained by the different paths. For instance, in the case of 
    if-then-else we would get the \emph{union} of the sharing relations of 
    the two branches (the same happens for possibly thrown exceptions).
\end{description}

Note that including control flow constructs we would get, again as it is
customary in type systems, which are never complete for exactly this reason,
examples which are ill-typed but well-behaved at runtime. For instance, an
expression of shape:
\begin{lstlisting}
D y= new D(y);
C$^\capsule$ z= {
  D x= new D(x); 
  if (...) x.f = x else x.f = y; 
  x
};
z
\end{lstlisting}
will be ill-typed, since the sharing relation computed for the right-hand
side of the declaration of \lstinline{z} is $\{\texttt{y}, \resV\}$, that is, the
type and effect system conservatively assumes that this expression does
\emph{not} denote a capsule. However, if the guard of the conditional
evaluates to true, then the right-hand side of the declaration reduces to a
capsule.

We have implemented in Coq the type and effect system and the reduction rules. 
The current code can be found at
\texttt{//github.com/paola-giannini/sharing}.  We did not include the
construct of method call since, as we can see from the typing and reduction
rules, it can be translated into a block. To complete the implementation of the operational semantics, we need an oriented version of the congruence relation on terms to be applied before reduction steps, as explained at the end of \refToSection{calculus}. Then, we plan to mechanize the full proof of soundness. As mentioned before, we argue that the Coq implementation nicely illustrates the advantages of our purely syntactic model, since proofs can be carried out inductively, without requiring more complicated techniques such as, e.g., bisimulation. 

In further work, in addition to completing the formalization in Coq of the
proof of soundness, we will enrich the type system to also handle
\emph{immutable} references.  A preliminary presentation of this extension is in \cite{GianniniSZ18}.

We also plan to formally state and prove
behavioural equivalence of the calculus with the conventional imperative
model. Intuitively, the reachability information which, in our calculus, is
directly encoded in terms, should be reconstructable from the
dependencies among references in a conventional model with a flat global
memory. 

As a long term goal, it would be interesting to investigate (a form of)
Hoare logic on top of our model. We believe that the hierarchical structure of
our memory representation should help local reasoning, allowing specifications
and proofs to mention only the relevant portion, similarly to what is achieved
by separation logic \cite{Reynolds02}.

\EZ{\paragraph{Acknowledgement} 
We are indebted to the anonymous SCP referees for the thorough job they did reviewing our paper and for their valuable suggestions.  We also thank the OOPS'17 and FTfJP'17 referees for their helpful comments on preliminary versions. {Finally, we thank Isaac Oscar Gariano for his careful reading.}}

\section*{References}
\bibliography{main}

\newpage

\appendix \section{Type derivation for \refToExample{Five}}\label{app:derivation}
{In \refToFigure{TypingFive} we
give the type derivation that shows that the expression $\ett$ of
Example~\ref{ex:Five} is well-typed, where
$\Gamma_1=\TypeDec{\cU}{\Ctt},\TypeDec{\outC}{\Ctt}$,
$\Gamma_2=\TypeDec{\cD}{\Ctt},\TypeDec{\inC}{\Type{\capsule}{\Ctt}}$, and
$\Gamma_3=\TypeDec{\cT}{\Ctt},\TypeDec{\restt}{\Ctt}$ are the type contexts
corresponding to the top level, outer, and inner block, respectively.}

\begin{figure}[t]
{\small
\begin{center}
\begin{math}
\begin{array}{l}
\deriv_1:\Space{
\prooftree
{\prooftree
\begin{array}{c}
\TypeCheck{\Gamma}{\cU}{\Ctt}{\{\cU,\resV\}}\\
\TypeCheck{\Gamma}{\cD}{\Ctt}{\{\cD,\resV\}}
\end{array}
\justifies
\TypeCheck{\Gamma}{\MethCall{\cD}{\mix}{\cU}}{\Ctt}{\{\cU,\cD,\resV\}}
\endprooftree
}
\justifies
\TypeCheck{\Gamma}{{\MethCall{\MethCall{\cD}{\mix}{\cU}}{\clone}{}}}{\Ctt}{\{\cU,\cD\}}
\endprooftree
}\BigSpace
\deriv_2:\Space
{\prooftree
\begin{array}{c}
\TypeCheck{\Gamma}{\restt}{\Ctt}{\{\restt,\resV\}}\\
\TypeCheck{\Gamma}{\cT}{\Ctt}{\{\cT,\resV\}}
\end{array}\justifies
\TypeCheck{\Gamma}{{\MethCall{\restt}{\mix}{\cT}}}{\Ctt}{\{\cT,\restt,\resV\}}
\endprooftree
} \\ \\ \\
\deriv_3:\Space\prooftree
\begin{array}{c}
\prooftree
\TypeCheck{\Gamma}{\cT}{\Ctt}{\{\cT,\resV\}}
\justifies
\TypeCheck{\Gamma}{{\ConstrCall{\Ctt}{\cT}}}{\Ctt}{\{\cT,\resV\}}
\endprooftree
\BigSpace
\deriv_1
\BigSpace
\deriv_2
\end{array}
\justifies
\TypeCheck{{\SubstFun{\Gamma_1}{\Gamma_2}}}{\Block{\Dec{\Ctt}{\cT}{\ConstrCall{\Ctt}{\cT}}\,\Dec{\Ctt}{\restt}{\MethCall{\MethCall{\cD}{\mix}{\cU}}{\clone}{}}}{\MethCall{\restt}{\mix}{\cT}}}{\Ctt}{\{\cU,\cD\}}
\endprooftree
\\ \\ \\
\deriv_4:\Space\prooftree
\prooftree
\TypeCheck{\SubstFun{\Gamma_1}{\Gamma_2}}{\cD}{\Ctt}{\{\cD,\resV\}}
\justifies
\TypeCheck{\SubstFun{\Gamma_1}{\Gamma_2}}{{\ConstrCall{\Ctt}{\cD}}}{\Ctt}{\{\cD,\resV\}}
\endprooftree
\Space
\deriv_3
\Space
\prooftree
\begin{array}{c}
\TypeCheck{\SubstFun{\Gamma_1}{\Gamma_2}}{\inC}{\Ctt}{\epsilon}\\
\TypeCheck{{\SubstFun{\Gamma_1}{\Gamma_2}}}{\cD}{\Ctt}{\{\cD,\resV\}}\\[0.5ex]
\end{array}\justifies
\TypeCheck{\SubstFun{\Gamma_1}{\Gamma_2}}{\MethCall{\inC}{\mix}{\cD}}{\Ctt}\{\cD,\resV\}
\endprooftree
\justifies
\TypeCheck{\Gamma_1}{\Block{\Dec{\Ctt}{\cD}{\ConstrCall{\Ctt}{\cD}}\,\Dec{\Type{\capsule}{\Ctt}}{\inC}{\eU}}{\MethCall{\inC}{\mix}{\cD}}}{\Ctt}{\{\cU,\resV\}}
\endprooftree\\ \\ \\
\deriv:\Space\prooftree
\prooftree
\TypeCheck{\Gamma_1}{\cU}{\Ctt}{\{\cU,\resV\}}
\justifies
\TypeCheck{\Gamma_1}{\ConstrCall{\Ctt}{\cU}}{\Ctt}{\{\cU,\resV\}}
\endprooftree
\BigSpace
\deriv_4
\BigSpace
\TypeCheck{\Gamma_1}{\outC}{\Ctt}{\{\outC,\resV\}}
\justifies
\TypeCheck{}{\Block{\Dec{\Ctt}{\cU}{\ConstrCall{\Ctt}{\cU}}\,\Dec{\Ctt}{\outC}{\eD}}{\outC}}{\Ctt}{\epsilon}
\endprooftree\\ \\
\end{array}
\end{math}
\end{center}
}
\hrulefill
{\small
\begin{itemize}
\item $\deriv_3$ yields $\eUA=\BlockLab{\Dec{\Ctt}{\cT}{\ConstrCall{\Ctt}{\cT}}\,\Dec{\Ctt}{\restt}{\MethCall{\MethCall{\cD}{\mix}{\cU}}{\clone}{}}}{\MethCall{\restt}{\mix}{\cT}}{\{\restt,\cT\}}$
\item $\deriv_4$ yields $\eDA=\BlockLab{\Dec{\Ctt}{\cD}{\ConstrCall{\Ctt}{\cD}}\,\Dec{\Type{\capsule}{\Ctt}}{\inC}{\eU}}{\MethCall{\inC}{\mix}{\cD}}{\{\cD\}}$
\item $\deriv$ yields
${\ett'=}\BlockLab{\Dec{\Ctt}{\cU}{\ConstrCall{\Ctt}{\cU}}\,\Dec{\Ctt}{\outC}{\eD}}{\outC}{\{\outC\}}$
\end{itemize}
}
\caption{Type derivation for \refToExample{Five}}\label{fig:TypingFive}
\end{figure}

Derivations $\deriv_1$ and $\deriv_2$ end with an application of rule
\rn{T-Invk}. Consider $\deriv_2$.  The method $\mix$ produces sharing between
its receiver, parameter, and result.  Then the call of $\mix$ with receiver
$\restt$ and  parameter $\cT$ returns a result connected with both these
variables. The call of $\mix$ with receiver $\cD$ and  parameter $\cU$ in the
derivation $\deriv_1$ does the same with $\cD$ and $\cU$. Instead, the call of
$\clone$ does not produce any connection from its receiver and result. So the
call of $\clone$ with receiver $\MethCall{\cD}{\mix}{\cU}$ in the derivation
$\deriv_1$ does not cause connections for $\resV$. 

The type derivation $\deriv_3$ justifies the judgment
$\TypeCheck{\SubstFun{\Gamma_1}{\Gamma_2}}{\eU}{\Ctt}{\{\cU,\cD\}}$ where $\eU$
is the inner block (the initialization expression of $\inC$). {The effects of
the evaluation of initialization expressions and body are to mix the external
variables $\cU$ and $\cD$, and the local variables $\restt$ and $\cT$.
Moreover, the result is only connected with the two local variables.  Hence,
the inner block denotes a capsule, so it can be used as initialization
expression of the affine variable $\inC$.} The sharing relation resulting from
the evaluation of the declarations and the body (before removing the local
variables $\restt$ and $\cT$) is represented by
$\{\cU,\cD\}\,\{\cT,\restt,\resV\}$.\label{sequence-two}
The annotation for the block, i.e., the set of local variables connected to the
result, is $\{\restt,\cT\}$, so, when applying the congruence relation, the
variables $\restt$ and $\cT$ cannot be moved outside this block.

The type derivation $\deriv_4$ justifies the judgment
{$\TypeCheck{{\Gamma_1}}{\eD}{\Ctt}{\{\cU,\resV\}}$,  where $\eD$ is the outer
block (the initialization expression of $\outC$). The effects of the evaluation
of initialization expressions and body are to mix the external variable $\cU$
with the local variable $\cD$ (this effect propagates from the inner block).
Moreover, the result is connected with variable $\cD$. Hence, the result turns
out to be connected with the external variable $\cU$ as well.  Therefore, this
block is not a capsule, and could not be used to initialize an affine variable.
Note that the variable $\inC$, being affine, is not in the domain of the
sharing relation.  Indeed, it will be substituted with the result of the
evaluation of $\eU$ and so it will disappear.}

Finally, $\deriv$ is the derivation for the expression $\ett$. The block is a
closed expression, and closed expressions are capsules. The block is annotated
with the local variable $\outC$, which is connected with its result.

\section{Proofs}\label{app:proofs}
\noindent
{\bf Proposition \ref{prop:value}.}
{\it If $\val$ is a value, then there exists $\val'$ such that $\congruence{\val}{\val'}$ and $\val'$ is in canonical form.}
\begin{proof}
By structural induction on values.\\
\underline{Consider $\ConstrCall{\C}{\vs}$}. 
By inductive hypotheses on $\vs$ we
have that $\congruence{\ConstrCall{\C}{\vs}}{\ConstrCall{\C}{\vs'}}$ for some 
$\vs'$ in canonical form.\\
\PG{By induction on the number $n$ of $\val\in\vs'$ such that $\val$ is not a 
variable, i.e., is a block value, we show that 
$\congruence{\ConstrCall{\C}{\vs'}}{{\BlockLab{\dvs}{\ConstrCall{\C}{\xs}}{\X}}}$
for some $\dvs$, $\X$ and $xs$.\\
If $n=0$, then
$\congruence{\ConstrCall{\C}{\vs'}}{\BlockLab{\Dec{\C}{\x}{\ConstrCall{\C}{\vs'}}}{\x}{{\{\x\}}}}$
where $\x$ is a fresh variable using congruence rule \rn{new}. \\
If $n>0$, then
$\ConstrCall{\C}{\vs'}=\ConstrCall{\C}{\xs,\BlockLab{\dvs}{\x}{\X},vs''}$ where 
$\dvs\neq\epsilon$ is in canonical form. We may assume,  by $\alpha$-renaming, that 
$(\FV{\vs''}\cup\{\xs\})\cap\dom{\dvs}=\emptyset$. Using rule \rn{val-ctx} we
have that $\congruence{\ConstrCall{\C}{\vs'}}{\BlockLab{\dvs}{\val'}{\X}}$ where
$\val'=\ConstrCall{\C}{\xs,\x,\vs''}$. Since $\val'$ has $n-1$ arguments which
are not variables by inductive hypothesis we have that 
$\congruence{\ConstrCall{\C}{\xs,\x,\vs''}}{\BlockLab{\dvs'}{\ConstrCall{\C}{\xs,\x,\ys}}{\Y}}$.
Therefore $\congruence{\ConstrCall{\C}{\vs'}}{\BlockLab{\dvs}{\BlockLab{\dvs'}{\ConstrCall{\C}{\xs,\x,\ys}}{\Y}}{\X}}$.\\
We may assume,  by $\alpha$-renaming, that 
$(\FV{\dvs}\cup\dom{\dvs})\cap\dom{\dvs'}=\emptyset$. 
Using congruence rule \rn{body}, we get
$\congruence{\ConstrCall{\C}{\vs'}}
{\BlockLab{\dvs\ \dvs'}{\ConstrCall{\C}{\xs,\x,\ys}}{\X\cup\Y}}$.
Applying congruence rule \rn{new} followed by \rn{body} we have that 
$\congruence{\ConstrCall{\C}{\vs'}}
{\BlockLab{\dvs\ \dvs'\ \Dec{\C}{\z}{\ConstrCall{\C}{\xs,\x,\ys}}}{\z}{\X\cup\Y\cup\{\z\}}}$  with $\z$ a fresh variable.\\
\underline{Consider $\BlockLab{\dvs}{\val}{\X}$}. 
By induction hypothesis on values $\congruence{\BlockLab{\dvs}{\val}{\X}}{\BlockLab{\dvs}{\val'}{\X}}$
with $\val'$ in canonical form.\\
By induction on the number $n$ of declarations $\Dec{\C}{\x}{\ConstrCall{\C}{\vs}}\in\dvs$ which are not in canonical
form we show that $\congruence{\BlockLab{\dvs}{\val'}{\X}}{\BlockLab{\dvs'}{\val'}{\Y}}$ where all the $\dvs'$
are in canonical form. \\
If $n=0$, then $\dvs$ is in canonical form. We have two cases $\val'=\x$ or $\val'=\BlockLab{\dvs'}{\y}{\Y}$
with $\dvs'$ in canonical form.
In the first case $\congruence{\BlockLab{\dvs}{\val}{\X}}{\BlockLab{\dvs}{\x}{\X}}$ and we are done.
In the second case $\congruence{\BlockLab{\dvs}{\val}{\X}}{\BlockLab{\dvs}{\BlockLab{\dvs'}{\y}{\Y}}{\X}}$. 
 We may assume,  by $\alpha$-renaming, that $\dom{\dvs'}\cap(\FV{dvs}\cup\dom{\dvs})=\emptyset$,
  so congruence rule \rn{body} can be applied to obtain
  $\congruence{\BlockLab{\dvs}{\val}{\X}}{\BlockLab{\dvs\,\dvs'}{\BlockLab{}{\y}{\emptyset}}{\X\cup\Y}}$
  and with rule \rn{{block-elim}} we obtain
  $\congruence{\BlockLab{\dvs}{\val}{\X}}{\BlockLab{\dvs\,\dvs'}{{\y}}{\X\cup\Y}}$.\\
If $n>0$, then  We may assume that
  $\dvs=\dvs_1\,\Dec{\C}{\z}{\ConstrCall{\C}{\vs}}\,\dvs_2$ with $\dvs_1$ in
  canonical form, the values in $\vs$ in canonical form, and some $\val\in\vs'$
  not a variable. As for the proof for constructor values we can prove that
  $\congruence{\ConstrCall{\C}{\vs}}{\BlockLab{\dvs'}{\ConstrCall{\C}{\ys}}{\Y}}$
  for some $\Y$, $\ys$ and $\dvs'$ in canonical form. (We just omit the application of the
  congruence rule \rn{new} from the proof.) Therefore
  $\congruence{\BlockLab{\dvs}{\val}{\X}}{\BlockLab{\dvs_1\,\Dec{\C}{\z}{\BlockLab{\dvs'}{\ConstrCall{\C}{\ys}}{\Y}}\,\dvs_2}{\val}{\X}}$.
  We may assume,  by $\alpha$-renaming, that
  $\dom{\dvs'}\cap(\FV{\dvs_1\,\dvs_2\,\val}\cup\dom{\dvs_1\,\dvs_2})=\emptyset$.
  Therefore, by applying congruence rule \rn{dec} and  \rn{block-elim} we get
  \begin{center}
    $\congruence{\BlockLab{\dvs_1\,\Dec{\C}{\z}{\BlockLab{\dvs'}{\ConstrCall{\C}{\vs}}{\Y}}\,\dvs_2}{\val}{\X}}{\val'}$
  \end{center} where
  $\val'=\BlockLab{\dvs_1\,\dvs'\,\Dec{\C}{\z}{\ConstrCall{\C}{\ys}}\,\dvs_2}{\val}{\X}$.
  Since $\dvs'$ is in canonical form, the number of declarations which are not
  in canonical form in $\val'$ is $n-1$, hence the thesis holds by the inductive
  hypothesis.
  }
\end{proof}

\PG{The following lemma shows that scope extrusion preserves types. In particular, annotations on blocks 
associated to capsule variables prevent extrusion of declarations that may be connected to the result
of the block. The lemma is the main result needed to prove that congruence preserves typability.}
\begin{lemma}\label{lemma:extrusion}
Let  
\begin{enumerate}[(i)]
\item $\dec=\Dec{\Type{\mu}{\C}}{\x}{\BlockLab{\dvs_1\,\decs_2}{\e}{\X}}$ and
\item $\decs=\dvs_1\,\Dec{\Type{\mu}{\C}}{\x}{\BlockLab{\decs_2}{\e}{\Y}}$ and
\item $\dom{\decs_2}\cap\FV{\dvs_1}=\emptyset$ and 
\item if $\mu=\capsule$, then $D_1\cap\X=\emptyset$.
\end{enumerate}
$\TypeCheckDecs{{\Gamma}[{\x{:}\Type{\mu}{\C}}]}{\dec}{\sharingRel}$ if and only if
$\TypeCheckDecs{{\Gamma}[{\x{:}\Type{\mu}{\C}},\Gamma_{\dvs_1}]}{\decs}{\sharingRel'}$ where $\sharingRel=\Remove{\sharingRel'}{D_1}$.
\end{lemma}
\begin{proof}
Let $D_1=D_1$ and $D_2=\dom{\decs_2}$. \\
\underline{We first show the ``only if'' implication}.\\
Let $\TypeCheckDecs{\Gamma[{\x{:}\Type{\mu}{\C}}]}{\dec}{\sharingRel}$ and
$\Gamma_1={\Gamma}[{\x{:}\Type{\mu}{\C}}][\Gamma_{\dvs_1},\Gamma_{\decs_2}]$. From \refToLemma{invBlock} we have that
 \begin{enumerate} [(1)]
     \item  $\sharingRel=\SubstEqRel{\sharingRel_x}{\x}{\resV}$ where $\sharingRel_x=\Remove{(\sharingRel_1+\sharingRel_2+\sharingRel_e)}{(D_1\cup D_2)}$
       \item $\TypeCheckDecs{\Gamma_1}{\dvs_1}{\sharingRel_1}$ and $\TypeCheckDecs{\Gamma_1}{\decs_2}{\sharingRel_2}$
       and $\TypeCheck{\Gamma_1}{\e}{\C}{\sharingRel_e}$ 
      \item $\X=\Closure{\resV}{(\sharingRel_1+\sharingRel_2+\sharingRel_e)}\cap(D_1\cup D_2)$
     \item if $\mu=\capsule$, then $\IsCapsule{\sharingRel_x}$
 \end{enumerate}
 By $\alpha$-renaming
we may assume that $\x\not\in(D_1\cup D_2)$ so let $\Gamma'_1={\Gamma}[{\x{:}\Type{\mu}{\C}},\Gamma_{\dvs_1}][\Gamma_{\decs_2}]$ we have $\Gamma_1=\Gamma'_1$. From (2) we derive
\begin{enumerate} [(a)]
       \item $\TypeCheckDecs{\Gamma'_1}{\decs_2}{\sharingRel_2}$
      and $\TypeCheck{\Gamma'_1}{\e}{\C}{\sharingRel_e}$ 
\end{enumerate}
and applying rule \rn{T-block} to (a)
\begin{enumerate} [(a)]\addtocounter{enumi}{1}
       \item $\TypeCheck{{\Gamma}[{\x{:}\Type{\mu}{\C}},\Gamma_{\dvs_1}]}{\BlockLab{\decs_2}{\e}{\Y}}{\C}{\sharingRel'_x}$ where
       \item $\sharingRel'_x=\Remove{(\sharingRel_2+\sharingRel_e)}{D_2}$ 
       \item $\Y=\Closure{\resV}{{(\sharingRel_2+\sharingRel_e)}}\cap{{D_2}}$
\end{enumerate}
If $\mu=\capsule$, from (iv) we get ${D_1}\cap(\Closure{\resV}{(\sharingRel_1+\sharingRel_2+\sharingRel_e)}\cap(D_1\cup D_2))=\emptyset$.
Therefore, $\Closure{\resV}{\sharingRel_x}=\Closure{\resV}{(\sharingRel_1+\sharingRel_2+\sharingRel_e)}\setminus(D_1\cup D_2)=\Closure{\resV}{(\sharingRel_1+\sharingRel_2+\sharingRel_e)}\setminus{D_2}$. From (d) 
$\Closure{\resV}{\sharingRel'_x}=\Closure{\resV}{(\sharingRel_2+\sharingRel_e)}\setminus{D_2}\subseteq\Closure{\resV}{(\sharingRel_1+\sharingRel_2+\sharingRel_e)}\setminus{D_2}$. Therefore from (4), $\IsCapsule{\sharingRel_x}$ and we get
that $\IsCapsule{\sharingRel'_x}$. From (b)
\begin{enumerate} [(a)]\addtocounter{enumi}{4}
\item $\TypeCheckDecs{{\Gamma}[{\x{:}\Type{\mu}{\C}},\Gamma_{\dvs_1}]}{\Dec{\Type{\mu}{\C}}{\x}{\BlockLab{\decs_2}{\e}{\Y}}}{\SubstEqRel{\sharingRel'_x}{\x}{\resV}}$ 
\end{enumerate}
From (iii) and \refToLemma{weakening}.2 and (2)
\begin{enumerate} [(a)]\addtocounter{enumi}{5}
       \item $\TypeCheckDecs{{\Gamma}[{\x{:}\Type{\mu}{\C}},\Gamma_{\dvs_1}]}{\dvs_1}{\sharingRel_1}$
\end{enumerate}
Therefore\\
 \centerline{$\TypeCheckDecs{{\Gamma}[{\x{:}\Type{\mu}{\C}},\Gamma_{\dvs_1}]}{\dvs_1\,\Dec{\Type{\mu}{\C}}{\x}{\BlockLab{\decs_2}{\e}{\X\setminus{D_1}}}}{\sharingRel'=\sharingRel_1+\SubstEqRel{\sharingRel'_x}{\x}{\resV}}$.}\\
 Let $\sharingRel''=\sharingRel_{2}+\sharingRel_e$. From (iii) and \refToProp{invTyping1} we have that 
 $\Remove{\sharingRel_1}{D_2}=\sharingRel_1$ and so by \refToProp{lessSrRel}.\ref{p3}
 \begin{enumerate} [(a)]\addtocounter{enumi}{6}
       \item $\Remove{(\sharingRel_1+\sharingRel'')}{D_2}=\Remove{\sharingRel_1}{D_2}+\Remove{\sharingRel''}{D_2}=\sharingRel_1+\Remove{\sharingRel''}{D_2}$
\end{enumerate}
Therefore\\
 \centerline{$
 \begin{array}{lcll}
 \sharingRel&=& \SubstEqRel{(\Remove{(\sharingRel_1+\sharingRel'')}{(D_1\cup D_2)})}{\x}{\resV}
\\
 &=& \SubstEqRel{(\Remove{(\Remove{(\sharingRel_1+\sharingRel'')}{D_2})}{D_1})}{\x}{\resV}& \\
&=& \SubstEqRel{(\Remove{(\sharingRel_1+\Remove{\sharingRel''}{D_2})}{D_1})}{\x}{\resV}&\text{by (g)} \\
&=& \Remove{(\SubstEqRel{(\sharingRel_1+\Remove{\sharingRel''}{D_2})}{\x}{\resV})}{D_1}&\text{since $\{\x,\resV\}\cap D_1=\emptyset$} \\
&=& \Remove{\sharingRel_1+(\SubstEqRel{(\Remove{\sharingRel''}{D_2})}{\x}{\resV})}
 {D_1}&\text{since $\Closure{\resV}{\sharingRel_1}=\{\resV\}$ and} \\
 & & & 
 \quad\Closure{\x}{\sharingRel_1}=\{\x\} \\
 &=& \Remove{\sharingRel_1+(\SubstEqRel{\sharingRel'_x}{\x}{\resV})}
 {D_1}&\text{since $\sharingRel'_x=\Remove{\sharingRel''}{\D_2}$} \\
 &=& \Remove{\sharingRel'}{D_1}
  \end{array}
 $}

 \underline{We now show the ``if'' implication}.\\
Let $\TypeCheckDecs{\Gamma[{\x{:}\Type{\mu}{\C}},\Gamma_{\dvs_1}]}{\decs}{\sharingRel'}$  and 
$\Gamma_1={{\Gamma}[{\x{:}\Type{\mu}{\C}},\Gamma_{\dvs_1}]}[\Gamma_{\decs_2}]$. From \refToLemma{invBlock} we have that
 \begin{enumerate} [(1)]
     \item  $\sharingRel'=\sharingRel_1+\SubstEqRel{\sharingRel_x}{\x}{\resV}$ where $\sharingRel_x=\Remove{(\sharingRel_2+\sharingRel_e)}{D_2}$
     \item $\TypeCheckDecs{{\Gamma}[{\x{:}\Type{\mu}{\C}},\Gamma_{\dvs_1}]}{\dvs_1}{\sharingRel_1}$ 
   \item $\TypeCheckDecs{\Gamma_1}{\decs_2}{\sharingRel_2}$
       and $\TypeCheck{\Gamma_1}{\e}{\C}{\sharingRel_e}$ 
      \item $\Y=\Closure{\resV}{(\sharingRel_2+\sharingRel_e)}\cap D_2$
     \item if $\mu=\capsule$, then $\IsCapsule{\sharingRel_x}$
 \end{enumerate}
 By $\alpha$-renaming
we may assume that $\x\not\in D_2$ so let $\Gamma'_1={\Gamma}[{\x{:}\Type{\mu}{\C}}][\Gamma_{\dvs_1},\Gamma_{\decs_2}]$ we have $\Gamma_1=\Gamma'_1$. From (3) we derive
\begin{enumerate} [(a)]
       \item $\TypeCheckDecs{\Gamma'_1}{\decs_2}{\sharingRel_2}$
      and $\TypeCheck{\Gamma'_1}{\e}{\C}{\sharingRel_e}$ 
\end{enumerate}
and, from (iii) and \refToLemma{weakening}.2 also
\begin{enumerate} [(a)]\addtocounter{enumi}{1}
       \item $\TypeCheckDecs{\Gamma'_1}{\dvs_1}{\sharingRel_1}$
\end{enumerate}
Applying rule \rn{T-block} to (a) and (b) we have
\begin{enumerate} [(a)]\addtocounter{enumi}{2}
       \item $\TypeCheck{{\Gamma}[{\x{:}\Type{\mu}{\C}}]}{\BlockLab{\dvs_1\,\decs_2}{\e}{\X}}{\C}{\sharingRel'_x}$ where
       \item $\sharingRel'_x=\Remove{(\sharingRel_1+\sharingRel_2+\sharingRel_e)}{(D_1\cup D_2)}$ and
       \item $\X=\Closure{\resV}{(\sharingRel_1+\sharingRel_2+\sharingRel_e)}\cap(D_1\cup D_2)$
\end{enumerate}
\PG{From (iii) and  \refToProp{invTyping1}  there are no $\Pair{\y}{\y'}\in\sharingRel_2$ 
with $\y\neq\y'$ such that either $\y\in{D_1}$ or $\y'\in{D_1}$. 
Therefore $\Pair{\z}{\resV}\in(\sharingRel_1+\sharingRel_2+\sharingRel_e)\setminus(D_1\cup D_2)$
and $\z\in{D_2}$ implies $\Pair{\z}{\resV}\in(\sharingRel_2+\sharingRel_e)\setminus{D_1}$.
If $\mu=\capsule$, then from (5) we have that $\IsCapsule{\sharingRel'_x}$.}
From  \refToDef{typeBlock}  we derive
\begin{enumerate} [(a)]\addtocounter{enumi}{5}
       \item $\TypeCheckDecs{{\Gamma}[{\x{:}\Type{\mu}{\C}}]}{\Dec{\Type{\mu}{\C}}{\x}{\BlockLab{\dvs_1\,\decs_2}{\e}{\X}}}{\SubstEqRel{\sharingRel'_x}{\x}{\resV}}$
\end{enumerate}
As for the ``ony if'' proof we can show that $\sharingRel=\Remove{\sharingRel'}{D_1}$.\end{proof}

\medskip

\noindent{\bf Lemma \ref{lemma:congruence}.} (Congruence preserves types)
{\it Let $\e_1$ and $\e_2$ be annotated expressions. 
If $\TypeCheck{\Gamma}{\e_1}{\C}{\sharingRel}$
and $\congruence{\e_1}{\e_2}$, then $\TypeCheck{\Gamma}{\e'_2}{\C}{\sharingRel}$ for some $\e'_2$ such that $\e_2'\variant\e_2$.
}
\begin{proof} 
  \PG{By cases on the congruence rule used. We do the 
  case for rules \rn{dec} and \rn{val-ctx} with  $\terminale{new}$ that are the most significative and show how scope extrusion preserves 
  typing. The other cases are similar and simpler. In both cases we first show that typability of the left-side of $\congruence{}{}$
  implies typability of the right-side with the same type and sharing relation and  then the viceversa.} 

\underline{Rule \rn{{dec}}}. Let
\begin{itemize} 
\item $\e_1=\BlockLab{\dvs\ \Dec{\Type{\mu}{\C}}{\x}{\BlockLab{\dvs_1\ \decs_2}{\e}{\X}}\ \decs'}{\e'}{\Y}$ and
\item$\e_2=\BlockLab{\dvs\ \dvs_1\ \Dec{\Type{\mu}{\C}}{\x}{\BlockLab{\decs_2}{\e}{\X'}\ \decs'}}{\e'}{\Y'}$.
\end{itemize} 
where
\begin{enumerate} [(1)]
\item $\FV{\dvs_1}\cap\dom{\decs_2}=\emptyset$ 
\item $\FV{\dvs\,\decs'\,\e'}\cap\dom{\dvs_1}=\emptyset$ 
\item $\mbox{if }\mu=\capsule\mbox{ then }\dom{\dvs_1}\cap\X=\emptyset$.
\end{enumerate}
Let $D_1=\dom{\dvs_1}$.\\
We show that \underline{$\TypeCheck{\Gamma}{\e_1}{\C}{\sharingRel}$ implies $\TypeCheck{\Gamma}{\e_2}{\C}{\sharingRel}$} 
for some $\X'$ and $\Y'$.\\
Let $\TypeCheck{\Gamma}{\e_1}{\C}{\sharingRel}$, define $\Gamma_1=\SubstFun{\Gamma}{\Gamma_{\dvs},\x{:}\Type{\mu}{\C},\Gamma_{\decs'}}$, 
and $\Z_d=\dom{\dvs}\cup\dom{\decs'}\cup\{\x\}$. From \refToLemma{invBlock} we have that
\begin{enumerate}[(a)]
\item $\sharingRel=\Remove{(\sharingRel_d+\sharingRel_x+\sharingRel_e)}{\Z_d}$ 
\item $\Y=\Closure{\resV}{(\sharingRel_d+\sharingRel_x+\sharingRel_e)}\cap\Z_d$
  \item $\TypeCheck{\Gamma_1}{\e'}{\C}{\sharingRel_e}$ and $\TypeCheckDecs{\Gamma_1}{\dvs\ \decs'}{\sharingRel_d}$
 \item $\TypeCheckDecs{\Gamma_1}{\Dec{\Type{\mu}{\C}}{\x}{\BlockLab{\dvs_1\ \decs_2}{\e}{\X}}}{\sharingRel_x}$
\end{enumerate}
From (d), (1), (3) and \refToLemma{extrusion}, letting $\Gamma'_1=\SubstFun{\Gamma}{\Gamma_{\dvs},\x{:}\Type{\mu}{\C}, \Gamma_{\dvs_1},\Gamma_{\decs'}}$, we have that
\begin{enumerate}[(a)]\addtocounter{enumi}{4}
 \item $\TypeCheckDecs{\Gamma'_1}{\dvs_1\ \Dec{\Type{\mu}{\C}}{\x}{\BlockLab{\decs_2}{\e}{\X\setminus{D_1}}}}{\sharingRel'_x}$  where $\sharingRel_x=\Remove{\sharingRel'_x}{D_1}$
\end{enumerate}
From (2), (c) and \refToLemma{weakening}.1 we have 
\begin{enumerate}[(a)]\addtocounter{enumi}{5}
  \item $\TypeCheck{\Gamma'_1}{\e'}{\C}{\sharingRel_e}$ and $\TypeCheckDecs{\Gamma'_1}{\dvs\ \decs'}{\sharingRel_d}$
\end{enumerate}
\begin{enumerate}[(a)]\addtocounter{enumi}{5}
  \item $\TypeCheck{\Gamma'_1}{\e'}{\C}{\sharingRel_e}$ and $\TypeCheckDecs{\Gamma'_1}{\dvs\ \decs'}{\sharingRel_d}$
\end{enumerate}
From (e), (f) and rule \rn{T-Block}\\
\centerline{$
\TypeCheck{\Gamma}{\BlockLab{\dvs\ \dvs_1\ \Dec{\Type{\mu}{\C}}{\x}{\BlockLab{\decs_2}{\e}{\X\setminus{D_1}}}\ \decs'}{\e'}{\Y'}}{\C}{\sharingRel'}
$}\\
where $\sharingRel'=\Remove{(\sharingRel_d+\sharingRel'_x+\sharingRel_e)}{(\Z_d\cup{D_1})}$
and $\Y'=\Closure{\resV}{(\sharingRel_d+\sharingRel'_x+\sharingRel_e)}\cap(\Z_d\cup{D_1})$. \\
From (2), \refToProp{invTyping1} and  \refToProp{lessSrRel}.\ref{p4} we have that 
\begin{enumerate}[(a)]\addtocounter{enumi}{6}
  \item $\Remove{(\sharingRel'_x+\sharingRel_d+\sharingRel_e)}{D_1}=\sharingRel_x+\sharingRel_d+\sharingRel_e$ 
  \end{enumerate}
Therefore\\
\centerline{$
\begin{array}{lcll}
\sharingRel'&=& \Remove{(\Remove{(\sharingRel'_x+\sharingRel_d+\sharingRel_e)}{D_1})}{\Z_d} &\text{by definition of $\setminus$} \\
 &=& \Remove{(\sharingRel_x+\sharingRel_d+\sharingRel_e)}{\Z_d} &\text{from (g)} \\
&=&\sharingRel
 \end{array}
$}\\

\medskip\noindent
We show that \underline{$\TypeCheck{\Gamma}{\e_2}{\C}{\sharingRel}$ implies $\TypeCheck{\Gamma}{\e_1}{\C}{\sharingRel}$} for some $\X$ and $\Y$.\\
Let $\TypeCheck{\Gamma}{\e_2}{\C}{\sharingRel}$, define $\Gamma_1=\SubstFun{\Gamma}{\Gamma_{\dvs},\Gamma_{\dvs_1},\x{:}\Type{\mu}{\C},\Gamma_{\decs'}}$, 
and $\Z_d=\dom{\dvs}\cup\dom{\decs'}\cup\{\x\}$. From \refToLemma{invBlock} we have that
\begin{enumerate}[(a)]
\item $\sharingRel=\Remove{(\sharingRel_d+\sharingRel_{1}+\sharingRel_x+\sharingRel_e)}{(\Z_d\cup{D_1})}$ 
\item $\Y'=\Closure{\resV}{(\sharingRel_d+\sharingRel_{1}+\sharingRel_x+\sharingRel_e)}\cap(\Z_d\cup{D_1})$
  \item $\TypeCheck{\Gamma_1}{\e'}{\C}{\sharingRel_e}$ and $\TypeCheckDecs{\Gamma_1}{\dvs\ \decs'}{\sharingRel_d}$ 
  \item $\TypeCheckDecs{\Gamma_1}{\dvs_1}{\sharingRel_1}$
 \item $\TypeCheckDecs{\Gamma_1}{\Dec{\Type{\mu}{\C}}{\x}{\BlockLab{\decs_2}{\e}{\X'}}}{\sharingRel''_x}$
\end{enumerate}
From (d), (e) and \refToDef{typeBlock}, letting $\sharingRel'_x=\sharingRel_1+\sharingRel''_x$, we have that $\TypeCheckDecs{\Gamma_1}{\dvs_1\, \Dec{\Type{\mu}{\C}}{\x}{\BlockLab{\decs_2}{\e}{\X'}}}{\sharingRel'_x}$. Therefore from \refToLemma{extrusion}, 
letting $\Gamma'_1=\SubstFun{\Gamma}{\Gamma_{\dvs},\x{:}\Type{\mu}{\C},\Gamma_{\decs'}}$, we have that
\begin{enumerate}[(a)]\addtocounter{enumi}{5}
 \item $\TypeCheckDecs{\Gamma'_1}{\Dec{\Type{\mu}{\C}}{\x}{\BlockLab{\dvs_1\ \decs_2}{\e}{\X}}}{\sharingRel_x}$  where $\sharingRel_x=\Remove{\sharingRel'_x}{D_1}$ and $\Remove{\X}{D_1}=\X'$.
\end{enumerate}
From (2), (c) and \refToLemma{weakening}.2 $\TypeCheck{\Gamma'_1}{\e'}{\C}{\sharingRel_e}$ and $\TypeCheckDecs{\Gamma'_1}{\dvs\ \decs'}{\sharingRel_d}$. From rule \rn{T-block} we have that\\
\centerline{$
\TypeCheck{\Gamma}{\BlockLab{\dvs\ \Dec{\Type{\mu}{\C}}{\x}{\BlockLab{\dvs_1\ \decs_2}{\e}{\X}}\ \decs'}{\e'}{\Y}}{\C}{\sharingRel'}
$}\\
where, $\Y=\Closure{\resV}{\sharingRel_x+\sharingRel_e+\sharingRel_d}\cap\Z_d$ and $\sharingRel'=\Remove{\sharingRel_x+\sharingRel_e+\sharingRel_d}{\Z_d}$. The prof that $\sharingRel'=\sharingRel$ is as for the previous 
implication.

\medskip
\underline{Rule \rn{val-ctx}} with $\terminale{new}$. Let
\begin{itemize} 
\item $\e_1=\ConstrCall{\C}{\val_1,\dots,\val_n,\BlockLab{\dvs_1\ \dvs_2}{\val}{\X},\val_{n+1},\dots,\val_{n+m}}$ and
\item$\e_2=\BlockLab{\dvs_1}{\ConstrCall{\C}{\val_1,\dots,\val_n,\BlockLab{\dvs_2}{\val}{\X'},\val_{n+1},\dots,\val_{n+m}}}{\Y}$.
\end{itemize} 
where
\begin{enumerate} [(1)]
\item $\FV{\dvs_1}\cap\dom{\dvs_2}=\emptyset$ 
\item $\FV{\vs\,\vs'}\cap\dom{\dvs_1}=\emptyset$ 
\end{enumerate}
Let $D_1=\dom{\dvs_1}$ and $D_2=\dom{\dvs_2}$.\\
We show that \underline{$\TypeCheck{\Gamma}{\e_1}{\C}{\sharingRel}$ implies $\TypeCheck{\Gamma}{\e_2}{\C}{\sharingRel}$} for some $\X'$ and $\Y$.\\
Let $\TypeCheck{\Gamma}{\e_1}{\C}{\sharingRel}$, define $\Gamma_1=\SubstFun{\Gamma}{\Gamma_{\dvs_1},\Gamma_{\dvs_2}}$. From rule \rn{T-new}  \refToLemma{invBlock} we have that
\begin{enumerate}[(a)]
\item $\sharingRel=\sharingRel_{\val}+\Remove{(\sharingRel'_1+\sharingRel'_2+\sharingRel')}{(D_1\cup D_2)}$ where $\sharingRel_{\val}=\sum\limits_{i=1}^{n+m}\sharingRel_i$
  \item $\TypeCheck{\Gamma}{\val_i}{\T_i}{\sharingRel_i}$ for all $i$, $1\leq i\leq n+m$
  \item $\TypeCheckDecs{\Gamma'}{\dvs_1}{\sharingRel'_1}$  
  \item $\TypeCheckDecs{\Gamma'}{\dvs_2}{\sharingRel'_2}$ and  $\TypeCheck{\Gamma'}{\val}{\T'}{\sharingRel'}$  and
 \item 
 $\X=\Closure{\resV}{(\sharingRel'_1+\sharingRel'_2+\sharingRel')}\cap(D_1\cup D_2)$
\end{enumerate}
where $\fields{\C}=\Field{\T_1}{\f_1}\dots\Field{\T_n}{\f_n}\Field{\T'}{\f'}\Field{\T_{n+1}}{\f_{n+1}}\dots\Field{\T_{n+m}}{\f_{n+m}}$.\\
By wellformedness of blocks $D_1\cap D_2=\emptyset$. Therefore $\Gamma_1={\Gamma}[{\Gamma_{\dvs_1}][\Gamma_{\dvs_2}}]$.
From \refToLemma{weakening}.2 and (d) we get $\TypeCheckDecs{\SubstFun{\Gamma}{\Gamma_{\dvs_2}}}{\dvs_2}{\sharingRel'_2}$ and  $\TypeCheck{\SubstFun{\Gamma}{\Gamma_{\dvs_2}}}{\val}{\T'}{\sharingRel'}$ and by rule \rn{T-block}
\begin{enumerate}[(A)]
  \item $\TypeCheck{\SubstFun{\Gamma}{\Gamma_{\dvs_1}}}{\BlockLab{\dvs_2}{\val}{\X'}}{\T'}{\Remove{(\sharingRel'_2+\sharingRel')}{D_2}}$ where  $\X'=\Closure{\resV}{(\sharingRel'_2+\sharingRel')}\cap D_2$  
\end{enumerate}
From (1), (c) and \refToLemma{weakening}.2 we get
\begin{enumerate}[(A)]\addtocounter{enumi}{1}
  \item $\TypeCheckDecs{\SubstFun{\Gamma}{\Gamma_{\dvs_1}}}{\dvs_1}{\sharingRel'_1}$
  \end{enumerate}
From (b), (2) and \refToLemma{weakening}.1 we get
\begin{enumerate}[(A)]\addtocounter{enumi}{2}
  \item $\TypeCheck{\SubstFun{\Gamma}{\Gamma_{\dvs_1}}}{\val_i}{\T_i}{\sharingRel_i}$ for all $i$,  $1\leq i\leq n+m$  
\end{enumerate}
From (A), (B), (C) and rules \rn{T-New} and \rn{T-block} we get\\
\centerline{
$
\TypeCheck{\Gamma}{\e_2}{\C}{\Remove{((\sharingRel'_1+\sharingRel_{\val})+\Remove{(\sharingRel'_2+\sharingRel')}{D_2})}{D_1}}
$
}
and $\Y=\Closure{\resV}{((\sharingRel'_1+\sharingRel_{\val})+\Remove{(\sharingRel'_2+\sharingRel')}{D_2})}\cap D_1$. 
By $\alpha$-congruence we may assume that $\FV{\vs\,\vs'}\cap(D_1\cup D_2)=\emptyset$ so  
\begin{enumerate}[(A)]\addtocounter{enumi}{3}
\item $\sharingRel_{\val}+\Remove{(\sharingRel'_1+\sharingRel'_2+\sharingRel')}{(D_1\cup D_2)}=
\Remove{(\sharingRel_{\val}+\sharingRel'_1+\sharingRel'_2+\sharingRel')}{(D_1\cup D_2)}$.
\end{enumerate}
By (1) and $\FV{\vs\,\vs'}\cap  D_2=\emptyset$ and \refToProp{invTyping1} we have that 
$\Remove{(\sharingRel_{\val}+\sharingRel'_1)}{D_2}=\sharingRel_{\val}+\sharingRel'_1$, so from
\refToProp{lessSrRel}.\ref{p4} we have that
\begin{enumerate}[(A)]\addtocounter{enumi}{4}
\item $\Remove{(\sharingRel_{\val}+\sharingRel'_1+\sharingRel'_2+\sharingRel')}{D_2}=(\sharingRel_{\val}+\sharingRel'_1)+(\Remove{(\sharingRel'_2+\sharingRel')}{D_2})$
\end{enumerate}
Therefore we have that\\
\centerline{
$
\begin{array}{lcll}
\sharingRel&=& \Remove{(\sharingRel_{\val}+\sharingRel'_1+\sharingRel'_2+\sharingRel')}{(D_1\cup D_2)}& \text{by (D)}\\
&=& \Remove{(\Remove{(\sharingRel_{\val}+\sharingRel'_1+\sharingRel'_2+\sharingRel')}{D_2})}{D_1}& \text{by definition of $\setminus$}\\
&=& \Remove{((\sharingRel_{\val}+\sharingRel'_1)+\Remove{(\sharingRel'_2+\sharingRel')}{D_2})}{D_1}& \text{by (E)}\\
\end{array}
$
}
which proves that the typing of $\e_1$ and $\e_2$ produce the same sharing relation. 

\medskip\noindent
We show that \underline{$\TypeCheck{\Gamma}{\e_2}{\C}{\sharingRel}$ implies $\TypeCheck{\Gamma}{\e_1}{\C}{\sharingRel}$} for some $\X$.\\
Let $\TypeCheck{\Gamma}{\e_2}{\C}{\sharingRel}$, define $\Gamma_1=\SubstFun{\Gamma,\Gamma_{\dvs_1}}{\Gamma_{\dvs_2}}$. From rule \rn{T-new}  \refToLemma{invBlock} we have that
\begin{enumerate}[(a)]
\item $\sharingRel=\Remove{\sharingRel_3}{D_1}$ where $\sharingRel_3=\sharingRel_{\val}+\sharingRel'_1+\Remove{(\sharingRel'_2+\sharingRel'')}{D_2}$ and $\sharingRel_{\val}=\sum\limits_{i=1}^{n+m}\sharingRel_i$
  \item $\TypeCheck{\Gamma[\Gamma_{\dvs_1}]}{\val_i}{\T_i}{\sharingRel_i}$ for all $i$, $1\leq i\leq n+m$
  \item $\TypeCheckDecs{\Gamma[\Gamma_{\dvs_1}]}{\dvs_1}{\sharingRel'_1}$  
  \item $\TypeCheckDecs{\Gamma_1}{\dvs_2}{\sharingRel'_2}$ and  $\TypeCheck{\Gamma_1}{\val}{\T'}{\sharingRel'}$  and
 \item 
 $\X'=\Closure{\resV}{(\sharingRel'_2+\sharingRel')}\cap(D_2)$ and $\Y=\Closure{\resV}{\sharingRel_3}\cap(D_2)$
\end{enumerate}
where $\fields{\C}=\Field{\T_1}{\f_1}\dots\Field{\T_n}{\f_n}\Field{\T'}{\f'}\Field{\T_{n+1}}{\f_{n+1}}\dots\Field{\T_{n+m}}{\f_{n+m}}$.\\
By wellformedness of blocks $D_1\cap D_2=\emptyset$. Therefore, letting $\Gamma_2={\Gamma}[{\Gamma_{\dvs_1},\Gamma_{\dvs_2}}]$ we have that $\Gamma_1=\Gamma_2$.
From \refToLemma{weakening}.1 and (c) and (d) we get $\TypeCheckDecs{\Gamma_2}{\dvs_2}{\sharingRel'_2}$ and  $\TypeCheck{\Gamma_2}{\val}{\T'}{\sharingRel'}$ and $\TypeCheckDecs{\Gamma_2}{\dvs_1}{\sharingRel'_1}$ and by rule \rn{T-block}
\begin{enumerate}[(A)]
  \item $\TypeCheck{\SubstFun{\Gamma}{\Gamma_{\dvs_1}}}{\BlockLab{\dvs_2}{\val}{\X}}{\T'}{\Remove{(\sharingRel'_1+\sharingRel'_2+\sharingRel')}{(D_2\cup D_1)}}$ where  $\X'=\Closure{\resV}{(\sharingRel'_1+\sharingRel'_2+\sharingRel')}\cap (D_2\cup D_1)$  
\end{enumerate}
From (1), (b) and \refToLemma{weakening}.2 we get
\begin{enumerate}[(A)]\addtocounter{enumi}{1}
   \item $\TypeCheck{\Gamma}{\val_i}{\T_i}{\sharingRel_i}$ for all $i$, $1\leq i\leq n+m$
 \end{enumerate}
From (A), (B) and rule \rn{T-New} we get\\
\centerline{
$
\TypeCheck{\Gamma}{\e_1}{\C}{\sharingRel_{\val}+\Remove{(\sharingRel'_1+\sharingRel'_2+\sharingRel')}{(D_1\cup D_2)}}
$
}
and $\X=\Closure{\resV}{(\sharingRel'_1+\sharingRel'_2+\sharingRel')}\cap(D_1\cup D_2)$.\\
The proof that $\sharingRel=\sharingRel_{\val}+\Remove{(\sharingRel'_1+\sharingRel'_2+\sharingRel')}{(D_1\cup D_2)}$ is as for the previous implication.
\end{proof}

The following lemma asserts that subexpressions of typable expressions are
themselves typable, and may be replaced with expressions that have the same
type and the same or possibly less sharing effects. 

\noindent{\bf Lemma \ref{lemma:context}.} (Context)
{\it Let $\TypeCheck{\Gamma}{\Ctx{\e}}{\C}{\sharingRel}$, then
\begin{enumerate}
  \item $\TypeCheck{\Gamma[\Gamma_{\ctx}]}{{\e}}{\D}{\sharingRel_1}$ for some
    $\D$ and $\sharingRel_1$,
   \item if $\TypeCheck{\Gamma[\Gamma_{\ctx}]}{{\e'}}{\D}{\sharingRel_2}$ where  
    $\Finer{\sharingRel_2}{\sharingRel_1}$ (${\sharingRel_2}={\sharingRel_1}$), 
    then $\TypeCheck{\Gamma}{\CtxP{\e'}}{\C}{\sharingRel'}$ for some $\ctxP$ such that
    $\CtxP{\e'}\variant\Ctx{\e'}$ and
    $\Finer{\sharingRel'}{\sharingRel}$ (${\sharingRel'}={\sharingRel}$). 
\end{enumerate}
}
\begin{proof}
\begin{enumerate}
  \item Easy, by induction on evaluation contexts.
  \item Let $\TypeCheck{\Gamma}{\Ctx{\e}}{\C}{\sharingRel}$. By point 1. of this lemma 
  $\TypeCheck{\Gamma[\Gamma_{\ctx}]}{{\e}}{\D}{\sharingRel_1}$ for
  some $\D$ and $\sharingRel_1$. By induction on evaluation contexts.\\
  If $\ctx=\emptyctx$, then $\C=\D$ and $\sharingRel_1=\sharingRel$.
  The result is immediate.\\
  If $\ctx=\BlockLab{\dvs\ \Dec{\T}{\x}{\ctx'}\ \decs}{\e_b}{\X}$,
  then $\TypeCheck{\Gamma}{\BlockLab{\dvs\ \Dec{\T}{\x}{\CtxP{\e}}\ \decs}{\e_b}{\X}}{\C}{\sharingRel}$. 
  Let $\Gamma'=\Gamma_{\dvs},\TypeDec{\x}{\T},\Gamma_{\decs}$, from \refToLemma{invBlock} we have
  \begin{enumerate}[(1)]
  \item $\sharingRel=\Remove{(\sharingRel_1+\sharingRel_2+\sharingRel_3)}{\dom{\Gamma'}}$ where 
   \item $\TypeCheckDecs{\Gamma[\Gamma']}{\Dec{\T}{\x}{\CtxP{\e}}}{\sharingRel_1}$ where   $\sharingRel_1=\SubstEqRel{\sharingRel_x}{\x}{\resV}$ and $\T=\Type{\mu}{\C_x}$ and $\TypeCheck{\Gamma[\Gamma']}{{\CtxP{\e}}}{\C_x}{\sharingRel_x}$,
  \item $\TypeCheckDecs{\Gamma[\Gamma']}{\dvs\ \decs}{\sharingRel_2}$,  
  \item $\TypeCheck{\Gamma[\Gamma']}{\e_b}{}{\sharingRel_3}$ and
  \item  $\X=\Closure{\resV}{\sharingRel_1+\sharingRel_2+\sharingRel_3}\cap(\dom{\dvs}\cup\dom{\decs}\cup\{\x\})$
   \end{enumerate}
   From (2)  and point 1. of this lemma $\TypeCheck{\Gamma[\Gamma'][\Gamma_{\ctxP}]}{{\e}}{\D}{\sharingRel_4}$ for some $\D$ and 
   $\sharingRel_4$. 
Let
  $\TypeCheck{\Gamma[\Gamma'][\Gamma_{\ctxP}]}{{\e'}}{\D}{\sharingRel'_4}$
  where $\Finer{\sharingRel'_4}{\sharingRel_4}$. From (2), by induction hypothesis on
  $\ctxP$, we have that
  $\TypeCheck{\Gamma[\Gamma']}{{\ctxS[\e']}}{\D}{\sharingRel'_x}$
  where $\ctxP[\e']\variant\ctxS[\e']$ and $\Finer{\sharingRel'_x}{\sharingRel_x}$.
  Moreover,
  from $\TypeCheck{\Gamma}{\BlockLab{\dvs\ \Dec{\T}{\x}{\CtxP{\e}}\
  \decs}{\e_b}{\X}}{\C}{\sharingRel}$, if $\mu=\capsule$, then
  $\IsCapsule{\sharingRel_x}$, and  so also  $\IsCapsule{\sharingRel'_x}$.
  Therefore
  \begin{enumerate}[(a)]
  \item $\TypeCheckDecs{\Gamma[\Gamma']}{\Dec{\T}{\x}{\CtxS{\e}}}{\sharingRel'_1}$ where $\ctxP[\e']\variant\ctxS[\e']$ and $\sharingRel'_1=\SubstEqRel{\sharingRel'_x}{\x}{\resV}$.
 \end{enumerate}
 From \refToProp{lessSrRel}.\ref{p2} and \ref{p3} we have that $\Finer{\sharingRel'_1}{\sharingRel_1}$. 
 Applying rule \rn{T-block} to  (2), (3), (4) and (a) we have\\
 \centerline{
  $\TypeCheck{\Gamma}{\BlockLab{\dvs\ \Dec{\T}{\x}{\CtxS{\e'}}\
  \decs}{\e_b}{\X'}}{\C}{\sharingRel''}$}\\
where $\sharingRel''=\Remove{(\sharingRel'_1+\sharingRel_2+\sharingRel_3)}{\dom{\Gamma'}}$ and $\X'=\Closure{\resV}{\sharingRel'_1+\sharingRel_2+\sharingRel_3}\cap(\dom{\dvs}\cup\dom{\decs}\cup\{\x\})$. From \refToProp{lessSrRel}.\ref{p2} and \ref{p3} we derive that 
  $\Finer{\sharingRel''}{\sharingRel}$. The case with equality is similar.\\ 
  If
  $\ctx=\BlockLab{\dvs}{\ctx}{\X}$, then the proof is similar to the previous one.
\end{enumerate}
\end{proof}

\noindent{\bf Lemma \ref{lemma:monotoneSharing}.} 
{\it Let $\TypeCheck{\Gamma}{\Ctx{e}}{\C}{\sharingRel}$ and $\TypeCheck{\Gamma}{\e}{\D}{\sharingRel'}$. 
If $\Pair{\x}{\y}\in\sharingRel'$ with $\x,\y\not\in\HB{\ctx}$ and $\x,\y\neq\resV$, 
then $\Pair{\x}{\y}\in\sharingRel$.}
\begin{proof}
By induction on $\ctx$. \\
For \underline{$\ctx=\emptyctx$} is obvious. \\
Let \underline{$\ctx=\BlockLab{\dvs\ \Dec{\Type{\mu}{\C_1}}{\z}{\ctxP}\ \decs}{\e_b}{\X}$}.
From $\TypeCheck{\Gamma}{\Ctx{e}}{\C}{\sharingRel}$ and \refToLemma{invBlock}, let
$\Gamma'=\Gamma_{\dvs},\TypeDec{\x}{\T},\Gamma_{\decs}$
\begin{enumerate}[(a)]
  \item$\TypeCheck{\Gamma[\Gamma']}{{\ctxP[\e]}}{\C_1}{\sharingRel_1}$,   
  \item  $\TypeCheckDecs{\Gamma[\Gamma']}{\dvs\ \decs}{\sharingRel_2}$, and 
  \item  $\TypeCheck{\Gamma[\Gamma']}{\e_b}{}{\sharingRel_3}$. 
  \item  $\sharingRel=\Remove{(\sharingRel'_1+\sharingRel_2+\sharingRel_3)}{\dom{\Gamma'}}$ where $\sharingRel'_1=\Remove{(\sharingRel_1+\{\z,\resV\})}{\resV}$. 
\end{enumerate}
Assume that, $\Pair{\x}{\y}\in\sharingRel'$ with $\x,\y\not\in\HB{\ctx}$ and $\x,\y\neq\resV$. From
$\HB{\ctxP}\subseteq\HB{\ctx}$ and $\TypeCheck{\Gamma}{\e}{\D}{\sharingRel'}$, by induction hypothesis
on $\ctxP$ we derive that  $\Pair{\x}{\y}\in\sharingRel_1$. Since $(\{\z\}\cup\dom{\Gamma'})\subseteq\HB{\ctx}$ we have that 
$\x,\y\not\in(\{\z\}\cup\dom{\Gamma'})$. Therefore $\Pair{\x}{\y}\in\sharingRel$.\\
Similar (and simpler) for \underline{$\ctx=\BlockLab{\dvs}{\ctxP}{\X}$}.
\end{proof}

\noindent{\bf Lemma \ref{lemma:fieldAcc}.} 
{\it
If $\TypeCheck{\Gamma}{\Ctx{e_1}}{\C}{\sharingRel_1}$, $\TypeCheck{\Gamma}{\Ctx{e_2}}{\C}{\sharingRel_2}$,
$\TypeCheck{\Gamma}{\e_1}{\D}{\{\xs,\resV\}}$ and $\TypeCheck{\Gamma}{\e_2}{\D}{\{\ys,\resV\}}$ with 
$\{\xs,\ys\}\cap\HB{\ctx}=\emptyset$. Then $\sharingRel_1+\{\xs,\ys\}=\sharingRel_2+\{\xs,\ys\}$.
}
\begin{proof}
Let \underline{$\ctx=\emptyctx$}. Then $\{\xs,\resV\}+\{\xs,\ys\}=\{\xs, \ys, \resV\}$ and $\{\ys,\resV\}+\{\xs,\ys\}=\{\xs, \ys, \resV\}$. \\
Let \underline{$\ctx=\BlockLab{\dvs\ \Dec{\Type{\mu}{\D}}{\z}{\ctxP}\ \decs}{\e}{\X}$},
$\TypeCheck{\Gamma}{\BlockLab{\dvs\ \Dec{\Type{\mu}{\D}}{\z}{\ctxP[\e_1]}\ \decs}{\e}{\X}}{\C}{\sharingRel_1}$, 
$\TypeCheck{\Gamma}{\BlockLab{\dvs\ \Dec{\Type{\mu}{\D}}{\z}{\ctxP[\e_2]}\ \decs}{\e}{\X}}{\C}{\sharingRel_2}$, 
$\TypeCheck{\Gamma}{\e_1}{\C}{\{\xs,\resV\}}$, and 
$\TypeCheck{\Gamma}{\e_2}{\C}{\{\ys ,\resV\}}$.
Let $\Gamma'=\Gamma_{\dvs},\TypeDec{\z}{\T},\Gamma_{\decs}$, from \refToLemma{invBlock}
\begin{enumerate}[(a)]
  \item $\TypeCheck{\Gamma[\Gamma']}{{\ctxP[\e_1]}}{\C_1}{\sharingRel_3}$
    and $\TypeCheck{\Gamma[\Gamma']}{{\ctxP[\e_2]}}{\C_1}{\sharingRel_4}$  
  \item $\TypeCheckDecs{\Gamma[\Gamma']}{\dvs\ \decs}{\sharingRel_d}$ 
  \item  $\TypeCheck{\Gamma[\Gamma']}{\e}{}{\sharingRel_e}$ 
  \item $\sharingRel_1=\Remove{(\sharingRel_z+\sharingRel_d+\sharingRel_e)}{\dom{\Gamma'}}$
    and $\sharingRel_2=\Remove{(\sharingRel'_z+\sharingRel_d+\sharingRel_e)}{\dom{\Gamma'}}$
    where $\sharingRel_z=\SubstEqRel{\sharingRel_3}{\z}{\resV}$ and
    $\sharingRel'_z=\SubstEqRel{\sharingRel_4}{\z}{\resV}$. 
\end{enumerate}
Let $\Y=\{\xs,\ys\}$. By (a) and induction hypotheses we have that
$\Sum{\sharingRel_3}{\Y}=\Sum{\sharingRel_4}{\Y}$. Therefore $\SubstEqRel{(\Sum{\sharingRel_3}{\Y})}{\z}{\resV}=\SubstEqRel{(\Sum{\sharingRel_4}{\Y})}{\z}{\resV}$. From $\Y\cap\HB{\ctx}=\emptyset$, we have that
for all $\x\in\xs$, $\x\not=\z$ and for all $\y\in\ys$, $\y\not=\z$ and so
$\Sum{(\SubstEqRel{\sharingRel_3}{\z}{\resV})}{\Y}=\Sum{(\SubstEqRel{\sharingRel_4}{\z}{\resV})}{\Y}$. Therefore
$\Sum{\sharingRel_z}{\Y}=\Sum{\sharingRel'_z}{\Y}$ and so also
$\sharingRel_z+\sharingRel_d+\sharingRel_e+\Y=\sharingRel'_z+\sharingRel_d+\sharingRel_e+\Y$,
which implies
$\Remove{(\sharingRel_z+\sharingRel_d+\sharingRel_e+\Y)}{\dom{\Gamma'}}=\Remove{(\sharingRel'_z+\sharingRel_d+\sharingRel_e+\Y)}{\dom{\Gamma'}}$.
Since $\Y\cap\dom{\ctx}=\emptyset$ we have that
$\Y\cap\dom{\Gamma'}=\emptyset$, and $\Remove{\Y}{\dom{\Gamma'}}=\Y$.
Therefore from  \refToProp{lessSrRel}.\ref{p4} we have that
$\Remove{((\sharingRel_z+\sharingRel_d+\sharingRel_e)+\Y)}{\dom{\Gamma'}}=\Remove{(\sharingRel_z+\sharingRel_d+\sharingRel_e)}{\dom{\Gamma'}}+\Y)=\Remove{(\sharingRel'_z+\sharingRel_d+\sharingRel_e+\Y)}{\dom{\Gamma'}}=\Remove{(\sharingRel'_z+\sharingRel_d+\sharingRel_e)}{\dom{\Gamma'}}+\Y)$ which implies $\sharingRel_1+\Y=\sharingRel_2+\Y$.\\
Similar (and simpler) for \underline{$\ctx=\BlockLab{\dvs}{\ctxP}{\X}$}.
\end{proof}

\noindent{\bf Theorem \ref{theo:subred}.} (Subject reduction)
{\it If $\TypeCheck{\Gamma}{\e_1}{\C}{\sharingRel}$ and $\reduce{\e_1}{\e_2}$, then
\begin{enumerate}
  \item  \PG{$\TypeCheck{\Gamma}{\e'_2}{\C}{\sharingRel'}$ for $\e'_2\variant\e_2$ and 
    $\Finer{\sharingRel'}{\sharingRel}$, and}
  {\item for all $\x$ such that $\e_1=\Decctx{\x}{\e}$, \PG{$\e'_2=\DecctxP{\x}{\e'}$},
    and $\TypeCheck{\TypeEnv{\decctx{\x}}}{\e}{\D}{\sharingRel_x}$ we have that:
    $\TypeCheck{\TypeEnv{\decctxP{\x}}}{\e'}{\D}{\sharingRel'_x}$ and 
    $\Finer{(\sharingRel'_x+\sharingRel_{\dvs'})}{(\sharingRel_x+\sharingRel_{\dvs})}$
    where $\dvs=\extractDec{\decctx{}}{\FV{\e}}$ and 
    $\dvs'=\extractDec{\decctxP{}}{\FV{\e'}}$.}
\end{enumerate}
}
\begin{proof}
\underline{Rule \rn{invk}}. 
\begin{enumerate}
  \item In this case $\redex=\MethCall{\val_0}{\m}{\val_1,..,\val_n}$ and
    \begin{center}
      $\e'=\Block{\Dec{\Type{\mu}{\C_0}}{\this}{\val_0}\, \Dec{\Type{\mu_1}{\C_1}}{\x_1}{\val_1}..\Dec{\Type{\mu_n}{\C_n}}{\x_n}{\val_n}}{\e_b}$ 
    \end{center} 
    where 
    $\method{\C_0}{\m}{=}\Method{\ReturnTypeNew{\D}{\sharingRel_b}}{\mu_0}{\Param{\Type{\mu_1}{\C_1}}{\x_1}\ldots\Param{\Type{\mu_n}{\C_n}}{\x_n}}{\e_b}$.
    From \refToLemma{context}.1 we have that 
    $\TypeCheck{\Gamma[\Gamma_{\ctx}]}{\MethCall{\val_0}{\m}{\val_1,..,\val_n}}{\D}{\sharingRel''}$ 
    for some $\sharingRel'_1$. From typing rule \rn{T-invk} 
    \begin{enumerate} [(1)]
      \item $\TypeCheck{\Gamma[\Gamma_{\ctx}]}{\val_i}{\C_i}{\sharingRel_i}$ 
        ( $0\leq i\leq n$)
      \item $\forall\ {0 \leq i \leq n}\ \ \mu_i=\capsule\Longrightarrow{\IsCapsule{\sharingRel_i}}$
      \item $\sharingRel'_0=\SubstEqRel{\sharingRel_0}{\this}{\resV}$
      \item $\sharingRel'_i=\SubstEqRel{\sharingRel_i}{\x_i}{\resV}$ ($1\leq i\leq n$)
      \item $\sharingRel''=\Remove{(\Sum{\sum\limits_{i=1}^{n}\sharingRel'_i}{\sharingRel_b})}{\{\this,\x_1,\ldots,\x_n\}}$
    \end{enumerate}
    From the fact that the class table is well-typed we have that
    $\TypeCheck{\Gamma'}{\e_b}{\D}{\sharingRel_b}$ where
    $\Gamma'=\TypeDec{\this}{\Type{\mu}{\C_0},\TypeDec{\x_1}{\Type{\mu_1}{\C_1}},\ldots,\TypeDec{\x_n}{\Type{\mu_n}{\C_n}}}$.
    Moreover, since we may assume that
    $\{\this,\x_1,\ldots,\x_n\}\cap\dom{\Gamma[\Gamma_{\ctx}]}=\emptyset$,
    from \refToLemma{weakening} we have that 
    \begin{enumerate}[(a)] 
      \item $\TypeCheck{\Gamma[\Gamma_{\ctx}][\Gamma']}{\e_b}{\D}{\sharingRel_b}$, and
      \item $\TypeCheck{\Gamma[\Gamma_{\ctx}][\Gamma']}{\val_i}{\C_i}{\sharingRel_i}$ ( $0\leq i\leq n$).
    \end{enumerate}
    Therefore by typing rule \rn{T-block}, (1)$\div$(5), (a) and (b) we have 
    that $\TypeCheck{\Gamma[\Gamma_{\ctx}]}{\e'}{\D}{\sharingRel''}$. From 
    \refToLemma{context}.2 we derive $\TypeCheck{\Gamma}{\CtxP{\e''}}{\D}{\sharingRel}$ where $\CtxP{\e''}\variant\Ctx{\e'}$.
  \item The result is proved as in the case of \rn{field-assign} just replacing 
    $\Finer{\sharingRel'_x}{\sharingRel_x}$ with 
    ${\sharingRel'_x}={\sharingRel_x}$ since the sharing relation of the redex 
    is equal to the one of the block to which it reduces.
\end{enumerate}

\underline{Rule \rn{alias-elim}}. 
\begin{enumerate}
  \item In this case 
    \begin{enumerate} [(1)]
      \item $\redex=\BlockLab{\decs' }{\e_b}{\X}$ where 
        $\decs'=\dvs\ \Dec{\C_1}{\x}{\y}\ \decs$ 
      \item $\e'=\BlockLab{\dvs\ \Subst{\decs}{\y}{\x}}{\Subst{\e_b}{\y}{\x}}{X\setminus\{\x\}}$
    \end{enumerate}
    From \refToLemma{context}.1 we have that 
    $\TypeCheck{\Gamma[\Gamma_{\ctx}]}{\BlockLab{\decs'}{\e_b}{\X}}{\D}{\sharingRel_1}$ 
    for some $\sharingRel_1$. Therefore from \refToLemma{invBlock} we have that
    \begin{enumerate} [(a)]
      \item $\TypeCheckDecs{\Gamma[\Gamma_{\ctx}][\Gamma_{\decs'}]}{\dvs}{\sharingRel_2}$ for some $\sharingRel_2$,
      \item $\TypeCheckDecs{\Gamma[\Gamma_{\ctx}][\Gamma_{\decs'}]}{\Dec{\C_1}{\x}{\y}}{\{\x,\y\}}$,
      \item $\TypeCheckDecs{\Gamma[\Gamma_{\ctx}][\Gamma_{\decs'}]}{\decs}{\sharingRel_3}$  for some $\sharingRel_3$,
      \item $\TypeCheck{\Gamma[\Gamma_{\ctx}][\Gamma_{\decs'}]}{\e_b}{\D}{\sharingRel_4}$  for some $\sharingRel_4$,
      \item $\sharingRel_1=\Remove{\sharingRel'_1}{\dom{\decs'}}$ where $\sharingRel'_1=\sum\limits_{i=2}^{4}\sharingRel_i+\{\x,\y\}$ 
        and $\X=\Closure{\resV}{\sharingRel'_1}\cap\dom{\decs'}$.
    \end{enumerate}
    Since $\x$ cannot be free in $\dvs$, from \refToLemma{weakening}.2 and (a) 
    we derive
    \begin{enumerate} [(A)]
      \item $\TypeCheckDecs{\Gamma[\Gamma_{\ctx}][\Remove{\Gamma_{\decs'}}{\x}]}{\dvs}{\Remove{\sharingRel_2}{\x}}$.
    \end{enumerate}
    From  \refToLemma{substitution}.1 and (c) and (d) we have that 
    \begin{enumerate} [(A)]\addtocounter{enumii}{2}
      \item $\TypeCheckDecs{\Remove{\Gamma[\Gamma_{\ctx}][\Gamma_{\decs'}]}{\x}}{\Subst{\decs}{\y}{\x}}{\Remove{\sharingRel_3}{\x}}$ and 
      \item $\TypeCheckDecs{\Remove{\Gamma[\Gamma_{\ctx}][\Gamma_{\decs'}]}{\x}}{\Subst{\e_b}{\y}{\x}}{\Remove{\sharingRel_4}{\x}}$.
    \end{enumerate}
    Moreover, 
    \begin{enumerate}[(A)]\addtocounter{enumii}{4}
          \item let $\sharingRel''=\sum\limits_{i=2}^{4}(\Remove{\sharingRel_i}{\x})$, 
    \end{enumerate}
    from (A), (C)$\div$(E) and rule \rn{T-block} we have that 
    \begin{center}
      $\TypeCheck{\Gamma[\Gamma_{\ctx}][\Gamma_{\decs'}]}{\BlockLab{\dvs\ \Subst{\decs}{\y}{\x}}{\Subst{\e_b}{\y}{\x}}{\Y}}{\D}{\Remove{\sharingRel''}{\dom{\dvs\,\decs}}}$
    \end{center}
    where $\Y=\Closure{\resV}{{\sharingRel''}}\cap\dom{\dvs\,\decs}$. If
    $\x\not\in\Closure{\resV}{\sharingRel'_1}$, then
    $\Closure{\resV}{\sharingRel'_1}=\Closure{\resV}{{\sharingRel''}}$, and
    since $\dom{\dvs\,\decs}\cup\{x\}=\dom{\decs'}$ we have that $\X=\Y$. If
    $\x\in\Closure{\resV}{\sharingRel'_1}$, then
    $\Closure{\x}{\sharingRel'_1}=\Closure{\resV}{\sharingRel'_1}$ and from (e)
    we have that
    $\Closure{\x}{\sharingRel'_1}=\Closure{\resV}{\sharingRel'_1}=\Closure{\y}{\sharingRel'_1}$.
    Therefore
    $\Closure{\resV}{\sharingRel''}=\Closure{\resV}{\sharingRel'_1}\setminus\{\x\}$
    and $\Y=X\setminus\{\x\}$. From
    $\Finer{\Remove{\sharingRel_i}{\x}}{\sharingRel_i}$ ($2\leq i\leq 4$) and
    \refToProp{lessSrRel}.\ref{p2} we have that
    $\Finer{\sharingRel''}{\sharingRel'_1}$.  Therefore from
    \refToProp{lessSrRel}.\ref{p3} we derive $\Finer{\sharingRel_2}{\sharingRel_1}$.
  \item The result is proved as in the case of \rn{field-assign} since from  
    $\Finer{\sharingRel_2}{\sharingRel_1}$ by \refToLemma{context} we derive
    ${\sharingRel'_x}={\sharingRel_x}$.
\end{enumerate}

\underline{Rule \rn{affine-elim}}.
\begin{enumerate}
  \item In this case 
    \begin{enumerate} [(1)]
      \item $\redex=\BlockLab{\decs' }{\e_b}{\X}$ where 
        $\decs'=\dvs\ \Dec{\Type{\capsule}{\C_1}}{\x}{\val}\ \decs$ 
      \item $\e'=\BlockLab{\dvs\ \Subst{\decs}{\val}{\x}}{\Subst{\e_b}{\val}{\x}}{X\setminus\{\x\}}$
    \end{enumerate}
  
    From \refToLemma{context}.1 we have that
    $\TypeCheck{\Gamma[\Gamma_{\ctx}]}{\BlockLab{\decs'}{\e_b}{\X}}{\D}{\sharingRel_1}$
    for some $\sharingRel_1$. Therefore from \refToLemma{invBlock} we 
    have that
    \begin{enumerate} [(a)]
      \item  $\TypeCheckDecs{\Gamma[\Gamma_{\ctx}][\Gamma_{\decs'}]}{\dvs}{\sharingRel_2}$ 
        for some $\sharingRel_2$,
      \item $\TypeCheck{\Gamma[\Gamma_{\ctx}][\Gamma_{\decs'}]}{\val}{\C_1}{\sharingRel_v}$ 
        where $\IsCapsule{\sharingRel_v}$, therefore also 
        $\TypeCheckDecs{\Gamma[\Gamma_{\ctx}][\Gamma_{\decs'}]}{\Dec{\Type{\capsule}{\C_1}}{\x}{\val}}{\sharingRel_v}$
      \item $\TypeCheckDecs{\Gamma[\Gamma_{\ctx}][\Gamma_{\decs'}]}{\decs}{\sharingRel_3}$
        for some $\sharingRel_3$,
      \item $\TypeCheck{\Gamma[\Gamma_{\ctx}][\Gamma_{\decs'}]}{\e_b}{\D}{\sharingRel_4}$ 
        for some $\sharingRel_4$,
      \item $\sharingRel_1=\Remove{\sharingRel'_1}{\dom{\decs'}}$ where 
        $\sharingRel'_1=\sum\limits_{i=2}^{4}\sharingRel_i+\sharingRel_v$ and 
        $\X=\Closure{\resV}{\sharingRel'_1}\cap\dom{\decs'}$.
    \end{enumerate}
    Let
    \begin{enumerate} [(A)]\addtocounter{enumii}{1}
         \item $\TypeCheck{\Gamma[\Gamma_{\ctx}][\Gamma_{\decs'}]}{\remGarbage(\val)}{\C_1}{\sharingRel'_v}$ 
    \end{enumerate}
    we also have $\IsCapsule{\sharingRel'_v}$.\\ 
    From (B), the fact that $\Gamma_{\decs'}(\x)=\Type{\capsule}{\C_1}$, and
    \refToLemma{sharingCapsule} we have that $\sharingRel'_v=\epsilon$. Since 
    we do not have forward references to unevaluated variables,
    $\x$ cannot be free in $\dvs$ and from \refToLemma{weakening}.2 and (a) we derive
    \begin{enumerate} [(A)]
      \item  $\TypeCheckDecs{\Gamma[\Gamma_{\ctx}][\Remove{\Gamma_{\decs'}}{\x}]}{\dvs}{\Remove{\sharingRel_2}{\x}}$.
    \end{enumerate}
    From (B) with $\sharingRel'_v=\epsilon$, \refToLemma{substitution}.2 and (c)
    and (d) we have that 
   \begin{enumerate} [(A)]\addtocounter{enumii}{2} 
      \item $\TypeCheckDecs{\Remove{\Gamma[\Gamma_{\ctx}][\Gamma_{\decs'}]}{\x}}{\Subst{\decs}{\remGarbage(\val)}{\x}}{\Remove{\sharingRel_3}{\x}}$ and 
      \item $\TypeCheckDecs{\Remove{\Gamma[\Gamma_{\ctx}][\Gamma_{\decs'}]}{\x}}{\Subst{\e_b}{\remGarbage(\val)}{\x}}{\Remove{\sharingRel_4}{\x}}$.
    \end{enumerate}
    Moreover, 
     \begin{enumerate} [(A)]\addtocounter{enumii}{4}  
      \item let 
        $\sharingRel''=\sum\limits_{i=2}^{4}(\Remove{\sharingRel_i}{\x})$, 
    \end{enumerate}
    from (A), (C)-(E) and rule \rn{T-block} we have that 
    \begin{center}
      $\TypeCheck{\Gamma[\Gamma_{\ctx}][\Gamma_{\decs'}]}{\BlockLab{\dvs\ \Subst{\decs}{\val}{\x}}{\Subst{\e_b}{\val}{\x}}{\Y}}{\D}{\Remove{\sharingRel''}{\dom{\dvs\,\decs}}}$
    \end{center}
    where $\Y=\Closure{\resV}{{\sharingRel''}}\cap\dom{\dvs\,\decs}$. Since
    $\Closure{\x}{\sharingRel'_1}=\{\x\}$, we have that 
    $\x\not\in\Closure{\resV}{\sharingRel'_1}$. Moreover 
    $\dom{\dvs\,\decs}\cup\{x\}=\dom{\decs'}$. Therefore we have that 
    $\X=\Remove{\X}{\x}=\Y$. From $\Finer{\Remove{\sharingRel_i}{\x}}{\sharingRel_i}$ 
    ($2\leq i\leq 4$) and \refToProp{lessSrRel}.\ref{p2} we get 
    $\Finer{\sharingRel''}{\sharingRel'_1}$. 
  \item The result is proved as in the case of \rn{field-assign} since from
    $\Finer{\sharingRel_2}{\sharingRel_1}$ by \refToLemma{context} we derive 
    ${\sharingRel'_x}={\sharingRel_x}$.
\end{enumerate}
\end{proof}

\noindent{\bf Lemma \ref{lemma:decomposition}.} (Decomposition)
{\it If  $\e$ is not a value, then there are 
$\ctx$ and $\redex$ such that $\congruence{\e}{\Ctx{\redex}}$.}
\begin{proof}
  By structural induction on expressions. \\ If
  \underline{$\FieldAssign{\val}{\f}{\val'}$}, then from \refToProp{value} we
  have that $\val=\x$ or $\val=\BlockLab{\dvs}{\x}{\X}$ and $\val'=\y$ or
  $\val'=\BlockLab{\dvs'}{\y}{\Y}$. If $\val=\BlockLab{\dvs}{\x}{\X}$ and
  $\val'=\BlockLab{\dvs'}{\y}{\Y}$, we may assume, by $\alpha$-renaming, that
  $\dom{dvs}\cap\dom{dvs'}=\emptyset$.  From rule \rn{val-ctx} (applied twice)
  $\congruence{\FieldAssign{\val}{\f}{\val'}}{\BlockLab{\dvs\
  \dvs'}{\FieldAssign{\x}{\f}{\y}}{\X\cup\Y}}$.  So
  $\congruence{\FieldAssign{\val}{\f}{\val'}}{\Ctx{\FieldAssign{\x}{\f}{\y}}}$
  where $\ctx=\BlockLab{\dvs\ \dvs'}{\emptyctx}{\X\cup\Y}$ The other cases are
  easier.\\ For field access the proof is similar. \\ Method call is a redex so
  the result holds with $\ctx=\emptyctx$.\\ If
  \underline{$\BlockLab{\decs}{\e}{\X}$} is not a value, then either 
  \begin{enumerate}[(1)]
    \item $\decs=\dvs\,\Dec{\T}{\x}{\e_1}\,\decs_1$ where $\e_1$ is not 
      $\ConstrCall{\C}{\xs}$ for some $\C$ and $\xs$ or
    \item $\decs=\dvs$ and $\e$ is not a value.
  \end{enumerate}
In \underline{case (1)}, either $\e_1$ is not a value or $\e_1$ is a value but not $\ConstrCall{\C}{\xs}$ for some $\C$
and $\xs$.\\
In case $\e_1$ is not a value, by induction hypothesis, there are $\ctx$, and $\redex$  such that
$\congruence{\e_1}{\Ctx{\redex}}$. Applying congruence rule \rn{reorder} of
\refToFigure{congruence} we have that
\begin{center}
$\congruence{\BlockLab{\dvs\,\dvs_1\Dec{\T}{\x}{\e_1}\,\decs_2}{\e}{\X}}{\BlockLab{\dvs\,\Dec{\T}{\x}{\e_1}\,\decs_1}{\e}{\X}}$ 
\end{center}
where $\dvs_1$ are all the evaluated declarations of $\decs_1$. Therefore
$\congruence{\BlockLab{\decs}{\e}{\X}}{\CtxP{\redex}}$ where
$\ctxP=\BlockLab{\dvs\,\dvs_1\Dec{\T}{\x}{\ctx}\,\decs_2}{\e}{\X}$.\\ 
In case $\e_1$ is a value but not $\ConstrCall{\C}{\xs}$ for some $\C$
and $\xs$, by Proposition~\ref{prop:value}, either $\congruence{\e_1}{\y}$ or
$\congruence{\e_1}{\BlockLab{\dvs'}{\y}{\Y}}$. If $\T=\Type{\capsule}{\D}$ for
some $\D$, then the block is a redex, else either $\congruence{\e_1}{\y}$ or
$\congruence{\e_1}{\BlockLab{\dvs'}{\y}{\Y}}$. \\ If $\congruence{\e_1}{\y}$,
then $\BlockLab{\decs}{\e}{\Y}$ is congruent to
$\BlockLab{\dvs\,\Dec{\T}{\x}{\y}\,\decs_1}{\e}{\X}$, which is a redex. \\ If
$\congruence{\e_1}{\BlockLab{\dvs'}{\y}{\X}}$, then
$\congruence{\BlockLab{\dvs\,\Dec{\T}{\x}{\e_1}\,\decs_1}{\e}{\X}}{\BlockLab{\dvs\,\Dec{\T}{\x}{\BlockLab{\dvs'}{\y}{\Y}}\,\decs_1}{\e}{\X}}$
Since $\T=\D$ for some $\D$, with $\alpha$-renaming of variables in $\dom{\dvs'}$, applying congruence rule \rn{Dec} we have
\begin{center}
  $\congruence{\BlockLab{\dvs\,\Dec{\T}{\x}{\e_1}\,\decs_1}{\e}{\X}}{\BlockLab{\dvs\,\dvs'\Dec{\T}{\x}{{\y}}\,\decs_1}{\e}{\X}}$
\end{center}
and the expression on the right is a redex.\\ 
In \underline{case (2)}, by
induction hypothesis, there are $\ctx$, and $\redex$  such that
$\congruence{\e}{\Ctx{\redex}}$. Therefore
$\congruence{\BlockLab{\decs}{\e}{\X}}{\CtxP{\redex}}$ where
$\ctxP=\BlockLab{\dvs}{\ctx}{\X}$.
 \end{proof}

\end{document}